\documentclass[11pt]{article}
\usepackage[margin=1in]{geometry}
\usepackage{float}
\usepackage{multirow}
\usepackage{bm}
\usepackage{amsmath,amsthm,amsfonts,amssymb}
\usepackage{enumerate}
\usepackage{graphicx}
\usepackage{subcaption}
\usepackage[authoryear,round]{natbib}
\usepackage{booktabs,tabularx}
\usepackage{enumitem}

\usepackage{xcolor}
\usepackage[normalem]{ulem}

\definecolor{revisionorange}{RGB}{230,120,0}

\PassOptionsToPackage{hyphens}{url}
\usepackage{url}
\usepackage{hyperref}
\hypersetup{
  colorlinks = true,
  linkcolor  = [rgb]{0.10,0.20,0.55},
  citecolor  = [rgb]{0.10,0.20,0.55},
  urlcolor   = [rgb]{0.10,0.20,0.55}
}

\usepackage{array}
\usepackage{mathrsfs}

\allowdisplaybreaks

\newtheorem{theorem}{Theorem}[section]
\newtheorem{corollary}[theorem]{Corollary}
\newtheorem{lemma}[theorem]{Lemma}
\newtheorem{proposition}[theorem]{Proposition}
\theoremstyle{definition}
\newtheorem{definition}[theorem]{Definition}
\newtheorem{assumption}[theorem]{Assumption}
\theoremstyle{remark}
\newtheorem{remark}[theorem]{Remark}

\newcommand{\R}{\mathbb{R}}
\newcommand{\E}{\mathbb{E}}
\newcommand{\Q}{\mathbb{Q}}

\newcommand{\1}{\mathbf{1}}

\title{Microstructural Foundation for the Rough Hawkes--Heston Model}
\author{%
Yingli Wang\thanks{School of Mathematical Sciences, Fudan University, Shanghai, People's Republic of China; \texttt{yingliwang@fudan.edu.cn}.}%
\and
Yinhao Wu\thanks{School of Mathematics, Shanghai University of Finance and Economics, Shanghai, People's Republic of China; \texttt{yinhaowu@stu.sufe.edu.cn}.}%
\and
Lingjiong Zhu\thanks{Corresponding Author. Department of Mathematics, Florida State University, Tallahassee, Florida, United States of America; \texttt{zhu@math.fsu.edu}.}
}
\date{\today}

\begin{document}
\maketitle

\begin{abstract}
Hawkes-based microstructural foundations for rough volatility, leverage, and rough Heston-type limits were developed by El Euch et al. (2018, \textit{Finance Stoch.}, \textbf{22}(2), 241--280) and connected to the affine rough Heston framework of El Euch and Rosenbaum (2019, \textit{Math. Finance}, \textbf{29}(1), 3--38). The rough Hawkes--Heston model with common price--volatility jumps of Bondi et al. (2024, \textit{Math. Finance}, \textbf{34}(4), 1197--1241) extends this framework by adding state-dependent common jumps to rough affine volatility. We provide a microstructural foundation for its variance and common-jump mechanism by constructing a Poisson-embedded marked Hawkes order-flow model. Ordinary arrivals generate rough continuous volatility and leverage through a nearly unstable heavy-tailed Hawkes mechanism, while rare marked arrivals represent common shock events that produce simultaneous price jumps and volatility excitation. Under the nearly unstable scaling and the reduced-form admissibility conditions,
the complete rescaled price/variance/jump system converges along the full
sequence to the unique complete canonical rough Hawkes--Heston weak solution.
The Hawkes renewal structure
yields a Mittag--Leffler Volterra representation, which is then rewritten in
Riemann--Liouville fractional form. The limiting coefficients are expressed
explicitly in terms of the microscopic parameters. The effective Volterra drift
is determined by the balance between the near-critical Hawkes gap and the
predictable mean of common-jump feedback. Within the subcritical scaling
regime, this balance restricts the drift to be negative. The continuous price
martingale and its leverage correlation are generated by regular order flow
through a martingale-balanced asymmetry of sign-specific arrival rates and
tick sizes. The construction combines this endogenous mechanism with
reduced-form exponential price normalization and direct marked-event price
loading. Hence, the construction provides a
microstructural foundation for the variance and common-jump mechanism of the
rough Hawkes--Heston model.
Numerical experiments illustrate the convergence  of our microstructural foundation to the rough Hawkes-Heston model.
\end{abstract}

\medskip


\medskip


\section{Introduction}
\label{sec:intro}

The Heston model \citep{heston1993} is a classical affine stochastic volatility
benchmark with tractable Markovian dynamics. Rough-volatility and common-jump
extensions enrich its path regularity and joint return--variance shocks.
Empirically, volatility is very rough; see
\cite{gatheral2018}. A microstructural explanation is given by nearly unstable
heavy-tailed Hawkes systems. In these models, the near-critical Hawkes renewal
resolvent first converges to a Mittag--Leffler resolvent kernel. The
corresponding mild equation can then be rewritten in Riemann--Liouville
fractional Volterra form; see
\cite{jaisson2016rough,eleuchfukasawarosenbaum2018}.

In the notation used below, with $P_t=\log(S_t/S_0)$ where $S_t$ is the asset price at time $t$, the rough Heston model can
be written in the schematic form
\[
\left\{
\begin{aligned}
P_t
&=
\int_0^t\sqrt{V_s}dW_s
-\frac12\int_0^tV_sds,\\
V_t
&=
G_0(t)
+
b\int_0^t k_\alpha(t-s)V_sds
+
\eta\int_0^t k_\alpha(t-s)\sqrt{V_s}dB_s,
\end{aligned}
\right.
\]
with
\begin{align}\label{k:alpha}
  k_\alpha(t):=\frac{t^{\alpha-1}}{\Gamma(\alpha)},
  \qquad \alpha\in(1/2,1),
\end{align}
where $W_{t},B_{t}$ are standard Brownian motions with $d\langle W,B\rangle_t=\rho dt$.
Here, $G_0$ is the deterministic forward variance curve, $\eta>0$ is the volatility-of-volatility coefficient, and $b$ is the Volterra mean-reversion loading.
With this sign convention,
stability corresponds to a negative value of this coefficient, in analogy with
the negative linear drift in the classical square-root variance equation.

Roughness is essential for describing volatility dynamics, and joint SPX/VIX
smile calibration motivates extensions of the baseline rough Heston
mechanism. One approach introduces a price-feedback, or
Zumbach, mechanism. Quadratic Hawkes limits lead to super-Heston rough
volatility models with an explicit Zumbach effect
\citep{dandapani2021quadratic}, while the quadratic rough Heston model shows
how roughness combined with price feedback can address the joint SPX/VIX
calibration problem \citep{gatheral2020quadratic}. A different and potentially
complementary approach is to introduce a Hawkes-type jump component.

This second route is motivated by the clustering of large volatility moves.
Volatility indices such as the VIX often show sharp upward moves during
stressed market periods, and these moves tend to occur in clusters
\citep{bondi2024}. Moreover, jumps in the underlying and jumps in volatility are
often modeled as common shocks \citep{contkokholm2013,sepp2008}. This motivates
enriching the rough Heston framework with self-exciting jumps and common jumps
in returns and variance. The rough Hawkes--Heston model provides such an
extension while preserving affine Volterra tractability
\citep{bondi2024,abijaber2019}.
Building on this motivation, the rough Hawkes--Heston model \citep{bondi2024}
combines rough volatility, leverage, and common jumps in an affine Volterra
framework of the type studied in \cite{abijaber2019}. In that model, the jump intensity is proportional to spot
variance, and the same compensated jump measure appears in the variance equation
and the log-return equation. In the notation of this paper, the reduced-form
rough Hawkes--Heston model is
\begin{equation}
\label{eq:limit-bondi-form}
\left\{
\begin{aligned}
P_t
&=
\int_0^t \sigma_sdW_s
-
\left(
\frac12
+
\int_{(0,\infty)}
\left(e^{-\Lambda z}-1+\Lambda z\right)\nu_J(dz)
\right)
\int_0^t\sigma_s^2ds\\
&\qquad
-\Lambda
\int_0^t\int_{(0,\infty)}z
\left(
\mu(ds,dz)
-
\sigma_s^2\nu_J(dz)ds
\right),\\
\sigma_t^2
&=
G_0(t)
+
b\int_0^t k_\alpha(t-s)\sigma_s^2ds
+
\sqrt{\bar c_B}\int_0^t k_\alpha(t-s)\sigma_sdB_s\\
&\quad
+
\bar\xi
\int_{[0,t)}k_\alpha(t-s)
\int_{(0,\infty)} z
\left(
\mu(ds,dz)-\sigma_s^2\nu_J(dz)ds
\right),
\end{aligned}
\right.
\end{equation}
where $k_\alpha(t)$ is given in \eqref{k:alpha}, 
$\sigma^2$ denotes the variance process, with
$\sigma$ being the volatility process, and $W$ and $B$ are
standard Brownian motions satisfying
$d\langle W,B\rangle_t=\rho dt$.
Here, $\Lambda\ge0$ controls the size of the common downward log-price jump,
whereas $\bar\xi>0$ is the common-jump loading in the variance equation.
The common-jump measure $\mu$ has predictable compensator
\begin{equation}
\label{eq:limit-compensator}
  \nu^\mu(dt,dz)=\sigma_t^2\nu_J(dz)dt.
\end{equation}
Thus, the same compensated jump measure drives both the price equation and the
variance equation. With the sign convention $\Lambda\ge0$ and $z>0$, a jump
mark produces a downward log-price jump of size $-\Lambda z$ and a positive
volatility excitation. If the common-jump component is suppressed, for instance
by taking $\bar\xi=0$, $\Lambda=0$, and $\nu_J=0$, then
\eqref{eq:limit-bondi-form} reduces to the rough Heston model with
$V=\sigma^2$ and $\eta=\sqrt{\bar c_B}$.

Define the semimartingale driver of the variance equation in
\eqref{eq:limit-bondi-form} by
\begin{equation}
\label{eq:variance-driver-Z}
Z_t
:=
b\int_0^t\sigma_s^2ds
+\sqrt{\bar c_B}\int_0^t\sigma_s\,dB_s
+\bar\xi\int_0^t\int_{(0,\infty)}z
\left(\mu(ds,dz)-\sigma_s^2\nu_J(dz)ds\right).
\end{equation}
Thus the variance equation reads
$\sigma^2-G_0=k_\alpha*dZ$.  When $\bar\xi>0$, set
$S_{\bar\xi}(z)=\bar\xi z$.  We call a weak solution
\emph{canonical} if
\[
  \left(\mathrm{id}\times S_{\bar\xi}\right)_\#\mu
  =
  \mathcal J(Z)
  :=
  \sum_{s:\,\Delta Z_s\ne0}\delta_{(s,\Delta Z_s)},
\]
that is, the rescaled common-jump measure is the canonical jump measure of
the semimartingale driver $Z$.

The normalization in \eqref{eq:limit-bondi-form} carries the common-jump
loading $\bar\xi$ in the variance equation, whereas in \citet{bondi2024} the
jump mark enters the variance equation with unit loading. When $\bar\xi>0$,
the two parametrizations are equivalent: replacing the mark $z$ by
$\bar\xi z$, that is, replacing $\nu_J$ by its pushforward under
$z\mapsto\bar\xi z$ and $\Lambda$ by $\Lambda/\bar\xi$, maps
\eqref{eq:limit-bondi-form} into the exact parametrization of
\citet{bondi2024}. When $\bar\xi=0$, the mark rescaling degenerates and the
well-posedness statement below concerns the law of the variance process.

The microscopic construction selects an explicitly characterized attainable
parameter region. Within the
subcritical branching regime it yields $b<0$, and the rare-mark scaling yields
the finite-activity jump measure $\nu_J=2\kappa F_J$. The martingale-balanced
asymmetry between the two signs' arrival rates and price impacts attains every
$\rho\in(-1,0)$; its symmetric benchmark recovers
$\rho\in(-1/\sqrt2,0)$. The effective initial curve $G_0$
must admit the baseline-resolvent representation generated by the microscopic
immigration rate. This is a realizability condition on a freely chosen
deterministic baseline, separate from the scalar restriction $b<0$ and the
finite-activity restriction on the jump measure. The mark law $F_J$ remains free subject
to the first and second moment conditions in Section~\ref{sec:scaling}, and
may in particular be the exponential law used in the joint SPX/VIX calibration
of \citet{bondi2024}.

This work is connected to several related strands of literature. Nearly unstable Hawkes limits,
Hawkes-driven limit order book limits, and related Hawkes functional limit
theorems provide microstructural routes from order flow to macroscopic stochastic
dynamics
\citep{jaisson2015limit,horst2019scaling,horst2023convergence,horstXu2026functional}.
Marked Hawkes point measures and their diffusion or shot-noise approximations
provide a flexible framework for incorporating event-level heterogeneity and
branching effects in self-excited systems
\citep{xu2024diffusion}.
The closest comparisons are the microstructure models of
\citet{horst2022microstructure},
\citet{horstXuZhang2024pathdependent}, and
\citet{hagerHorstWagenhoferXu2026roughlognormal}, which differ from the
present construction along complementary dimensions.
\citet{horst2022microstructure} obtain Markovian affine
stochastic-volatility jump-diffusion limits from exponential-kernel marked
Hawkes systems with exogenous and induced orders. Their exponential-kernel
branching mechanism yields Markovian price--volatility co-jumps.
\citet{horstXuZhang2024pathdependent} combine heavy-tailed Hawkes market
orders with limit-order and cancellation effects having random, possibly
long-lived durations. Their limit is a rough Heston-type volatility equation
with an additional path-dependent Poisson term generating volatility spikes
and clusters.  The Poisson term acts through the volatility equation and its
path-dependent order-flow state.
Finally, \citet{hagerHorstWagenhoferXu2026roughlognormal} consider
non-self-exciting Poisson order arrivals whose direct and slowly decaying
impact on log-volatility yields a continuous rough log-normal, or rough
Bergomi, limit, with the microscopic Poisson impacts aggregating into a
continuous limiting price--volatility system.

Our rare marks are embedded in a single nearly critical
heavy-tailed Hawkes order flow. A rare marked arrival produces an immediate
price jump and, through its Hawkes descendants, an amplified rough volatility
response. Consequently, the same state-dependent finite-activity random
measure survives in both the price equation and the affine Volterra variance
equation. This yields the canonical rough Hawkes--Heston common-jump system,
together with an explicit micro-to-macro coefficient map and its associated
attainability restrictions.

Discrete-time heavy-tailed integer-valued autoregressive (INAR) constructions
give related routes to fractional Cox--Ingersoll--Ross (CIR) and rough Heston
limits \citep{wang2026scaling,wang2025rough}. Affine-transform, numerical, and
calibration methods for rough volatility are developed, among others, in
\citep{BayerFrizGatheral2015PricingRough,el2019characteristic,
callegaro2021fast,BayerBreneis2024WeakRoughHeston,guyon2025dispersion}; a broader
comparison of rough-volatility microfoundations is given by
\citet{alonso2026microstructure}.

In our construction, regular marks $i=0$ represent routine order-flow events,
whereas rare marks $i=1$ with size $z>0$ represent common shocks. Regular events
generate the continuous rough component and leverage. A rare marked event has
an immediate price impact and excites the primitive Hawkes kernel $a_T\phi$;
its branching cascade produces the macroscopic Mittag--Leffler response and,
after fractional rewriting, the common-jump Volterra feedback.

Throughout the rest of the paper, let $\Q$ denote the underlying probability measure, interpreted as
a pricing measure when the price process is considered, and $\E$ denotes
expectation under $\Q$.

\paragraph{Contributions.}
Our paper makes three main contributions. 
Their common technical feature is the derivation of the limiting rough
Hawkes--Heston dynamics from a single near-critical marked Hawkes mechanism.
The main difficulty is that rare marked events are
negligible at the microscopic event level, while their descendants are amplified
by the nearly unstable Hawkes cascade. We overcome this by first proving that
the primitive marked excitation is filtered by a single near-critical resolvent (Lemma~\ref{lem:single-resolvent}), then identifying the joint Brownian and
jump characteristics in one common filtration 
(Proposition~\ref{prop:joint-tightness-identification}), proving time simplicity
(Lemma~\ref{lem:marked-no-clustering}), and
finally passing the marked jump feedback through the singular Volterra kernel 
(Proposition~\ref{prop:raw-feedback-limit}). The Mittag--Leffler-to-fractional
rewriting with both Brownian and jump drivers is justified by the stochastic
convolution associativity result 
(Lemma~\ref{lem:stochastic-convolution-associativity}). Complete weak uniqueness
is obtained by recovering the variance-driving noise from the variance path and
identifying the remaining orthogonal price noise through a localized Girsanov
argument; see Proposition~\ref{prop:complete-weak-uniqueness}. We summarize
the three main contributions of the paper as follows.

\begin{enumerate}[label=(C\arabic*)]
  \item We introduce a Poisson-embedded marked-Hawkes order-flow model
  tailored to common price--volatility jumps. Ordinary arrivals generate the
  continuous rough volatility component, while rare jump-marked arrivals
  generate simultaneous price jumps and volatility excitation.
We obtain the full-sequence convergence of the complete
  rescaled price/variance/jump system (Theorem~\ref{thm:main}).

  \item We derive the common-jump Volterra feedback from primitive marked
  excitation. The near-critical Hawkes resolvent filters the excitation
  (Lemma~\ref{lem:single-resolvent}) and generates the macroscopic Volterra
  jump feedback (Proposition~\ref{prop:raw-feedback-limit} and
  Corollary~\ref{cor:compensated-rewriting}).

  \item We characterize the structural parameter region selected by the
  microscopic model. We obtain the explicit coefficient
  map and show that subcriticality yields $b<0$, the rare-mark scaling yields
  $\nu_J=2\kappa F_J$, and the martingale-balanced two-sign construction
  attains every $\rho\in(-1,0)$ (Theorem~\ref{thm:main}). 
  Finally, conditional on the baseline-realizability condition for $G_0$, we
  prove the converse attainability statement
  (Proposition~\ref{prop:coefficient-attainability}).
\end{enumerate}

Finally, numerical experiments are provided in Section~\ref{sec:numerical-illustrations}
to illustrate the convergence 
of our microstructural foundation to the rough Hawkes-Heston model.

\section{The rough Hawkes--Heston model}

\subsection{Reduced-form model and common-jump representation}
\label{sec:rHH}

The limiting reduced-form equation used throughout the paper is
\eqref{eq:limit-bondi-form}. The microscopic construction below identifies
$\mu$ as a time-simple common-jump measure whose compensator, relative to the
common limiting filtration, is \eqref{eq:limit-compensator}; after scaling by
$\bar\xi$, it is the canonical jump measure of the variance driver. The
compensated jump measure is
denoted by
\[
  \widetilde\mu(dt,dz)
  =
  \mu(dt,dz)-\sigma_t^2\nu_J(dz)dt.
\]
The continuous martingales
$\int_0^\cdot \sigma_sdW_s$
and
$\int_0^\cdot \sigma_sdB_s$
are orthogonal to bounded stochastic integrals against $\widetilde\mu$ in the
sense of zero predictable covariation.

\begin{assumption}[Jump integrability and sign of price jumps]
\label{ass:jump-integrability}
We assume $\Lambda\ge0$ and that the limiting jump measure $\nu_J$ is
supported on $(0,\infty)$. It satisfies
\[
  \nu_J((0,\infty))<\infty,
  \qquad
  \int_{(0,\infty)}z^2\nu_J(dz)<\infty.
\]
\end{assumption}

Note that Assumption~\ref{ass:jump-integrability} implies that
\[
\int_{(0,\infty)}(z\vee z^{2})\nu_{J}(dz)<\infty,
\qquad
  \int_{(0,\infty)}
  \left|e^{-\Lambda z}-1+\Lambda z\right|\nu_J(dz)<\infty.
\]

\begin{remark}[Mild, fractional, and compensated forms]
The microscopic marked-Hawkes derivation first yields the raw variance equation
in Mittag--Leffler mild form, where $V$ corresponds to $\sigma^2$ in the
notation of \eqref{eq:limit-bondi-form}:
\begin{align}
\label{eq:raw-limit-variance}
  V_t
  &=
  g_0(t)
  +
  \sqrt{c_B}\int_0^t K(t-s)\sqrt{V_s}dB_s
  +
  \xi\int_{[0,t)}K(t-s)\int_{(0,\infty)} z\mu(ds,dz).
\end{align}
Here, $g_0$ denotes the deterministic limiting baseline profile in the raw
Mittag--Leffler mild equation. More precisely, $g_0$ is the $L^2([0,1])$
limit of the rescaled microscopic baseline term $g_0^T$ defined later in
\eqref{eq:g0T} and specified in Assumption~\ref{ass:baseline}.
The fractional resolvent identity for $K$, followed by compensation of the
raw jump measure, transforms \eqref{eq:raw-limit-variance} into
\eqref{eq:limit-bondi-form}. Section~\ref{sec:fractional-rewriting} gives the
complete derivation of $G_0$, $b$, $\bar c_B$, and $\bar\xi$; their final
microscopic expressions are collected in Theorem~\ref{thm:main}.
\end{remark}


\subsection{Reduced-form admissibility}
\label{sec:reduced-form-admissibility}

\begin{assumption}[Reduced-form admissibility]
\label{ass:wellposed}
Define the transformed baseline curve by
\[
  G_0(t)
  :=
  g_0(t)+\frac{\mu_0}{\Gamma(1-\alpha)}
  \int_0^t\frac{(t-s)^{\alpha-1}}{\Gamma(\alpha)}g_0(s)ds.
\]
Here $g_0$ on $\mathbb R_+$ is the extension specified in
Assumption~\ref{ass:baseline}. This extension permits application of the
global-in-time affine Volterra well-posedness results cited in the proof of
Proposition~\ref{prop:wellposed} and formulation of uniqueness in
$L^2_{\mathrm{loc}}(\mathbb R_+)$. By causality, its restriction to $[0,1]$
depends only on $g_0|_{[0,1]}$.
We assume that the resulting curve $G_0$ is continuous and nondecreasing on
$\mathbb R_+$ with $G_0(0)\ge0$.
\end{assumption}

For $\alpha\in(1/2,1)$, the fractional kernel
$k_\alpha(t)=t^{\alpha-1}/\Gamma(\alpha)$ is nonnegative, nonincreasing, and
locally square-integrable. This and all shifted kernels
$k_\alpha(\cdot+1/n)$ are completely monotone and therefore satisfy
\citet[Hypothesis~2.1]{bondi2024}, as required in the proof of
Proposition~\ref{prop:wellposed}. The next result records reduced-form
well-posedness under the preceding jump-integrability and admissibility
conditions.

\begin{proposition}[Reduced-form well-posedness]
\label{prop:wellposed}
Under Assumptions~\ref{ass:jump-integrability} and
\ref{ass:wellposed}, for every $b\in\R$, $\bar c_B>0$ and $\bar\xi\ge0$,
the reduced-form variance equation in
\eqref{eq:limit-bondi-form} admits a nonnegative weak solution $\sigma^2$
whose law is unique in $L^2_{\mathrm{loc}}(\mathbb R_+)$. If
$\bar\xi>0$ and the solution is canonical in the sense defined above, then
the joint law of $(\sigma^2,\mu)$ is unique.
\end{proposition}

Before we proceed to the proof of Proposition~\ref{prop:wellposed}, we first state and prove the following technical lemma that will be used in the proof of Proposition~\ref{prop:wellposed}. 

\begin{lemma}[Recovery of the common-jump measure]
\label{lem:recover-jump-measure}
Let $\bar\xi>0$ and let $(\sigma^2,\mu)$ be a canonical weak solution. Then
the variance path determines the semimartingale driver $Z$ and its jump
measure. Consequently, the law of $\sigma^2$ determines the joint law of
$(\sigma^2,\mu)$.
\end{lemma}

\begin{proof}
For \(k_\alpha(t)=t^{\alpha-1}/\Gamma(\alpha)\), the resolvent of the first
kind is
\begin{equation}
\label{eq:first-kind-resolvent}
  L_\alpha(dt)=\frac{t^{-\alpha}}{\Gamma(1-\alpha)}\,dt.
\end{equation}
Recall Euler's beta-integral: for $a,b>0$,
\begin{equation}
\label{eq:beta-integral}
  \int_0^t(t-s)^{a-1}s^{b-1}ds
  =\frac{\Gamma(a)\Gamma(b)}{\Gamma(a+b)}t^{a+b-1}.
\end{equation}
Taking $a=\alpha$ and $b=1-\alpha$ in
\eqref{eq:beta-integral} gives, for $t>0$,
\[
  (L_\alpha*k_\alpha)(t)
  =\frac{1}{\Gamma(1-\alpha)\Gamma(\alpha)}
    \int_0^t(t-s)^{\alpha-1}s^{-\alpha}ds
  =1.
\]
Applying
\citet[Proposition~4]{bondi2024affine} with \(F\equiv1\) to
\(\sigma^2-G_0=k_\alpha*dZ\) yields
\[
  Z_t=(L_\alpha*(\sigma^2-G_0))(t)
\]
for Lebesgue-a.e.\ \(t\), almost surely. Appendix~\ref{app:measurable-recovery}
constructs a Borel map \(\mathcal R:L^2_{\mathrm{loc}}\to D_{\mathrm{loc}}\)
that recovers the c\`adl\`ag version \(Z=\mathcal R(\sigma^2)\) and verifies
that the jump-measure map is Borel. Canonicity then gives
\[
  \mu
  =
  \left(\mathrm{id}\times S_{\bar\xi}^{-1}\right)_\#
  \mathcal J(\mathcal R(\sigma^2)),
\]
so \(\mu\) is a Borel functional of \(\sigma^2\).
\end{proof}

\begin{proof}[Proof of Proposition~\ref{prop:wellposed}]
The fractional kernel $k_\alpha$ and its shifted kernels
$k_\alpha(\cdot+1/n)$, $n\in\mathbb N$, are completely monotone and satisfy
the kernel conditions in \citet[Hypothesis~2.1]{bondi2024} for the affine
Volterra variance equation in \eqref{eq:limit-bondi-form}. The
resolvent of the first kind of $k_\alpha$ is the nonnegative, nonincreasing
measure $L_\alpha$ defined in \eqref{eq:first-kind-resolvent}; see also
\citep{abijaber2019} and
\citep[Hypothesis~2.1 and Remark~2.3]{bondi2024}.
By Assumption~\ref{ass:wellposed}, the prescribed curve $G_0$ is continuous and nondecreasing with
$G_0(0)\ge0$, matching the initial-curve hypothesis of
\citep[Hypothesis~2.2]{bondi2024}.

Suppose first that $\bar\xi>0$. Define
\[
  S_{\bar\xi}(z):=\bar\xi z,
  \qquad
  \nu^B:=(S_{\bar\xi})_\#\nu_J,
  \qquad
  \mu^B:=(\mathrm{id}\times S_{\bar\xi})_\#\mu.
\]
Then
\[
\bar\xi\int_{(0,\infty)}z
\left(\mu(ds,dz)-\sigma_s^2\nu_J(dz)ds\right)
=
\int_{(0,\infty)}y
\left(\mu^B(ds,dy)-\sigma_s^2\nu^B(dy)ds\right),
\]
and
\[
  \nu^B(\{0\})=0,
  \qquad
  \int_{(0,\infty)}y^2\nu^B(dy)
  =\bar\xi^2\int_{(0,\infty)}z^2\nu_J(dz)<\infty.
\]
Thus the mark-rescaled variance equation has unit jump loading, diffusion
coefficient $c=\bar c_B$, and jump measure $\nu^B$. Weak existence and
state-process uniqueness follow from \citet[Theorem~2.13]{abijaber2021},
\citet[Lemma~9 and Corollary~12]{bondi2024affine}, and
\citet[Remark~2.3]{bondi2024}. Translating the marks back gives a solution of
the variance equation in \eqref{eq:limit-bondi-form}, and joint uniqueness of
the canonical pair follows from
Lemma~\ref{lem:recover-jump-measure}.

If $\bar\xi=0$, the variance equation in \eqref{eq:limit-bondi-form} reduces
to the continuous affine Volterra case. Weak existence follows from
\citet[Theorem~2.13]{abijaber2021} with zero jump measure, and state-process
uniqueness follows from
\citet[Lemma~9 and Corollary~12]{bondi2024affine}.
If $\mu$ is still
required for the price equation, it may be realized on a product extension
as a random measure with compensator $\sigma_t^2\nu_J(dz)dt$. In this case
the uniqueness conclusion concerns the state process $\sigma^2$.
\end{proof}

\begin{remark}[On the sign of the drift]
\label{rem:drift-sign}
Proposition~\ref{prop:wellposed} establishes reduced-form well-posedness for
every $b\in\R$. The microscopic model under the subcritical scaling regime of
Assumption~\ref{ass:subcritical} generates the subfamily $b<0$, via the identity
$b=\bar\xi\int_{(0,\infty)}z\nu_J(dz)-\mu_0/\Gamma(1-\alpha)$ derived in
Section~\ref{sec:fractional-rewriting}. Critical and supercritical scalings
constitute separate regimes.
\end{remark}

Theorem~\ref{thm:subseq} first proves that every subsequential variance/jump
limit is canonical. Proposition~\ref{prop:wellposed} and
Lemma~\ref{lem:recover-jump-measure} then identify the law of the
variance/jump pair. The common-filtration time-simplicity, compensator, and
orthogonality results, together with the recovery map in
Proposition~\ref{prop:complete-weak-uniqueness}, identify the complete joint
law including the price-driving Brownian motion.

\begin{remark}[On the monotonicity of $G_0$]
\label{rem:G0-monotonicity}
The monotonicity condition on $G_0$ matches the initial-curve admissibility
condition in the reduced-form rough Hawkes--Heston well-posedness theory;
see \citet[Hypothesis~2.2 and Remark~2.3]{bondi2024}, as well as
\citet[Theorem~2.13]{abijaber2021} and
\citet[Lemma~9 and Corollary~12]{bondi2024affine} for the underlying affine
Volterra existence and uniqueness results.
A sufficient condition is that $g_0$ is continuous, nonnegative, and
nondecreasing. Indeed, after writing
$(k_\alpha*g_0)(t)=\int_0^t k_\alpha(s)g_0(t-s)ds$, monotonicity follows
directly by comparing the common integration interval and using
$k_\alpha,g_0\ge0$. Hence
$G_0=g_0+\mu_0\Gamma(1-\alpha)^{-1}k_\alpha*g_0$ is nondecreasing as well.
\end{remark}

The reduced-form system above is the target of the microscopic construction
developed in the next sections. After introducing the marked Hawkes order-flow
model and the rescaled microscopic price, variance proxy, and common jump
measure, we state the main convergence theorem in
Section~\ref{sec:main-convergence-theorem}.

\section{Microscopic marked Hawkes model}
\label{sec:model}

\subsection{Model overview and economic mechanism}
\label{sec:model-overview}

For each scaling parameter $T$, microscopic time runs on $[0,T]$, while
$t\in[0,1]$ denotes macroscopic time through $t\mapsto tT$. Accepted events
share the endogenous activity process $\lambda^T$ and are classified as either
regular signed quote changes or rare events carrying a positive mark $z$.
Regular events generate the continuous price martingale and rough continuous
volatility, while a marked event produces a common price jump $-\Lambda z$ and
excites subsequent activity proportionally to $z$. This construction combines
the two-sign Hawkes order-flow mechanism with a common-jump volatility channel
\citep{bowsher2007modelling,bacry2013modelling,
eleuchfukasawarosenbaum2018,contkokholm2013,bondi2024}.

\begin{center}
\renewcommand{\arraystretch}{1.25}
\begin{tabularx}{\textwidth}{@{}>{\raggedright\arraybackslash}p{0.19\textwidth}
>{\raggedright\arraybackslash}X>{\raggedright\arraybackslash}X
>{\raggedright\arraybackslash}X@{}}
\toprule
Microscopic event & Immediate price effect & Effect on future activity &
Macroscopic contribution \\
\midrule
Positive regular event & Small upward move & Ordinary Hawkes excitation &
Continuous return noise and rough volatility \\
Negative regular event & Larger downward move & Stronger regular loading,
controlled by $\beta$ & Leverage and rough volatility \\
Rare marked event of size $z$ & Common downward jump $-\Lambda z$ & Amplified
excitation $c_Jr_Tz\,a_T\phi$ & Common price--variance jump and volatility
clustering \\
\bottomrule
\end{tabularx}
\end{center}

The principal parameters have the following roles. The pair $(a_T,\phi)$
controls endogeneity and the memory profile; $p_T$ is the probability that an
accepted event is a common shock and $F_J$ is its size distribution; $c_Jr_T$
controls the additional excitation caused by such a shock; $q$ balances the
arrival frequencies and price impacts of positive and negative regular events;
$\beta$ controls the stronger variance loading of negative regular events; and
$\hat h^T$ represents exogenous information flow and the initial activity
profile. We now give the precise construction at a fixed microscopic scale,
before imposing the asymptotic relations among these parameters.
\subsection{Poisson embedding and marked Hawkes intensity}
\label{sec:poisson-intensity}

Fix $\alpha\in(1/2,1)$ and $T\ge1$. For the finite-$T$ construction, take $a_T,p_T\in(0,1)$,
$r_T,c_J>0$, a non-negative probability density $\phi$ on $\mathbb R_+$,
a probability measure $F_J$ on $(0,\infty)$, and a non-negative locally
bounded measurable baseline $\hat h^T$. Fix also a price-impact ratio
$q\in(0,1]$ and an asymmetric variance loading $\beta>1$. Define the
sign-specific activity
weights
\begin{equation}
\label{eq:sign-activity-weights}
  \omega_+:=\frac{2}{1+q},
  \qquad
  \omega_-:=\frac{2q}{1+q}.
\end{equation}
Thus $\omega_++\omega_-=2$ and $q\omega_+=\omega_-$. The first identity
keeps total event activity unchanged; the second balances the predictable
drift of the raw signed price defined below.

For each sign $\pm$, let
\[
  \Pi_\pm^T(dt,du,di,dz)
\]
be independent Poisson random measures on
$(0,\infty)\times(0,\infty)\times\{0,1\}\times(0,\infty)$ with deterministic
intensity
\[
  dtdu\eta_T(di,dz),
\]
where
\[
  \eta_T(di,dz)
  :=(1-p_T)\delta_0(di)\delta_1(dz)
  +p_T\delta_1(di)F_J(dz).
\]
Here $u$ is the thinning variable, $i=0$ and $i=1$ indicate regular
and marked events, respectively, and $z$ records the shock size when
$i=1$; when $i=0$, $z$ is fixed at $1$. Given the microscopic intensity $\lambda^T$, define
\begin{equation}
\label{eq:embedded-mu}
  \mu^{T,\pm}(dt,di,dz)
  :=\int_0^\infty\1_{\{u\le \omega_\pm\lambda_{t-}^T\}}
  \Pi_\pm^T(dt,du,di,dz).
\end{equation}

The regular and jump-marked measures are
\[
  \mu_{\rm reg}^{T,\pm}(dt)
  :=\mu^{T,\pm}(dt,\{0\},(0,\infty)),
  \qquad
  \mu_J^{T,\pm}(dt,dz)
  :=\mu^{T,\pm}(dt,\{1\},dz).
\]
The common microscopic jump measure is
\begin{equation}
\label{eq:common-jump-measure}
  \mu_{J,{\rm com}}^T(dt,dz)
  :=\mu_J^{T,+}(dt,dz)+\mu_J^{T,-}(dt,dz).
\end{equation}

Set
\begin{equation}
\label{eq:regular-counts}
  N_t^{T,\pm,{\rm reg}}:=\mu_{\rm reg}^{T,\pm}([0,t]).
\end{equation}

\begin{remark}[Two-sign balance and common shocks]
\label{rem:two-sign-poisson}
The signs of $\Pi_+^T$ and $\Pi_-^T$ describe regular upward and downward
quote changes. When $q<1$, upward moves are more frequent but smaller, and the
identity $q\omega_+=\omega_-$ makes their uncentered signed flow a local
martingale. Negative moves have variance loading $\beta>1$, which produces the
leverage correlation in \eqref{eq:rho}. Marked events from both channels are
aggregated in $\mu_{J,{\rm com}}^T$ and therefore enter the limiting price and
variance equations through the same compensated jump measure.
\end{remark}

The predictable compensators of the regular measures
$\mu_{\rm reg}^{T,\pm}$ and of their common marked sum
$\mu_{J,{\rm com}}^T$ are, respectively,
\[
  \nu_{\rm reg}^{T,\pm}(dt)
  :=(1-p_T)\omega_\pm\lambda_{t-}^Tdt,
\]
and
\begin{equation}
\label{eq:common-jump-compensator}
  \nu_{J,{\rm com}}^T(dt,dz)
  :=p_T(\omega_++\omega_-)\lambda_{t-}^TF_J(dz)dt
  =2p_T\lambda_{t-}^TF_J(dz)dt.
\end{equation}
Both compensators are proportional to the same process $\lambda^T$.
This common dependence is the microscopic source of the state-dependent
limiting jump compensator $V_t\nu_J(dz)dt$; see
\eqref{eq:rescaled-compensator} for its rescaled form and
Proposition~\ref{prop:joint-tightness-identification} for its limiting
identification.

\paragraph{Marked Hawkes intensity.}

Let
\[
  \bar\mu_{\rm reg}^T(dt)
  :=\frac{\mu_{\rm reg}^{T,+}(dt)+\beta\mu_{\rm reg}^{T,-}(dt)}
  {\omega_++\beta\omega_-},
  \qquad \beta>1.
\]
The denominator makes the compensator of $\bar\mu_{\rm reg}^T$ equal to
$(1-p_T)\lambda_{t-}^Tdt$. The parameter $\beta$ controls the stronger
variance response to the negative side, while $q$ controls the activity--tick
tradeoff in \eqref{eq:sign-activity-weights}.
The factor $2c_J$ in the reproduction ratio is computed in the proof of
Lemma~\ref{lem:finite-T-subcriticality}.

Let $\hat h^T\colon[0,\infty)\to\mathbb{R}_+$ be a deterministic
non-negative locally bounded measurable function, called the
\emph{exogenous baseline intensity}.
It represents arrivals driven by external information flow or initial
system conditions.  Self-excited arrivals enter through the two convolution
terms in \eqref{eq:lambda-marked}.
The proofs use this baseline through the scaling condition in
Assumption \ref{ass:baseline}.

The common scalar microscopic intensity is
\begin{equation}
\label{eq:lambda-marked}
\lambda_t^T
:=
\hat h^T(t)
+
\int_{(0,t)}a_T\phi(t-s)\bar\mu_{\rm reg}^T(ds)
+
c_Jr_T\int_{(0,t)}a_T\phi(t-s)
\int_{(0,\infty)} z\mu_{J,{\rm com}}^T(ds,dz).
\end{equation}
The three terms on the right hand side of \eqref{eq:lambda-marked} are the baseline, the feedback from regular
events, and the additional feedback from marked events. A marked event of size $z$ contributes $c_Jr_Tz\,a_T\varphi$ to future activity. The renewal equation in Section \ref{sec:renewal-limits} propagates this primitive-kernel excitation through the same near-critical resolvent as the regular fluctuations.

\subsection{Rescaled microscopic observables}
\label{sec:microscopic-observables}

The objects whose joint macroscopic limit will be studied are the rescaled
variance proxy, the rescaled common-jump measure, and the microscopic log-price.
They are all defined here so that the full model is visible before the limit
analysis begins.

The variance proxy is
\begin{equation}
\label{eq:VT-definition}
  V_t^T:=\mathfrak c_T\lambda_{tT-}^T,
  \qquad
  \mathfrak c_T:=\frac{\theta(1-a_T)}{m_\star T^{\alpha-1}}.
\end{equation}
Here, $\theta>0$ and $m_\star>0$ are normalization constants. The constant $m_\star$ fixes the scale of the rescaled martingale array $\widehat{\mathbf N}_t^T$, and $\theta$ appears in the limiting bracket normalization; see Proposition~\ref{prop:regular-fclt}.

Define the rescaled common-jump measure and its compensator by
\begin{equation}
\label{eq:rescaled-common-jump-measure}
  \bar\mu_J^T(du,dz):=\mu_{J,{\rm com}}^T(Tdu,dz),
  \qquad
  \bar\nu_J^T(du,dz)
  :=2p_TT\lambda_{Tu-}^TF_J(dz)du.
\end{equation}

The microscopic continuous price component is
\begin{equation}
\label{eq:PTc}
P_t^{T,c}
:=
\sqrt{\frac{\theta}{2q}}
\sqrt{\frac{1-a_T}{m_\star T^\alpha}}
\left(qN_{tT}^{T,+,{\rm reg}}-N_{tT}^{T,-,{\rm reg}}\right)
-\frac12\int_0^tV_s^Tds.
\end{equation}
The jump price component is
\begin{equation}
\label{eq:PTJ}
P_t^{T,J}
:=
-\Lambda\int_0^t\int_{(0,\infty)} z
\left(\bar\mu_J^T(ds,dz)-\bar\nu_J^T(ds,dz)\right)
-
\int_0^t\int_{(0,\infty)}\left(e^{-\Lambda z}-1+\Lambda z\right)
\bar\nu_J^T(ds,dz).
\end{equation}
Set
\[
  P^T:=P^{T,c}+P^{T,J}.
\]

\begin{remark}[Sources of the model components]
\label{rem:price-inserted}
The marked Hawkes order flow generates the rough continuous variance
component, the martingale-balanced leverage correlation $\rho$, the
state-dependent limiting compensator \eqref{eq:limit-compensator}, and the
Volterra jump feedback.  The price definition includes the drift correction
$-\frac12\int_0^tV_s^Tds$ in \eqref{eq:PTc} and the compensator correction in
\eqref{eq:PTJ}, which normalize the limiting asset price $S=S_0e^P$ as a local
martingale under the pricing measure.  The loading $-\Lambda z$ parameterizes
the instantaneous price impact of a common shock, while the marked Hawkes
mechanism determines its state-dependent intensity and rough volatility
response.
\end{remark}

\subsection{Near-critical scaling and subcriticality}
\label{sec:scaling}

We now specify the asymptotic regime. The three scales below serve different
purposes: near-criticality produces a long Hawkes cascade, rare-event scaling
leaves only finitely many marked shocks on the macroscopic horizon, and the
amplification of marked excitation makes the volatility response of each such
shock non-degenerate.

For the fixed $\alpha\in(1/2,1)$, let $a_T\in(0,1)$ satisfy
\begin{equation}
\label{eq:near-critical}
  1-a_T\sim\mu_0T^{-\alpha},
  \qquad \mu_0>0.
\end{equation}

\begin{assumption}[Canonical heavy-tailed Hawkes kernel]
\label{ass:kernel-tail-regularity}
For the fixed $\alpha\in(1/2,1)$, we take
\begin{equation}
\label{eq:canonical-kernel}
  \phi(t)=\alpha(1+t)^{-1-\alpha},
  \qquad t\ge0.
\end{equation}
Then $\phi$ is a non-negative probability density on $\mathbb R_+$,
continuous, locally bounded, decreasing, and satisfies
\begin{equation}
\label{eq:kernel-tail-regularity}
  \phi(t)\sim \alpha t^{-1-\alpha},
  \qquad t\to\infty.
\end{equation}
\end{assumption}

Define the associated Hawkes resolvent kernel as in \cite{jaisson2016rough} by
\[
  \Psi_T:=\sum_{k\ge1}(a_T\phi)^{*k}.
\]
Since $a_T<1$ and $\|\phi\|_1=1$, this series is well defined and satisfies
\[
  \Psi_T=a_T\phi+\Psi_T*(a_T\phi).
\]
Its Mittag--Leffler scaling is stated and proved at the beginning of
Section~\ref{sec:renewal-limits}, where it first enters the limit analysis.

Let
\begin{equation}
\label{eq:pT}
  p_T:=(1-a_T)^2\bar p_T,
  \qquad \bar p_T\to\bar p>0.
\end{equation}
Let $F_J$ be a probability measure on $(0,\infty)$ with
\begin{equation}
\label{eq:FJ-moments}
  m_1:=\int_{(0,\infty)} zF_J(dz)<\infty,
  \qquad
  m_2:=\int_{(0,\infty)} z^2F_J(dz)<\infty.
\end{equation}
The amplified marked-excitation strength satisfies
\begin{equation}
\label{eq:rT}
  r_T\sim\frac{r_\star}{1-a_T},
  \qquad r_\star>0.
\end{equation}

The activity scaling determines the orders in these rates.
Since $V_t^T=\mathfrak c_T\lambda_{tT-}^T$ and
$\mathfrak c_T$ is of order $T^{1-2\alpha}$, an order-one variance level
corresponds to microscopic intensity of order $T^{2\alpha-1}$. Over a time
interval of microscopic length $T$, the total number of events is therefore of
order $T^{2\alpha}$. On the other hand,
$p_T\asymp(1-a_T)^2\asymp T^{-2\alpha}$, so the number of marked shocks remains
of order one. Thus the marked component survives as a finite-activity jump
measure. Finally, $r_T\asymp T^\alpha$ amplifies the excitation created by each
rare shock; after propagation through the same near-critical resolvent as the
regular flow, this produces a non-vanishing rough Volterra response.

\begin{assumption}[Marked near-critical subcriticality]
\label{ass:subcritical}
The rare but amplified jump-marked reproduction remains subcritical in the
near-critical regime:
\[
  2c_Jr_\star\bar p m_1<1,
  \qquad
  m_1=\int_{(0,\infty)}zF_J(dz).
\]
With the limiting constants defined later in \eqref{eq:nuJ} and
\eqref{eq:xi}, this condition becomes
\[
  \xi\int_{(0,\infty)}z\nu_J(dz)<\mu_0,
\]
which is equivalent, after the Riemann--Liouville rewriting, to the negative
Volterra drift condition $b<0$.
\end{assumption}

The asymptotic subcriticality condition must also hold at the microscopic
finite-$T$ level. The next lemma shows that the marked Hawkes branching system
is genuinely subcritical for all sufficiently large $T$, which ensures
existence and non-explosion of the Poisson-embedded construction.

\begin{lemma}[Finite-$T$ subcriticality]
\label{lem:finite-T-subcriticality}
Under Assumption \ref{ass:subcritical}, for all sufficiently large $T$,
\begin{equation}
\label{eq:finite-T-subcriticality}
  a_T\left((1-p_T)+2c_Jr_Tp_Tm_1\right)<1.
\end{equation}
Consequently the marked Hawkes system admits a unique non-explosive
Poisson-embedded solution.
\end{lemma}

\begin{proof}
Using \eqref{eq:near-critical}, \eqref{eq:pT}, and \eqref{eq:rT},
\[
2c_Jr_Tp_Tm_1
=
2c_Jr_\star\bar p\mu_0m_1T^{-\alpha}+o(T^{-\alpha}),
\qquad
p_T=\mathcal O(T^{-2\alpha}).
\]
Since $a_T=1-\mathcal O(T^{-\alpha})$, the products
$a_Tp_T=p_T+\mathcal O(T^{-3\alpha})$ and
$a_T(2c_Jr_Tp_Tm_1)=2c_Jr_Tp_Tm_1+\mathcal O(T^{-2\alpha})$ differ from their
leading terms only by $o(T^{-\alpha})$.
Therefore,
\[
\begin{aligned}
1-a_T\left((1-p_T)+2c_Jr_Tp_Tm_1\right)
&=
(1-a_T)+a_Tp_T-2a_Tc_Jr_Tp_Tm_1  \\
&=
\mu_0\left(1-2c_Jr_\star\bar p m_1\right)T^{-\alpha}
+o(T^{-\alpha}).
\end{aligned}
\]
The coefficient is positive by Assumption~\ref{ass:subcritical}, which proves
\eqref{eq:finite-T-subcriticality} for all large $T$.

For non-explosion, use the immigration--birth representation of linear marked
Hawkes processes \citep{hawkes1974cluster,bremaud1981}. Its mean offspring
matrix is
\begin{equation}
\label{eq:mean-offspring-matrix}
M_T=a_T\left((1-p_T)\chi+p_Tr_Tm_1\chi_J\right),
\end{equation}
where
\[
\chi:=\frac1{\omega_++\beta\omega_-}
\begin{pmatrix}
\omega_+&\beta\omega_+\\
\omega_-&\beta\omega_-
\end{pmatrix},
\qquad
\chi_J:=c_J
\begin{pmatrix}
\omega_+&\omega_+\\
\omega_-&\omega_-
\end{pmatrix}.
\]
For $e:=(\omega_+,\omega_-)^\top$, the identities
$\chi e=e$ and $\chi_Je=2c_Je$ show that $M_T$ is rank one with
\begin{equation}
\label{eq:offspring-spectral-radius}
\rho(M_T)=a_T\left((1-p_T)+2c_Jr_Tp_Tm_1\right).
\end{equation}

By \eqref{eq:finite-T-subcriticality}, this mean offspring number is strictly
smaller than one for all large $T$. If $Z_n^T\in\mathbb N^2$ denotes the
sign-count vector in generation $n$ of a cluster, then, conditional on a
single initial parent of type $j$,
\[
  \E\left[Z_n^T\mid Z_0^T=e_j\right]=M_T^ne_j.
\]
Consequently,
\[
  \E\left[
    \left.
    \sum_{n\ge0}\mathbf 1^\top Z_n^T
    \right|Z_0^T=e_j
  \right]
  =
  \mathbf 1^\top(I-M_T)^{-1}e_j
  <\infty.
\]
Every cluster is therefore finite almost surely and has finite expected size.
Since immigrants are generated by the locally integrable baseline $\hat h^T$,
only finitely many clusters are born on compact time intervals almost surely.
Thus the Poisson-embedded marked Hawkes system admits a unique non-explosive
solution.
This completes the proof.
\end{proof}

Choose and fix $T_0<\infty$ so that
\eqref{eq:finite-T-subcriticality} holds and the Poisson-embedded microscopic
model is well defined and subcritical for every $T\ge T_0$. All stochastic
asymptotic families below are indexed by $T\ge T_0$.

\section{Renewal representation and driver limits}
\label{sec:renewal-limits}

This section first derives the scalar renewal equation and then establishes
the separate tightness and convolution limits for its regular and marked
drivers. These estimates establish the product tightness used in
Proposition~\ref{prop:joint-tightness-identification}, where the martingale,
bracket, and compensator identities are placed in one common limiting
filtration.

We begin with the deterministic resolvent asymptotics that transmit the
microscopic heavy-tailed Hawkes memory to the macroscopic Volterra kernel.

\begin{proposition}[Near-critical resolvent scaling for the canonical kernel]
\label{prop:resolvent-scaling}
Under Assumption \ref{ass:kernel-tail-regularity} and
$1-a_T\sim\mu_0T^{-\alpha}$ with $\mu_0>0$, the canonical kernel satisfies
\begin{equation}
\label{eq:canonical-laplace-tail}
  1-\widehat\phi(\lambda)
  \sim
  \Gamma(1-\alpha)\lambda^\alpha,
  \qquad \lambda\downarrow0.
\end{equation}
Consequently, Lemma~\ref{lem:local-strong-renewal} applies to
$\Psi_T:=\sum_{k\ge1}(a_T\phi)^{*k}$, yielding
\begin{equation}
\label{eq:psi-scaling}
  T^{1-\alpha}\Psi_T(Tu)\to K(u),
  \qquad u>0,
\end{equation}
locally uniformly on compact subsets of $(0,\infty)$, where
\[
  \widehat K(\lambda)
  =
  \frac{1}{\mu_0+\Gamma(1-\alpha)\lambda^\alpha},
\]
and there exists $C<\infty$ such that
\begin{equation}
\label{eq:psi-bound-R}
  T^{1-\alpha}\Psi_T(Tu)
  \le Cu^{\alpha-1},
  \qquad u>0,\ T\ge1.
\end{equation}
\end{proposition}

\begin{proof}
The canonical kernel $\phi(t)=\alpha(1+t)^{-1-\alpha}$ is a probability
density, since $\int_0^\infty\alpha(1+t)^{-1-\alpha}dt=1$, and is regularly
varying at infinity with $\phi(t)\sim\alpha t^{-1-\alpha}$. Since
$\int_0^\infty\phi(t)dt=1$, we have
\[
  1-\widehat\phi(\lambda)
  =
  \int_0^\infty\left(1-e^{-\lambda t}\right)\phi(t)dt.
\]
By regular variation of $\phi$, the Abelian theorem for regularly varying
densities implies that
\[
  1-\widehat\phi(\lambda)
  \sim
  \alpha\lambda^\alpha\int_0^\infty(1-e^{-x})x^{-1-\alpha}dx,
  \qquad \lambda\downarrow0.
\]
Integration by parts yields
\[
  \int_0^\infty(1-e^{-x})x^{-1-\alpha}dx
  =
  \frac{1}{\alpha}\int_0^\infty e^{-x}x^{-\alpha}dx
  =
  \frac{\Gamma(1-\alpha)}{\alpha}.
\]
This proves \eqref{eq:canonical-laplace-tail}.

Lemma \ref{lem:local-strong-renewal} applies to the canonical kernel by
construction, and yields the local uniform convergence and fractional
envelope. This completes the proof.
\end{proof}

\subsection{Renewal representation}
\label{sec:renewal}

Recall the regular counts $N^{T,\pm,{\rm reg}}$ from
\eqref{eq:regular-counts}.  Define regular martingales
\[
  M_t^{T,\pm,{\rm reg}}
  :=N_t^{T,\pm,{\rm reg}}-
  \int_0^t(1-p_T)\omega_\pm\lambda_{s-}^Tds,
\]
and
\[
  \bar M_t^T
  :=\frac{M_t^{T,+,{\rm reg}}+\beta M_t^{T,-,{\rm reg}}}
  {\omega_++\beta\omega_-}.
\]
Then,
\begin{equation}
\label{eq:regular-measure-decomposition}
  \bar\mu_{\rm reg}^T(ds)=(1-p_T)\lambda_{s-}^Tds+d\bar M_s^T.
\end{equation}
Let
\[
  \phi_T:=a_T(1-p_T)\phi,
  \qquad
  \Psi_T^{\rm reg}:=\sum_{k\ge1}\phi_T^{*k}.
\]
The regular resolvent satisfies the renewal identity
\begin{equation}
\label{eq:psi-reg-renewal}
  \Psi_T^{\rm reg}=\phi_T+\phi_T*\Psi_T^{\rm reg}.
\end{equation}
Substituting the martingale decomposition
\eqref{eq:regular-measure-decomposition} into \eqref{eq:lambda-marked} gives
\begin{align}
\label{eq:lambda-pre-renewal}
\lambda_t^T
&=\hat h^T(t)
+\int_0^t\phi_T(t-s)\lambda_s^Tds
+\int_{(0,t)}a_T\phi(t-s)d\bar M_s^T\nonumber\\
&\quad
+c_Jr_T\int_{(0,t)}a_T\phi(t-s)
\int_{(0,\infty)} z\mu_{J,{\rm com}}^T(ds,dz).
\end{align}
Solving this renewal equation yields
\begin{align}
\label{eq:lambda-renewal}
\lambda_t^T
&=\hat h^T(t)+\int_0^t\Psi_T^{\rm reg}(t-s)
\hat h^T(s)ds
+\int_{(0,t)}H_T(t-s)d\bar M_s^T\nonumber\\
&\quad
+c_Jr_T\int_{(0,t)}H_T(t-s)\int_{(0,\infty)} z\mu_{J,{\rm com}}^T(ds,dz),
\end{align}
where
\begin{equation}
\label{eq:HT}
  H_T:=\left(\delta_0+\Psi_T^{\rm reg}\right)*a_T\phi.
\end{equation}

Solving the renewal equation filters the primitive-kernel marked excitation through the same near-critical resolvent as the ordinary order flow. The identity in the following lemma makes this reduction exact.

\begin{lemma}[Single-resolvent identity]
\label{lem:single-resolvent}
The kernel $H_T$ satisfies the exact identity
\begin{equation}
\label{eq:HT-identity}
  H_T=\frac{1}{1-p_T}\Psi_T^{\rm reg}.
\end{equation}
\end{lemma}

\begin{proof}
Taking Laplace transforms in the renewal identity
\eqref{eq:psi-reg-renewal} gives
\[
  1+\widehat\Psi_T^{\rm reg}(\lambda)
  =\frac{1}{1-a_T(1-p_T)\widehat\phi(\lambda)}.
\]
Hence,
\[
  \widehat H_T(\lambda)
  =\frac{a_T\widehat\phi(\lambda)}{1-a_T(1-p_T)\widehat\phi(\lambda)}
  =\frac{1}{1-p_T}\widehat\Psi_T^{\rm reg}(\lambda).
\]
The Laplace inversion proves the identity \eqref{eq:HT-identity}.
\end{proof}

The regular resolvent is generated by
$\widetilde a_T\phi$, where $\widetilde a_T:=a_T(1-p_T)$, rather than by
$a_T\phi$. Its critical gap nevertheless has the same leading order:
\[
  1-\widetilde a_T
=1-a_T(1-p_T)
=(1-a_T)+a_Tp_T
  =\mu_0T^{-\alpha}+o(T^{-\alpha}),
\]
because $p_T=\mathcal O(T^{-2\alpha})$. Therefore
Proposition~\ref{prop:resolvent-scaling}, applied with $\widetilde a_T$ in
place of $a_T$, gives the same Mittag--Leffler limit
\[
  \widehat K(\lambda)
  =\frac{1}{\mu_0+\Gamma(1-\alpha)\lambda^\alpha}
\]
for $T^{1-\alpha}\Psi_T^{\rm reg}(T\cdot)$.
Since $(1-p_T)^{-1}\to1$ as $T\to\infty$, the identity \eqref{eq:HT-identity} gives the same
limit for $T^{1-\alpha}H_T(T\cdot)$. Lemma~\ref{lem:resolvent-perturbation},
proved in Appendix~\ref{app:resolvent-perturbation}, also provides the uniform
fractional envelope used in all later stochastic-convolution and tightness
arguments.

Since an event is classified as regular with probability $1-p_T$, the regular
resolvent differs slightly from the unmarked near-critical resolvent. The next
lemma shows that this perturbation is negligible at the macroscopic scale.

\begin{lemma}[Regular resolvent perturbation]
\label{lem:resolvent-perturbation}
There exists $\epsilon_T\to0$, with $\epsilon_T=\mathcal O(T^{-\alpha})$, such that
\[
  T^{1-\alpha}\left|\Psi_T^{\rm reg}(Tu)-\Psi_T(Tu)\right|
  \le C\epsilon_Tu^{\alpha-1},
  \qquad u\in(0,1].
\]
Consequently, as $T\to\infty$, $T^{1-\alpha}H_T(Tu)\to K(u)$ uniformly on $[\varepsilon,1]$
for every $\varepsilon\in(0,1)$, with the uniform fractional envelope: there
exists a constant $C<\infty$, independent of $T$ and $u$, such that for all
sufficiently large $T$,
\[
  T^{1-\alpha}H_T(Tu)
  \le C u^{\alpha-1},
  \qquad u\in(0,1].
\]
\end{lemma}

\begin{proof}
The critical-gap calculation above identifies the limiting parameter as
$\mu_0$. See Appendix \ref{app:resolvent-perturbation} for the quantitative
perturbation estimate and uniform envelope.
\end{proof}

Multiplying \eqref{eq:lambda-renewal} by $\mathfrak c_T$ and evaluating at
$tT$, we get
\begin{equation}
\label{eq:VT-mild}
  V_t^T=g_0^T(t)+\mathcal M_t^T+\mathcal J_t^T,
\end{equation}
where
\begin{equation}
\label{eq:g0T}
  g_0^T(t):=\mathfrak c_T\left(\hat h^T(tT)
  +\int_0^{tT}\Psi_T^{\rm reg}(tT-s)\hat h^T(s)ds\right),
\end{equation}
\begin{equation}
\label{eq:MT}
  \mathcal M_t^T:=\mathfrak c_T\int_{(0,tT)}H_T(tT-s)d\bar M_s^T,
\end{equation}
\begin{equation}
\label{eq:JT-raw}
  \mathcal J_t^T:=\mathfrak c_Tc_Jr_T\int_{(0,tT)}H_T(tT-s)
  \int_{(0,\infty)} z\mu_{J,{\rm com}}^T(ds,dz).
\end{equation}

\begin{assumption}[Baseline scaling]
\label{ass:baseline}
There exists a continuous, nonnegative, nondecreasing deterministic function
$g_0:[0,1]\to\R_+$ such that, as $T\to\infty$,
\[
  \|g_0^T-g_0\|_{L^2([0,1])}\to0.
\]
We extend $g_0$ to $\mathbb R_+$ by setting $g_0(t):=g_0(1)$ for
$t>1$; this extension is continuous, nonnegative, and nondecreasing.
\end{assumption}

\begin{remark}[A canonical baseline scaling]
\label{rem:muhat_canonical}
A canonical family satisfying Assumption \ref{ass:baseline} is given, on the
microscopic horizon $[0,T]$, by
\[
  \hat h^T(t)
  =
  \frac{m_\star}{\theta}T^{\alpha-1}h(t/T),
  \qquad 0\le t\le T,
\]
with any non-negative locally bounded extension to $[0,\infty)$, where
$h\in C([0,1],\mathbb R_+)$ is fixed and nondecreasing. Then, as $T\to\infty$,
\[
  g_0^T(t)\to
  \mu_0\int_0^t K(t-u)h(u)du
\]
in $L^2([0,1])$. The limiting function is continuous, nonnegative, and
nondecreasing. Hence Assumption \ref{ass:baseline} is satisfied.
\end{remark}

\begin{remark}[Endogenous regular drift after fractional rewriting]
\label{rem:endogenous-regular-drift}
In the renewal mild representation, the Hawkes resolvent absorbs predictable
regular feedback. Fractional rewriting produces the effective drift, where the
near-critical gap contributes negatively and the predictable jump feedback
contributes positively. Their exact coefficients and the implication $b<0$
are derived in Corollary~\ref{cor:compensated-rewriting}.
\end{remark}

\paragraph{A priori estimates.}
Appendix~\ref{app:apriori} establishes uniform first and second moments for
$V^T$, followed by tightness in $L^2([0,1])$. These estimates support the
subsequential limit arguments below.

Recall the rescaled common-jump measure and its compensator from
\eqref{eq:rescaled-common-jump-measure}. From
\eqref{eq:common-jump-compensator} and \eqref{eq:VT-definition},
\begin{align}
\label{eq:rescaled-compensator}
  \bar\nu_J^T(du,dz)
 =2p_T\lambda_{Tu-}^TF_J(dz)Tdu
 =2\kappa_TV_u^TF_J(dz)du,
\end{align}
where
\[
  \kappa_T:=\frac{\bar p_Tm_\star(1-a_T)T^\alpha}{\theta}
  \to
  \kappa:=\frac{\bar p m_\star\mu_0}{\theta}.
\]
Set
\begin{equation}
\label{eq:nuJ}
  \nu_J(dz):=2\kappa F_J(dz).
\end{equation}
Since $F_J$ is a probability measure, the constructed limiting jump measure
$\nu_J(dz)=2\kappa F_J(dz)$ is finite. Moreover, $m_2<\infty$
implies $\int_{(0,\infty)} z^2\nu_J(dz)<\infty$. Thus, the limiting jump measure
produced by the microscopic construction satisfies the finiteness and
second-moment requirements of Assumption~\ref{ass:jump-integrability}. The
price-jump convention imposes $\Lambda\ge0$.

We next separate the raw marked-jump feedback into its predictable and martingale parts. The predictable part is responsible for the positive jump-feedback contribution to the effective Volterra drift.

\begin{lemma}[Predictable mean of the raw marked-jump feedback]
\label{lem:jump-mean}
The predictable mean of $\mathcal J^T$ is
\begin{equation}
\label{eq:JT-predictable}
  \mathcal J_t^{T,{\rm pred}}
  =\int_0^tK_T^{J,{\rm mean}}(t-u)V_u^Tdu,
\end{equation}
where
\begin{equation}
\label{eq:KJmean}
  K_T^{J,{\rm mean}}(u)
:=2p_Tc_Jr_TT\left(\int_{(0,\infty)} zF_J(dz)\right)H_T(Tu).
\end{equation}
Moreover
\begin{equation}
\label{eq:KJmean-limit}
  K_T^{J,{\rm mean}}(u)
  =\left(\xi\int_{(0,\infty)} z\nu_J(dz)+o(1)\right)T^{1-\alpha}H_T(Tu),
\end{equation}
where
\begin{equation}
\label{eq:xi}
  \xi:=\frac{\theta c_Jr_\star}{m_\star}.
\end{equation}
\end{lemma}

\begin{proof}
Using \eqref{eq:JT-raw} and the compensator \eqref{eq:common-jump-compensator},
then changing variables $s=Tu$, gives
\[
\mathcal J_t^{T,{\rm pred}}
=\mathfrak c_Tc_Jr_T\int_0^{tT}H_T(tT-s)
2p_T\lambda_s^T\left(\int_{(0,\infty)} zF_J(dz)\right)ds.
\]
Since $\mathfrak c_T\lambda_{Tu}^T=V_u^T$, this is \eqref{eq:KJmean}. Also,
\[
2p_Tc_Jr_TT\int_{(0,\infty)} zF_J(dz)
\sim
2c_Jr_\star\bar p\mu_0\left(\int_{(0,\infty)} zF_J(dz)\right)T^{1-\alpha}.
\]
On the other hand,
\[
\xi\int_{(0,\infty)} z\nu_J(dz)
=\frac{\theta c_Jr_\star}{m_\star}\cdot
2\frac{\bar p m_\star\mu_0}{\theta}\int_{(0,\infty)} zF_J(dz)
=2c_Jr_\star\bar p\mu_0\int_{(0,\infty)} zF_J(dz).
\]
This proves the asymptotic relation.
\end{proof}

Let
\[
  \widetilde M_J^T(du,dz)
  :=\bar\mu_J^T(du,dz)-\bar\nu_J^T(du,dz).
\]
Then
\begin{equation}
\label{eq:JT-decomposition}
  \mathcal J_t^T=\mathcal J_t^{T,{\rm pred}}+\bar{\mathcal J}_t^T,
\end{equation}
where
\begin{equation}
\label{eq:JT-compensated}
\bar{\mathcal J}_t^T
:=\int_{(0,t)}K_J^T(t-u)\int_{(0,\infty)} z\widetilde M_J^T(du,dz),
\qquad
K_J^T(u):=\mathfrak c_Tc_Jr_TH_T(Tu).
\end{equation}
By \eqref{eq:rT}, \eqref{eq:HT-identity}, and Lemma
\ref{lem:resolvent-perturbation},
as $T\to\infty$,
\begin{equation}
\label{eq:KJ-limit}
  K_J^T(u)\to\xi K(u),
  \qquad
  \left|K_J^T(u)\right|\le Cu^{\alpha-1},
  \qquad u\in(0,1],
\end{equation}
where the convergence is locally uniform in $u\in(0,1]$, and $C<\infty$
is independent of $T$ and $u$, for all sufficiently large $T$.

\subsection{Regular martingale and jump limits}
\label{sec:driver-limits}

Define
\[
  \widehat{\mathbf N}_t^T
  :=\sqrt{\frac{1-a_T}{m_\star T^\alpha}}
  \begin{pmatrix}
  M_{tT}^{T,+,{\rm reg}}\\
  M_{tT}^{T,-,{\rm reg}}
\end{pmatrix}.
\]

The continuous Brownian noises in the limiting model come from the centered
regular order-flow fluctuations. The following martingale functional central
limit theorem provides the tightness and bracket estimates for the common
filtration argument below.

\begin{proposition}[Regular martingale FCLT]
\label{prop:regular-fclt}
The jumps of $\widehat{\mathbf N}^T$ vanish uniformly, and the family is
$C$-tight in $D([0,1],\mathbb R^2)$. Its predictable brackets are
\[
  \left\langle\widehat N^{T,i},\widehat N^{T,j}\right\rangle_t
  =
  \delta_{ij}\omega_i\frac{1-p_T}{\theta}\int_0^tV_s^Tds,
  \qquad i,j\in\{+,-\}.
\]
Consequently, along every subsequence on which $V^T\Rightarrow V$,
\[
  \left\langle\widehat N^{T,i},\widehat N^{T,j}\right\rangle
  \Rightarrow
  \delta_{ij}\omega_i\theta^{-1}\int_0^\cdot V_sds
\]
in $C([0,1])$.
\end{proposition}

\begin{proof}
The jumps vanish because their sizes satisfy, as $T\to\infty$,
\[
  \sup_{t\le1}\left|\Delta \widehat N_t^{T,i}\right|
  =
  \sqrt{\frac{1-a_T}{m_\star T^\alpha}}
  \to0,
  \qquad i\in\{+,-\}.
\]
Thus, for every $\varepsilon>0$, the sum of the conditional second moments
of jumps larger than $\varepsilon$ is eventually zero; this is the Lindeberg
condition for the triangular martingale array.
Moreover,
\[
  \left\langle \widehat N^{T,\pm}\right\rangle_t
  =\omega_\pm\frac{1-p_T}{\theta}\int_0^tV_s^Tds,
  \qquad
  \left\langle \widehat N^{T,+},\widehat N^{T,-}\right\rangle_t=0.
\]
Since the map
$v\mapsto \left(\int_0^t v_sds\right)_{t\in[0,1]}$
is continuous from $L^2([0,1])$ to $C([0,1])$, we have, as $T\to\infty$
along the considered subsequence,
\[
  \left(\int_0^tV_s^Tds\right)_{t\in[0,1]}
  \Rightarrow
  \left(\int_0^tV_sds\right)_{t\in[0,1]}
\]
in $C([0,1])$. Hence, the diagonal brackets converge to
$\omega_\pm\theta^{-1}\int_0^\cdot V_sds$, while the cross brackets vanish
identically.
The bracket bounds and Rebolledo's criterion
\citep[Proposition~II.1]{rebolledo1980central}, equivalently Aldous'
criterion for this array, give tightness; the vanishing jumps make it
$C$-tight.
\end{proof}

Let
\[
  \mathbf v:=\frac{(1,\beta)^\top}
  {\sqrt{\omega_++\beta^2\omega_-}},
  \qquad
  \mathbf w_q:=\frac{(q,-1)^\top}
  {\sqrt{q^2\omega_++\omega_-}}
  =\frac{(q,-1)^\top}{\sqrt{2q}}.
\]
Define the normalized projections
\[
  B_t^T:=\sqrt\theta\mathbf v\cdot\widehat{\mathbf N}_t^T,
  \qquad
  W_t^T:=\sqrt\theta\mathbf w_q\cdot\widehat{\mathbf N}_t^T.
\]
The microscopic balance $q\omega_+=\omega_-$ gives the exact identity
\[
  qN_t^{T,+,{\rm reg}}-N_t^{T,-,{\rm reg}}
  =qM_t^{T,+,{\rm reg}}-M_t^{T,-,{\rm reg}}.
\]
Consequently, the continuous price component defined in \eqref{eq:PTc}
satisfies
\[
  P_t^{T,c}=W_t^T-\frac12\int_0^tV_s^Tds.
\]
Thus the exact balance makes the uncentered signed regular flow a local
martingale and yields $W^T$ directly from that flow.
Their predictable brackets satisfy
\begin{equation}
\label{eq:rho}
\begin{aligned}
\left\langle B^T\right\rangle_t
&=(1-p_T)\int_0^tV_s^Tds,\\
\left\langle W^T\right\rangle_t
&=(1-p_T)\int_0^tV_s^Tds,\\
\left\langle B^T,W^T\right\rangle_t
&=(1-p_T)\rho\int_0^tV_s^Tds,
\end{aligned}
\qquad
\rho
=\frac{q\omega_+-\beta\omega_-}
{\sqrt{(\omega_++\beta^2\omega_-)(q^2\omega_++\omega_-)}}
=-\frac{(\beta-1)\sqrt q}
{\sqrt{(1+q)(1+\beta^2q)}}.
\end{equation}
Since $\beta>1$ and $q>0$, the correlation is strictly negative.
Proposition~\ref{prop:coefficient-attainability} proves that varying
$(\beta,q)$ attains every value in $(-1,0)$.
For $q=1$, one recovers the symmetric benchmark
$\rho=(1-\beta)/\sqrt{2(1+\beta^2)}\in(-1/\sqrt2,0)$.

\begin{lemma}[Canonical predictable density]
\label{lem:canonical-predictable-density}
Let $V$ be a nonnegative $L^2([0,1])$ coordinate and set
\begin{equation}
\label{eq:canonical-left-density}
  A_t:=\int_0^tV_sds,
  \qquad
  V_t^p:=\liminf_{n\to\infty}
  n\left(A_t-A_{(t-1/n)^+}\right).
\end{equation}
Set $V_t^p=0$ where the extended-real liminf is infinite. In any usual
filtration to which $A$ is adapted, $V^p$ is predictable and satisfies
$V^p=V$ $dt$-almost everywhere; in particular, $dA_t=V_t^pdt$.
\end{lemma}

\begin{proof}
Each left difference quotient in \eqref{eq:canonical-left-density} is
continuous and adapted, hence predictable. Its pointwise $\liminf$ and the
finite modification above are predictable. Extend $V(\omega,\cdot)$ by zero
outside $[0,1]$. By the Lebesgue differentiation theorem
\cite[Theorem~3.21]{folland1999real}, for almost every $t$,
\[
  \lim_{h\downarrow0}\frac1h\int_{t-h}^tV_s\,ds=V_t.
\]
Taking $h=1/n$ gives $V^p=V$ $dt$-almost everywhere, pathwise.
\end{proof}

After the martingale FCLT, the regular noise still has to be passed through the singular Hawkes resolvent. The next proposition identifies the limiting Volterra stochastic convolution and the corresponding volatility-of-volatility constant.

\begin{proposition}[Regular Volterra martingale limit]
\label{prop:regular-volterra-limit}
Along every jointly convergent subsequence for which
\[
  \left(V^T,B^T\right)\Rightarrow(V,B^\infty),
  \qquad
  B_t^\infty=\int_0^t\sqrt{V_s^p}\,dB_s,
\]
\[
  \mathcal M^T
  \Rightarrow
  \left(
  \sqrt{c_B}\int_0^t K(t-s)\sqrt{V_s^p}dB_s
  \right)_{t\in[0,1]}
\]
in $L^2([0,1])$, jointly with the other convergent canonical coordinates,
where
\begin{equation}
\label{eq:cB}
  c_B
:=\frac{\theta\mu_0(\omega_++\beta^2\omega_-)}
  {m_\star(\omega_++\beta\omega_-)^2}.
\end{equation}
\end{proposition}

\begin{proof}
By the definition of $B^T$,
\[
  \bar M_{Tu}^T
  =\frac{\sqrt{\omega_++\beta^2\omega_-}}
  {\omega_++\beta\omega_-}
  \sqrt{\frac{m_\star T^\alpha}{1-a_T}}
  \frac{1}{\sqrt\theta}B_u^T.
\]
Therefore,
\[
\mathcal M_t^T
=\sqrt{c_{B,T}}\int_0^t T^{1-\alpha}H_T(T(t-u))dB_u^T,
\]
where
\[
  c_{B,T}
:=\frac{\theta(\omega_++\beta^2\omega_-)}
  {m_\star(\omega_++\beta\omega_-)^2}(1-a_T)T^\alpha
  \to c_B.
\]
The common-filtration local-martingale property and Brownian representation
of $B^\infty$ are supplied by
Proposition~\ref{prop:joint-tightness-identification}.
The kernel convergence and uniform envelope follow from Lemma
\ref{lem:resolvent-perturbation}. The singularity at the diagonal is handled by
Lemma~\ref{lem:singular-martingale-convolution}, applied jointly with the
remaining canonical coordinates.
\end{proof}

It remains to pass the marked common-jump forcing through the same singular resolvent. The next proposition gives the raw, uncompensated jump-feedback limit before the Riemann--Liouville rewriting.

\begin{proposition}[Raw jump-feedback limit]
\label{prop:raw-feedback-limit}
Along every jointly convergent subsequence for which
\[
  \left(V^T,\bar\mu_J^T\right)\Rightarrow(V,\mu)
\]
in $L^2([0,1])\times
\mathcal M_p([0,1]\times(0,\infty))$,
\[
  \mathcal J^T
  \Rightarrow
  \left(
  \xi\int_{[0,t)}K(t-s)\int_{(0,\infty)} z\mu(ds,dz)
  \right)_{t\in[0,1]}
\]
in $L^2([0,1])$.
\end{proposition}

\begin{proof}
For $R>0$, define the truncated feedback
\[
\mathcal J_t^{T,R}
:=
\mathfrak c_Tc_Jr_T
\int_{(0,tT)}H_T(tT-s)
\int_{(0,\infty)} z\mathbf 1_{\{z\le R\}}\mu_{J,{\rm com}}^T(ds,dz).
\]
The convergence for fixed $R$ and the uniform tail estimate for the difference $\mathcal J^T-\mathcal J^{T,R}$ are given in Appendix \ref{app:raw-jump}.
\end{proof}

\subsection{Fractional rewriting}
\label{sec:fractional-rewriting}

The preceding limits yield the variance equation in the Mittag--Leffler mild representation. The limiting Hawkes resolvent density $K$ satisfies the fractional resolvent identity
\begin{equation}
\label{eq:fractional-resolvent-identity}
  K
  =\frac{1}{\Gamma(1-\alpha)}k_\alpha
  -\frac{\mu_0}{\Gamma(1-\alpha)}k_\alpha*K,
  \qquad
  k_\alpha(t)=\frac{t^{\alpha-1}}{\Gamma(\alpha)}.
\end{equation}
To distinguish the raw Mittag--Leffler mild driver from the compensated
canonical driver $Z$ in \eqref{eq:variance-driver-Z}, define
\begin{equation}
\label{eq:raw-mild-driver}
  d\widetilde Z_t
  =\sqrt{c_B}\sqrt{V_t^p}dB_t
  +\xi\int_{(0,\infty)} z\mu(dt,dz).
\end{equation}
Using the predictable compensator $V_t^p\nu_J(dz)dt$, the raw driver has the
componentwise decomposition
\begin{align}
\label{eq:raw-driver-decomposition}
d\widetilde Z_t
=\sqrt{c_B}\sqrt{V_t^p}\,dB_t 
+\xi\int_{(0,\infty)}z
\left(\mu(dt,dz)-V_t^p\nu_J(dz)dt\right)
+\xi V_t^p
\left(\int_{(0,\infty)}z\nu_J(dz)\right)dt.
\end{align}
The mild limiting equation is $V=g_0+K*d\widetilde Z$ in $L^2([0,1])$.
The following lemma justifies the convolution interchange required to apply
the fractional resolvent identity to its Brownian, compensated-jump, and
predictable-jump components in \eqref{eq:raw-driver-decomposition}.

\begin{lemma}[Associativity of limiting stochastic convolutions]
\label{lem:stochastic-convolution-associativity}
Let $\alpha>1/2$, $K$ be the Mittag--Leffler kernel satisfying the
fractional resolvent identity above, and suppose that
Assumption~\ref{ass:jump-integrability} holds. Let $V\ge0$, let $B$ be a
Brownian motion, and let $\mu$ have predictable compensator
$V_t^p\nu_J(dz)dt$ in their common filtration. Assume further that
\[
  \E\int_0^1V_sds<\infty,
  \qquad
  \E\|V\|_{L^2([0,1])}^2<\infty.
\]
Then the identity
\[
  k_\alpha*\left(K*d\widetilde Z\right)=(k_\alpha*K)*d\widetilde Z
\]
holds on $[0,1]$, componentwise for the three terms in
\eqref{eq:raw-driver-decomposition}.
\end{lemma}

\begin{proof}
The bounds $k_\alpha(u),K(u)\le Cu^{\alpha-1}$ imply
\[
  (k_\alpha*K)(u)
  \le
  C\int_0^u(u-r)^{\alpha-1}r^{\alpha-1}dr
  =
  CB(\alpha,\alpha)u^{2\alpha-1}.
\]
For the Brownian component in \eqref{eq:raw-driver-decomposition}, It\^o's
isometry and Tonelli's theorem give
\[
\E\int_0^1\int_0^t
  |(k_\alpha*K)(t-s)|^2V_s\,ds\,dt
  \le
C\E\int_0^1V_s
\int_s^1(t-s)^{4\alpha-2}dt\,ds
<\infty,
\]
because $\E\int_0^1V_sds<\infty$ and $\alpha>1/2$. This is the
square-integrability condition for stochastic Fubini's theorem.

For the compensated-jump component in
\eqref{eq:raw-driver-decomposition}, the random-measure isometry gives the same
bound multiplied by
\[
  \int_{(0,\infty)}z^2\nu_J(dz)<\infty.
\]
Thus, stochastic Fubini's theorem also applies to that component.

For the predictable-jump component in
\eqref{eq:raw-driver-decomposition}, Young's inequality yields
\[
  \left\|
  \int_{(0,\infty)}z\nu_J(dz)\,(k_\alpha*K)*V
  \right\|_2
  \le
  \left(\int_{(0,\infty)}z\nu_J(dz)\right)
  \|k_\alpha*K\|_1\|V\|_2<\infty
\]
almost surely, and the assumed second-moment bound makes this estimate
integrable. Here the first mark moment is finite under the assumption
of the lemma. 
These three estimates verify the hypotheses of stochastic Fubini's theorem
for the two martingale components and deterministic Fubini's theorem for the
predictable component. This proves the associativity identity.
\end{proof}

We can now decompose the raw jump term into its compensated martingale part and predictable mean. Combining the latter with the regular fractional drift produces the effective Volterra drift coefficient.

\begin{corollary}[Compensated fractional rewriting]
\label{cor:compensated-rewriting}
Assume the hypotheses of Lemma~\ref{lem:stochastic-convolution-associativity}
and suppose that the raw mild equation $V=g_0+K*d\widetilde Z$ holds in
$L^2([0,1])$. Then
\[
\begin{aligned}
V_t
&=G_0(t)
+ b\int_0^t k_\alpha(t-s)V_sds
+ \sqrt{\bar c_B}\int_0^t k_\alpha(t-s)\sqrt{V_s^p}dB_s\\
&\quad
+ \bar\xi\int_{[0,t)}k_\alpha(t-s)
\int_{(0,\infty)} z\left(\mu(ds,dz)-V_s^p\nu_J(dz)ds\right),
\end{aligned}
\]
with effective parameters
\[
  \bar\xi=\frac{\xi}{\Gamma(1-\alpha)},
  \qquad
  \bar c_B=\frac{c_B}{\Gamma(1-\alpha)^2},
  \qquad
  b=\bar\xi\int_{(0,\infty)} z\nu_J(dz)-\frac{\mu_0}{\Gamma(1-\alpha)}.
\]
Within the microscopic subcritical regime, $b<0$.
\end{corollary}

\begin{proof}
Lemma~\ref{lem:stochastic-convolution-associativity} and the fractional
resolvent identity \eqref{eq:fractional-resolvent-identity} give
\begin{align*}
K*d\widetilde Z
&=\frac1{\Gamma(1-\alpha)}k_\alpha*d\widetilde Z
-\frac{\mu_0}{\Gamma(1-\alpha)}k_\alpha*(K*d\widetilde Z)\\
&=\frac1{\Gamma(1-\alpha)}k_\alpha*d\widetilde Z
-\frac{\mu_0}{\Gamma(1-\alpha)}k_\alpha*(V-g_0).
\end{align*}
Consequently,
\[
V=G_0-\frac{\mu_0}{\Gamma(1-\alpha)}k_\alpha*V
+\frac1{\Gamma(1-\alpha)}k_\alpha*d\widetilde Z,
\qquad
G_0=g_0+\frac{\mu_0}{\Gamma(1-\alpha)}k_\alpha*g_0.
\]
The compensator identity decomposes the raw jump term as
\begin{align*}
\bar\xi\int_{[0,t)}k_\alpha(t-s)\int_{(0,\infty)}z\mu(ds,dz)
&=\bar\xi\int_{[0,t)}k_\alpha(t-s)\int_{(0,\infty)}z
\left(\mu(ds,dz)-V_s^p\nu_J(dz)ds\right)\\
&\quad+\bar\xi\left(\int_{(0,\infty)}z\nu_J(dz)\right)
\int_0^tk_\alpha(t-s)V_s^pds.
\end{align*}
Since $V^p=V$ $dt$-almost everywhere, combining the last term with
$-\mu_0\Gamma(1-\alpha)^{-1}k_\alpha*V$ gives the stated coefficient $b$;
the Brownian and jump coefficients give $\bar c_B$ and $\bar\xi$.
\end{proof}

\section{Main convergence theorem and identification}
\label{sec:identification}

\subsection{Main convergence theorem}
\label{sec:main-convergence-theorem}

Before stating the main convergence result, we recall that the three rescaled
microscopic objects entering the theorem are the microscopic price process
$P^T$ defined in \eqref{eq:PTc}--\eqref{eq:PTJ}, the variance proxy $V^T$
defined in \eqref{eq:VT-definition}, and the rescaled common-jump measure
$\bar\mu_J^T$ defined in \eqref{eq:rescaled-common-jump-measure}.
The preceding renewal and driver limits identify the candidate macroscopic
dynamics and its coefficients. Combined with the joint noise identification
of Proposition~\ref{prop:joint-tightness-identification} and the complete weak
uniqueness of Proposition~\ref{prop:complete-weak-uniqueness}, they yield the
full-sequence convergence result below.

\begin{theorem}[Full microscopic convergence]
\label{thm:main}
Fix $\Lambda\ge0$. Assume
\[
  1-a_T\sim\mu_0T^{-\alpha},
  \qquad
  p_T=(1-a_T)^2\bar p_T,\quad \bar p_T\to\bar p>0,
  \qquad
  r_T\sim\frac{r_\star}{1-a_T},
\]
as $T\to\infty$, and further assume
\[
  \int_{(0,\infty)}zF_J(dz)<\infty,
  \qquad
  \int_{(0,\infty)}z^2F_J(dz)<\infty.
\]
Suppose that Assumptions
\ref{ass:kernel-tail-regularity},
\ref{ass:subcritical},
and \ref{ass:baseline} hold. Then the following holds as
$T\to\infty$.
\begin{enumerate}
\item The augmented complete system converges along the full sequence:
\[  \left(V^T,B^T,W^T,\bar\mu_J^T,P^T\right)
  \Rightarrow
  \left(\sigma^2,B^\infty,W^\infty,\mu,P\right)
\]
in
\[
  L^2([0,1])\times D([0,1],\R^2)\times
  \mathcal M_p([0,1]\times(0,\infty))\times D([0,1],\R).
\]
\item The limiting law is the unique complete canonical weak solution in the
sense of Definition~\ref{def:complete-canonical-solution}, with the coefficient
map specified in \eqref{eq:limiting-coefficient-map}. In particular, both
\[
(V^T,\bar\mu_J^T)\Rightarrow(\sigma^2,\mu)
\qquad\text{and}\qquad  \left(P^T,V^T,\bar\mu_J^T\right)
  \Rightarrow(P,\sigma^2,\mu)
\]
along the full sequence.
\end{enumerate}

The limiting constants are
\begin{equation}
\label{eq:limiting-coefficient-map}
\begin{gathered}
  \nu_J(dz)=2\frac{\bar p m_\star\mu_0}{\theta}F_J(dz),
  \qquad
  \xi=\frac{\theta c_Jr_\star}{m_\star},\\
  \bar\xi=\frac{\xi}{\Gamma(1-\alpha)},
  \qquad
  c_B=\frac{\theta\mu_0(\omega_++\beta^2\omega_-)}
  {m_\star(\omega_++\beta\omega_-)^2},
  \qquad
  \bar c_B=\frac{c_B}{\Gamma(1-\alpha)^2},\\
  b=\bar\xi\int_{(0,\infty)}z\nu_J(dz)-\frac{\mu_0}{\Gamma(1-\alpha)}<0,
  \qquad
  \rho=-\frac{(\beta-1)\sqrt q}
  {\sqrt{(1+q)(1+\beta^2q)}}.
\end{gathered}
\end{equation}
Here, the microscopic parameters $c_J>0$, $\beta>1$, $q\in(0,1]$,
$\theta>0$, and $m_\star>0$ are fixed constants of the marked Hawkes
order-flow model introduced in Section~\ref{sec:model}, and $\omega_\pm$ are
given by \eqref{eq:sign-activity-weights}. The induced forward
curve is
\[
  G_0(t)
  =
  g_0(t)
  +
  \frac{\mu_0}{\Gamma(1-\alpha)}
  \int_0^t
  \frac{(t-s)^{\alpha-1}}{\Gamma(\alpha)}
  g_0(s)ds.
\]
\end{theorem}

\begin{proof}
We provide the proof in Section~\ref{subsec:proof_of_main_theorem}.
\end{proof}

The coefficient map in Theorem~\ref{thm:main} gives necessary restrictions
on the reduced-form parameters. Conditional on baseline realizability, the
next proposition characterizes the corresponding attainable coefficient
region.

\begin{proposition}[Attainability of the structural coefficient region]
\label{prop:coefficient-attainability}
Fix
\[
  \alpha\in(1/2,1),\qquad b<0,\qquad
  \bar c_B>0,\qquad \bar\xi>0,\qquad \rho\in(-1,0),
\]
and let $\nu_J$ be a nonzero finite measure on $(0,\infty)$ satisfying
\[
  \int_{(0,\infty)}z^2\nu_J(dz)<\infty.
\]
Then there exist positive microscopic parameters
$\mu_0,\bar p,c_J,r_\star,\theta,m_\star$, together with
$\beta>1$, $q\in(0,1]$, and a probability law $F_J$, satisfying
Assumption~\ref{ass:subcritical}, such that the coefficient map
\eqref{eq:limiting-coefficient-map} equals the prescribed
$(b,\bar c_B,\bar\xi,\nu_J,\rho)$.

For $\mu>0$, define the set of realizable forward curves
\[
  \mathfrak G_\mu
  :=
  \left\{
  g_0+\frac{\mu}{\Gamma(1-\alpha)}k_\alpha*g_0:
  g_0\text{ is realizable under Assumption~\ref{ass:baseline}}
  \right\}.
\]
For a reduced-form tuple
$y=(b,\bar c_B,\bar\xi,\nu_J,\rho)$ satisfying the conditions above, set
\[
  \mu_0(y)
  :=
  \Gamma(1-\alpha)
  \left(
  \bar\xi\int_{(0,\infty)}z\nu_J(dz)-b
  \right)>0.
\]
Then the full tuple $(y,G_0)$ is attained if and only if
$G_0\in\mathfrak G_{\mu_0(y)}$.
\end{proposition}

\begin{proof}
Write
\[
  \Gamma_\alpha:=\Gamma(1-\alpha),\qquad
  M_1:=\int_{(0,\infty)}z\nu_J(dz),\qquad
  \xi:=\Gamma_\alpha\bar\xi,
  \qquad c_B:=\Gamma_\alpha^2\bar c_B.
\]
Set
\[
  \mu_0:=\xi M_1-\Gamma_\alpha b>0.
\]
Choose
\[
  \beta>\frac{1+|\rho|}{1-|\rho|}.
\]
The continuous function
\[
  q\longmapsto
  -\frac{(\beta-1)\sqrt q}{\sqrt{(1+q)(1+\beta^2q)}}
\]
tends to $0$ as $q\downarrow0$ and equals
$-(\beta-1)/(\beta+1)<\rho$ at $q=1/\beta$. Hence, the intermediate value
theorem gives a $q\in(0,1/\beta)$ realizing the prescribed $\rho$. Define
$\omega_\pm$ by \eqref{eq:sign-activity-weights}.

\[
  f_{\beta,q}
  :=\frac{\omega_++\beta^2\omega_-}
  {(\omega_++\beta\omega_-)^2}
  =\frac{(1+q)(1+\beta^2q)}{2(1+\beta q)^2},
  \qquad
  m_\star:=1,
  \qquad
  \theta:=\frac{c_B}{\mu_0f_{\beta,q}}.
\]
If $N_J:=\nu_J((0,\infty))$, set
\[
  \kappa:=\frac{N_J}{2},
  \qquad F_J:=\frac{\nu_J}{N_J},
  \qquad \bar p:=\frac{\kappa\theta}{\mu_0},
  \qquad c_J:=1,
  \qquad r_\star:=\frac{\xi}{\theta}.
\]
These choices give
\[
  2\frac{\bar p m_\star\mu_0}{\theta}F_J=\nu_J,
  \qquad
  \frac{\theta c_Jr_\star}{m_\star}=\xi,
  \qquad
  \frac{\theta\mu_0(\omega_++\beta^2\omega_-)}
  {m_\star(\omega_++\beta\omega_-)^2}=c_B.
\]
Moreover, since $\mu_0=\xi M_1-\Gamma_\alpha b$ and $b<0$,
\[
  \xi M_1<\mu_0.
\]
This is precisely Assumption~\ref{ass:subcritical} in limiting-coefficient
form, and the drift identity recovers the prescribed $b$. All remaining
claims follow directly from the construction.
\end{proof}

\begin{remark}[Non-identifiability of microscopic parameters]
\label{rem:micro-parameter-nonidentifiability}
The coefficient map in Proposition~\ref{prop:coefficient-attainability}
identifies $\theta$ and $m_\star$ through the ratio $\theta/m_\star$, and
$c_J$ and $r_\star$ through the product $c_Jr_\star$. The individual
microscopic parameters therefore remain non-identifiable.
\end{remark}

\begin{definition}[Complete canonical weak solution]
\label{def:complete-canonical-solution}
Consider the canonical coordinate vector
\[
  (V,B^\infty,W^\infty,\mu,P)
\]
on
\[
L^2_+([0,1])
\times D([0,1],\mathbb R^2)
\times\mathcal M_p([0,1]\times(0,\infty))
\times D([0,1],\mathbb R),
\]
and fix the convergence-determining family $(h_\ell)$ used in
Proposition~\ref{prop:joint-tightness-identification}. The common filtration
is the right-continuous completion generated up to time $t$ by
$A_r,B_r^\infty,W_r^\infty$ and
$\int_{(0,r]}h_\ell(z)\mu(ds,dz)$ for $r\le t$ and $\ell\ge1$.
Set $A_t:=\int_0^tV_sds$ and let $V^p$ be the canonical predictable density
defined in Lemma~\ref{lem:canonical-predictable-density}. Thus
$V^p=V$ $dt$-almost everywhere and $dA_t=V_t^pdt$.
A probability law on this space is a complete canonical weak solution if:
\begin{enumerate}
\item $B^\infty,W^\infty$ are continuous local martingales in the common
filtration and
\[
\langle B^\infty\rangle=A,\qquad
\langle W^\infty\rangle=A,\qquad
\langle B^\infty,W^\infty\rangle=\rho A;
\]
both continuous martingales are strongly orthogonal to the purely
discontinuous local martingales generated by
$\mu-V_t^p\nu_J(dz)dt$;
\item $\mu$ is simple in time and has predictable compensator
$V_t^p\nu_J(dz)dt$;
\item with
\[
Z_t
=bA_t+\sqrt{\bar c_B}\,B_t^\infty
+\bar\xi
\int_{(0,t]\times(0,\infty)}
z\left(\mu(ds,dz)-V_s^p\nu_J(dz)ds\right),
\]
the canonical relation
\[
(\mathrm{id}\times S_{\bar\xi})_\#\mu=\mathcal J(Z)
\]
holds;
\item $V$ satisfies, in $L^2([0,1])$,
\[
\begin{aligned}
V_t
&=G_0(t)+b\int_0^tk_\alpha(t-s)V_sds
+\sqrt{\bar c_B}\int_0^tk_\alpha(t-s)dB_s^\infty\\
&\quad+\bar\xi\int_{[0,t)}k_\alpha(t-s)
\int_{(0,\infty)}z
\left(\mu(ds,dz)-V_s^p\nu_J(dz)ds\right);
\end{aligned}
\]
\item the price coordinate is the adapted path functional
\[
\begin{aligned}
P_t
&=W_t^\infty-\frac12A_t
-\Lambda\int_{(0,t]\times(0,\infty)}
z\left(\mu(ds,dz)-V_s^p\nu_J(dz)ds\right)\\
&\quad-\int_0^t\int_{(0,\infty)}
\left(e^{-\Lambda z}-1+\Lambda z\right)
V_s^p\nu_J(dz)ds.
\end{aligned}
\]
\end{enumerate}
Weak uniqueness means uniqueness of the probability law on this complete
canonical space. Since $P$ is an adapted functional of
$(A,W^\infty,\mu)$, adjoining it does not enlarge the canonical filtration.
\end{definition}

\begin{proposition}[Weak uniqueness of complete canonical solutions]
\label{prop:complete-weak-uniqueness}
Under Assumptions~\ref{ass:jump-integrability} and
\ref{ass:wellposed}, suppose that
\[
  \bar\xi>0,\qquad \bar c_B>0,
  \qquad \rho\in(-1,1).
\]
Then complete canonical weak solutions in the sense of
Definition~\ref{def:complete-canonical-solution} are weakly unique.
\end{proposition}

\begin{proof}
We first identify the part of the complete system that drives the variance.
By Lemma~\ref{lem:recover-jump-measure}, the variance path determines the
c\`adl\`ag semimartingale driver and the common-jump measure through the
Borel maps
\[
  Z=\mathcal R(V),
  \qquad
  \mu
  =
  \left(\mathrm{id}\times S_{\bar\xi}^{-1}\right)_\#
  \mathcal J(\mathcal R(V)).
\]
By Assumption~\ref{ass:jump-integrability}, the compensated first-mark
integral in \eqref{eq:recover-B-from-V} is well-defined; as a path map it
is obtained from bounded mark truncations and is Borel.
The defining identity for $Z$ then recovers
$B^\infty$ as
\begin{equation}
\label{eq:recover-B-from-V}
  B_t^\infty
  =
  \frac{1}{\sqrt{\bar c_B}}
  \left[
  Z_t-bA_t
  -\bar\xi\int_{(0,t]\times(0,\infty)}z
  \left(\mu(ds,dz)-V_s^p\nu_J(dz)ds\right)
  \right].
\end{equation}
Thus $(V,B^\infty,\mu)$ is a Borel functional of $V$. Proposition
\ref{prop:wellposed} and Lemma~\ref{lem:recover-jump-measure} imply that its
joint law is unique. Denote this base coordinate by
\[
  \mathsf X:=(V,B^\infty,\mu).
\]

We identify the orthogonal continuous price coordinate by defining
\begin{equation}
\label{eq:orthogonal-price-coordinate}
  C_t^\infty
  :=
  \frac{W_t^\infty-\rho B_t^\infty}{\sqrt{1-\rho^2}}.
\end{equation}
Then $C^\infty$ is a continuous local martingale in the common filtration and
\[
  \langle C^\infty\rangle=A,
  \qquad
  \langle C^\infty,B^\infty\rangle=0.
\]
It is also strongly orthogonal to the compensated jump martingales generated
by $\mu-V_s^p\nu_J(dz)ds$.

We identify its conditional law by a localized change of measure. Fix a
bounded deterministic step function $h$ on $[0,1]$ and let
\[
  M_t:=\int_0^t h(s)dC_s^\infty,
  \qquad
  Q_t:=\langle M\rangle_t=\int_0^t (h(s))^2dA_s.
\]
For $n\ge1$, let
\[
  \tau_n:=\inf\{t\in[0,1]:Q_t\ge n\}\wedge1.
\]
For every $\lambda\in\R$, the stochastic exponential
\[
  D_t^{\lambda,n}
  :=
  \mathcal E\left(\lambda M^{\tau_n}\right)_t
  =
  \exp\left(
  \lambda M_{t\wedge\tau_n}
  -\frac{\lambda^2}{2}Q_{t\wedge\tau_n}
  \right)
\]
is a uniformly integrable martingale because its driving martingale has
bracket bounded by $n$, so Novikov's criterion applies. 
Let
$\Q^{\lambda,n}$ be the probability measure with
density $D_1^{\lambda,n}$ relative to $\Q$.

Set $L^{\lambda,n}:=\lambda M^{\tau_n}$, so that
$D^{\lambda,n}=\mathcal E(L^{\lambda,n})$. Girsanov's theorem
\citep[Theorem~III.3.24]{jacod2013limit} states that, under
$\Q^{\lambda,n}$, a $\Q$-local martingale $N$ is transformed into
\[
  N^{\lambda,n}
  :=N-\left\langle N,L^{\lambda,n}\right\rangle
\]
whenever the predictable covariation is defined. Since
\[
  \left\langle B^\infty,L^{\lambda,n}\right\rangle
  =
  \lambda\left\langle B^\infty,M^{\tau_n}\right\rangle
  =0,
\]
we have
\begin{equation}
\label{eq:girsanov-B-unchanged}
  \left(B^\infty\right)^{\lambda,n}=B^\infty,
  \qquad
  \left\langle (B^\infty)^{\lambda,n}\right\rangle=A.
\end{equation}
For every bounded $f\in C_c((0,\infty))$, define the compensated jump
martingale
\[
  J_t^f
  :=
  \int_{(0,t]\times(0,\infty)}
  f(z)\left(\mu(ds,dz)-V_s^p\nu_J(dz)ds\right).
\]
It is purely discontinuous, whereas $L^{\lambda,n}$ is continuous, and the
strong orthogonality in Definition~\ref{def:complete-canonical-solution}
therefore gives
\[
  \left\langle J^f,L^{\lambda,n}\right\rangle=0,
  \qquad
  \left(J^f\right)^{\lambda,n}=J^f.
\]
Equivalently, in the semimartingale characteristic formula the jump
likelihood ratio is $Y\equiv1$, and hence
\begin{equation}
\label{eq:girsanov-compensator-unchanged}
  \nu_{\Q^{\lambda,n}}^\mu(ds,dz)
  =V_s^p\nu_J(dz)ds.
\end{equation}
Equations \eqref{eq:girsanov-B-unchanged} and
\eqref{eq:girsanov-compensator-unchanged} show that the continuous bracket,
jump compensator, and continuous--jump orthogonality of the base system are
unchanged. All defining Volterra and canonical relations are pathwise
identities and are therefore unchanged as well. Consequently, under
$\Q^{\lambda,n}$ the base coordinate $\mathsf X$ is again a canonical
variance/jump solution with the same coefficients. Its weak uniqueness gives,
for every bounded Borel function $F$ of $\mathsf X$,
\begin{equation}
\label{eq:girsanov-base-invariance}
  \E\left[F(\mathsf X)D_1^{\lambda,n}\right]
  =\E\left[F(\mathsf X)\right].
\end{equation}
Because this holds for every bounded Borel $F$, it is equivalent to
\begin{equation}
\label{eq:conditional-density-one}
  \E\left[D_1^{\lambda,n}\mid\mathsf X\right]=1.
\end{equation}

Since $Q_{\tau_n}$ is a measurable function of the base coordinate,
\eqref{eq:conditional-density-one} implies
\[
  \E\left[
  \left.
  \exp\left(\lambda M_{\tau_n}
  -\frac{\lambda^2}{2}Q_{\tau_n}\right)
  \right|\mathsf X
  \right]=1.
\]
First intersect the full-measure sets on which this identity
holds over $\lambda\in\mathbb Q$. The conditional moment generating function
is finite on this dense set and is convex, so continuity on the interior of
its effective domain extends the identity to every $\lambda\in\mathbb R$ on
one common full-measure set. Hence, conditionally on $\mathsf X$,
$M_{\tau_n}$ is centered Gaussian with variance $Q_{\tau_n}$; equivalently,
for every $\vartheta\in\mathbb R$,
\[
  \E\left[e^{i\vartheta M_{\tau_n}}\Big|\mathsf X\right]
  =\exp\left(-\frac{\vartheta^2}{2}Q_{\tau_n}\right).
\]
Since
$Q_1<\infty$ almost surely, $\tau_n=1$ for all sufficiently large $n$ along
each path. Conditional dominated convergence therefore gives
\begin{equation}
\label{eq:conditional-price-noise-transform}
  \E\left[
  \left.
  \exp\left(i\int_0^1h(s)dC_s^\infty\right)
  \right|\mathsf X
  \right]
  =
  \exp\left(-\frac12\int_0^1(h(s))^2dA_s\right).
\end{equation}

Fix $m\ge1$, $t_1,\ldots,t_m\in[0,1]$, and
$\vartheta=(\vartheta_1,\ldots,\vartheta_m)\in\mathbb R^m$. Applying
\eqref{eq:conditional-price-noise-transform} with the deterministic step
function
\[
  h_\vartheta(s)
  :=
  \sum_{j=1}^m\vartheta_j\mathbf 1_{[0,t_j]}(s)
\]
gives
\begin{equation}
\label{eq:conditional-fdd-price-noise}
  \E\left[
  \left.
  \exp\left(i\sum_{j=1}^m\vartheta_jC_{t_j}^\infty\right)
  \right|\mathsf X
  \right]
  =
  \exp\left(
  -\frac12\int_0^1\left(h_\vartheta(s)\right)^2dA_s
  \right)
  =
  \exp\left(
  -\frac12\sum_{i,j=1}^m
  \vartheta_i\vartheta_jA_{t_i\wedge t_j}
  \right).
\end{equation}
This is the Cram\'er--Wold identification
\citep[Theorem~29.4]{billingsley1995probability} of the conditional
finite-dimensional distributions.
It follows that, conditionally on $\mathsf X$, $C^\infty$ is the centered
continuous Gaussian process with covariance
\[
\operatorname{Cov}\left(C_s^\infty,C_t^\infty\mid\mathsf X\right)
  =A_{s\wedge t}.
\]
Equivalently, its conditional law is that of $U_{A_t}$, where $U$ is a
standard Brownian motion independent of the base coordinate. This uniquely
determines the joint law of $(\mathsf X,C^\infty)$ on path space.

Finally, \eqref{eq:orthogonal-price-coordinate} determines
\[
  W^\infty=\rho B^\infty+\sqrt{1-\rho^2}\,C^\infty,
\]
and the price $P$ is the prescribed path functional of
$(A,W^\infty,\mu)$. Hence the law of
$(V,B^\infty,W^\infty,\mu,P)$ is unique.
\end{proof}

\subsection{Variance identification}
\label{sec:variance-identification}

Throughout this section,
$\mathcal M_p([0,1]\times(0,\infty))$ is the space of boundedly finite point
measures endowed with the vague topology. The space $D([0,1],\mathbb R)$
is endowed with the Skorokhod $J_1$ topology.

We now collect the tightness and noise-identification results needed to pass from subsequential limits to the limiting variance equation. The next proposition identifies all limiting noise sources jointly.

\begin{proposition}[Joint tightness and subsequential noise identification]
\label{prop:joint-tightness-identification}
The family
\[
  \left(V^T,B^T,W^T,\bar\mu_J^T\right)_{T\ge T_0}
\]
is tight in
\[
  L^2([0,1])\times D([0,1],\R^2)\times
  \mathcal M_p([0,1]\times(0,\infty)).
\]
For every joint limit point $(V,B^\infty,W^\infty,\mu)$, after modification
on a null set, the following statements hold in one common filtration. On
\[
\mathsf E
=
L^2_+([0,1])\times D([0,1],\mathbb R^2)
\times\mathcal M_p([0,1]\times(0,\infty))
\]
with coordinates $(v,b,w,m)$, set $A_t(v):=\int_0^tv_sds$ and define $v^p$
from $A(v)$ by Lemma~\ref{lem:canonical-predictable-density}.
Let $(h_\ell)_{\ell\ge1}\subset C_c((0,\infty))$ be countable and
convergence determining, and let $\mathcal G$ be the usual right-continuous
completion of
\[
\mathcal G_t^0
=
\sigma\left(
A_r,b_r,w_r,
\int_{(0,r]\times(0,\infty)}h_\ell(z)m(ds,dz):
r\le t,\ \ell\ge1
\right).
\]
Then $V^p$ is $\mathcal G$-predictable and is a version of the
Radon--Nikodym density $dA/dt$: $V^p=V$ $dt$-a.e. and
$dA_t=V_t^pdt$. Moreover, $B^\infty,W^\infty$ are continuous
$\mathcal G$-local martingales,
\[
\left\langle B^\infty\right\rangle_t=A_t,\qquad
\left\langle W^\infty\right\rangle_t=A_t,\qquad
\left\langle B^\infty,W^\infty\right\rangle_t=\rho A_t,
\]
and, for every $h\in C_c((0,\infty))$,
\[
M_t^h
:=
\int_{(0,t]\times(0,\infty)}h(z)\mu(ds,dz)
-\nu_J(h)A_t
\]
is a $\mathcal G$-local martingale. Consequently, the
$\mathcal G$-predictable compensator of $\mu$ is
$V_t^p\nu_J(dz)dt$. The measure $\mu$ is simple in time, and the continuous
martingales are orthogonal to compensated stochastic integrals against
$\mu-V_t^p\nu_J(dz)dt$.

On an independent product extension there are correlated Brownian motions
$B,W$, satisfying $d\langle B,W\rangle_t=\rho\,dt$, such that
\[
B_t^\infty=\int_0^t\sqrt{V_s^p}\,dB_s,
\qquad
W_t^\infty=\int_0^t\sqrt{V_s^p}\,dW_s.
\]
\end{proposition}

\begin{proof}
Tightness of $V^T$ is Proposition \ref{prop:tightness}. Tightness and
time simplicity of $\bar\mu_J^T$ are given by Proposition
\ref{prop:point-process-stability}. The bracket bounds in the proof of
Proposition \ref{prop:regular-fclt}, together with the vanishing jumps, give
tightness of $(B^T,W^T)$ in $D([0,1],\R^2)$. Product tightness gives joint
tightness.

Since each $V^T$ is nonnegative, every limit point $V$ in $L^2([0,1])$
satisfies $V\ge0$, $dt\otimes\Q$-a.e. The process $A$ is a continuous
$\mathcal G$-adapted canonical coordinate, so
Lemma~\ref{lem:canonical-predictable-density} gives the stated
predictability and density properties of $V^p$.

We next show that the limiting measure has no atoms at the time endpoints.
Let $C$ be a compact subset of $(0,\infty)$, and choose an open set
$O\subset(0,\infty)$ containing $C$ whose closure is compact in
$(0,\infty)$.
For $\delta\in(0,1)$, the compensator
formula, Cauchy--Schwarz inequality, and Proposition~\ref{prop:second-moment} give
\[
\begin{aligned}
\E\bar\mu_J^T([0,\delta]\times O)
&=2\kappa_TF_J(O)\E\int_0^\delta V_s^Tds
\le C_O\delta^{1/2},\\
\E\bar\mu_J^T([1-\delta,1]\times O)
&=2\kappa_TF_J(O)\E\int_{1-\delta}^1 V_s^Tds
\le C_O\delta^{1/2},
\end{aligned}
\]
uniformly in $T$. On a Skorokhod representation, the sets
$[0,\delta)\times O$ and $(1-\delta,1]\times O$ are relatively open and
relatively compact in $[0,1]\times(0,\infty)$. The Portmanteau theorem and
Fatou's lemma therefore transfer the same upper bounds to their limiting
counts. Since $\{0,1\}\times C$ is contained in their union, letting
$\delta\downarrow0$ gives $\E\mu(\{0,1\}\times C)=0$. A countable
exhaustion of the mark space $(0,\infty)$
by compact sets yields
\[
  \mu(\{0,1\}\times(0,\infty))=0
  \qquad\text{almost surely}.
\]
Thus the martingale identities below, written with $(0,t]$, identify the
compensator on the full time interval.

At the prelimit level,
\[
\begin{aligned}
\left\langle B^T\right\rangle_t&=(1-p_T)A_t\left(V^T\right),\\
\left\langle W^T\right\rangle_t&=(1-p_T)A_t\left(V^T\right),\\
\left\langle B^T,W^T\right\rangle_t&=(1-p_T)\rho A_t\left(V^T\right),
\end{aligned}
\]
and, for $h\in C_c((0,\infty))$,
\[
M_t^{T,h}
=
\int_{(0,t]}h(z)\bar\mu_J^T(ds,dz)
-2\kappa_TF_J(h)A_t\left(V^T\right)
\]
is a martingale. We pass all these identities in the same canonical
filtration. Fix $0\le s<t\le1$, times $r_i\le s$, indices $\ell_i$, and a
bounded continuous function $\varphi$. Set
\[
  \Phi^T
  :=\varphi\left(
  A_{r_i}(V^T),B_{r_i}^T,W_{r_i}^T,
  \int_{(0,r_i]\times(0,\infty)}h_{\ell_i}(z)
  \bar\mu_J^T(du,dz):1\le i\le m
  \right).
\]
These cylinder variables form an algebra generating the prelimit past
$\sigma$-field; the extension to all bounded past-measurable variables will
follow from a monotone-class argument. The martingale and bracket identities
give
\[
\E\left[\Phi^T\left(B_t^T-B_s^T\right)\right]=0,\qquad
\E\left[\Phi^T\left(W_t^T-W_s^T\right)\right]=0,
\]
\[
\E\left[\Phi^T\left(
(B_t^T)^2-(B_s^T)^2
-(1-p_T)\left(A_t(V^T)-A_s(V^T)\right)\right)\right]=0,
\]
with the analogous identity for $W^T$,
\[
\E\left[\Phi^T\left(
B_t^TW_t^T-B_s^TW_s^T
-(1-p_T)\rho\left(A_t(V^T)-A_s(V^T)\right)\right)\right]=0,
\]
and
\[
\E\left[\Phi^T\left(M_t^{T,h}-M_s^{T,h}\right)\right]=0.
\]

These identities pass directly to the limit because
Proposition~\ref{prop:second-moment} and the Cauchy--Schwarz inequality give
\[
  \sup_{T\ge T_0}\E\left[\left(A_1(V^T)\right)^2\right]
  \le
  \sup_{T\ge T_0}\E\|V^T\|_{L^2([0,1])}^2
  <\infty.
\]
Moreover, set
\[
  d_T:=\sup_{u\le1}\left(|\Delta B_u^T|\vee|\Delta W_u^T|\right).
\]
By Proposition~\ref{prop:regular-fclt}, $d_T$ is bounded by a deterministic
number tending to zero. The Burkholder--Davis--Gundy inequality and the
bounded-jump bracket estimate give, for either $N^T=B^T$ or $N^T=W^T$,
\[
  \E\left[\sup_{u\le1}\left|N_u^T\right|^4\right]
  \le C\E\left[\left(A_1(V^T)\right)^2+d_T^2A_1(V^T)\right].
\]
Consequently,
\[
  \sup_{T\ge T_0}\E\left[
    \sup_{u\le1}\left|B_u^T\right|^4+\sup_{u\le1}\left|W_u^T\right|^4
  \right]<\infty.
\]
Thus the linear, square, and cross-product martingale expressions stated
above are uniformly integrable. For the marked martingale, the random-measure
isometry implies that
\[
  \E\left[\left|M_1^{T,h}\right|^2\right]
  =2\kappa_TF_J(h^2)\E A_1(V^T)
  \le C_h,
\]
so these expressions are uniformly integrable as well.

Let $(K_j)_{j\ge1}$ be a compact exhaustion of $(0,\infty)$. The uniform
first-moment bound for the prelimit compensators implies
$\E[\mu([0,1]\times K_j)]<\infty$ for every $j$. We may therefore choose a
countable dense set $\mathbb T\subset[0,1]$, containing $0$ and $1$, such
that
\begin{equation}
\label{eq:no-fixed-atoms-dense-times}
  \Q\big(\mu(\{r\}\times K_j)>0\big)=0,
  \qquad r\in\mathbb T,\quad j\ge1.
\end{equation}
Indeed, for each $j$ the finite measure
$D\mapsto\E[\mu(D\times K_j)]$ has at most countably many atoms.

Fix $s<t$ in $\mathbb T$. It is enough first to consider past-cylinder
variables of the form
\begin{equation}
\label{eq:past-cylinder-variable}
  \Phi
  =
  \varphi\left(
  A_{r_i},B_{r_i}^\infty,W_{r_i}^\infty,
  \int_{(0,r_i]\times(0,\infty)}h_{\ell_i}(z)\mu(du,dz):
  1\le i\le m
  \right),
\end{equation}
where $r_i\in\mathbb T\cap[0,s]$ and $\varphi$ is bounded and continuous.
By \eqref{eq:no-fixed-atoms-dense-times}, these cylinder maps are almost
surely continuous at the limiting coordinate. Since
$v\mapsto A(v)$ is continuous from $L^2([0,1])$ to $C([0,1])$,
$p_T\to0$, $2\kappa_TF_J(h)\to\nu_J(h)$, and the relevant random variables
are uniformly integrable by the estimates above, the prelimit martingale
identities pass to the limit and give
\begin{align}
\label{eq:limit-martingale-identities-linear}
  \E\left[\Phi(B_t^\infty-B_s^\infty)\right]&=0,
  &
  \E\left[\Phi(W_t^\infty-W_s^\infty)\right]&=0,\\
\label{eq:limit-martingale-identities-square}
  \E\left[\Phi\left(
  (B_t^\infty)^2-(B_s^\infty)^2-(A_t-A_s)
  \right)\right]&=0,
  &
  \E\left[\Phi\left(
  (W_t^\infty)^2-(W_s^\infty)^2-(A_t-A_s)
  \right)\right]&=0,\\
\label{eq:limit-martingale-identities-cross}
  \E\left[\Phi\left(
  B_t^\infty W_t^\infty-B_s^\infty W_s^\infty
  -\rho(A_t-A_s)
  \right)\right]&=0,
\end{align}
and, for every $h\in C_c((0,\infty))$,
\begin{equation}
\label{eq:limit-marked-martingale-identity}
  \E\left[\Phi\left(M_t^h-M_s^h\right)\right]=0.
\end{equation}
Right continuity extends
\eqref{eq:limit-martingale-identities-linear}--
\eqref{eq:limit-marked-martingale-identity} from $\mathbb T$ to all
$0\le s<t\le1$.
A functional monotone-class argument then extends these identities from
the continuous past-cylinder variables $\Phi$ defined in
\eqref{eq:past-cylinder-variable} to every bounded
$\mathcal G_s$-measurable random variable. Equations
\eqref{eq:limit-martingale-identities-linear}--
\eqref{eq:limit-martingale-identities-cross} identify
$B^\infty,W^\infty$ as continuous local martingales with
\[
  \langle B^\infty\rangle=A,
  \qquad
  \langle W^\infty\rangle=A,
  \qquad
  \langle B^\infty,W^\infty\rangle=\rho A.
\]
Equation \eqref{eq:limit-marked-martingale-identity}, first for a
convergence-determining countable family and then for all nonnegative Borel
test functions by a monotone-class argument, identifies the predictable
compensator of $\mu$ as
\[
  \nu^\mu(dt,dz)=V_t^p\nu_J(dz)dt.
\]

Time simplicity follows from Proposition~\ref{prop:point-process-stability}.
Each $M^h$ is purely discontinuous, whereas $B^\infty,W^\infty$ are
continuous. Their quadratic, and hence predictable, covariations vanish,
which gives the asserted orthogonality.

Finally set
\[
C_\rho:=
\begin{pmatrix}1&\rho\\ \rho&1\end{pmatrix}.
\]
This matrix is positive definite. Write
$L=(B^\infty,W^\infty)^\top$, and on an independent product extension let
$\bar U$ be a two-dimensional Brownian motion. Define
\[
  U_t
  :=
  \int_0^t
  \mathbf 1_{\{V_s^p>0\}}
  (V_s^p)^{-1/2}C_\rho^{-1/2}dL_s
  +
  \int_0^t\mathbf 1_{\{V_s^p=0\}}d\bar U_s.
\]
The two integrals are orthogonal and the bracket identities above yield
\[
  \langle U\rangle_t
  =
  \int_0^t\mathbf 1_{\{V_s^p>0\}}I_2ds
  +
  \int_0^t\mathbf 1_{\{V_s^p=0\}}I_2ds
  =tI_2.
\]
L\'evy's characterization shows that $U$ is a two-dimensional Brownian
motion, and the definition gives
\[
  \begin{pmatrix}B_t^\infty\\W_t^\infty\end{pmatrix}
  =
  \int_0^t\sqrt{V_s^p}\,C_\rho^{1/2}dU_s.
\]
Thus $(B,W)^\top:=C_\rho^{1/2}U$ is a pair of standard Brownian motions with
$d\langle B,W\rangle_t=\rho dt$ and
\[
  B_t^\infty=\int_0^t\sqrt{V_s^p}dB_s,
  \qquad
  W_t^\infty=\int_0^t\sqrt{V_s^p}dW_s.
\]
Because the extension is an independent product, the already identified
compensator of $\mu$ is unchanged.
\end{proof}

Given the joint noise identification in the previous proposition, each subsequential limit can be tested against the limiting Volterra equation. The following theorem shows that every such limit satisfies the compensated rough Hawkes--Heston variance equation.

\begin{theorem}[Subsequential variance identification]
\label{thm:subseq}
Assume \eqref{eq:near-critical}, \eqref{eq:pT},
\eqref{eq:FJ-moments}, and \eqref{eq:rT}, together with
Assumptions~\ref{ass:kernel-tail-regularity}, \ref{ass:subcritical}, and
\ref{ass:baseline}. Let
\[
  \left(V^T,B^T,W^T,\bar\mu_J^T\right)
  \Rightarrow
  (V,B^\infty,W^\infty,\mu)
\]
along a joint convergent subsequence in
\[
  L^2([0,1])\times D([0,1],\mathbb R^2)\times
  \mathcal M_p([0,1]\times(0,\infty)).
\]
Suppose the joint limit is represented as in Proposition
\ref{prop:joint-tightness-identification}, namely, there exist standard Brownian motions
$B$ and $W$ with $d\langle W,B\rangle_t=\rho dt$ and 
\[
  B_t^\infty=\int_0^t\sqrt{V_s^p}dB_s,
  \qquad
  W_t^\infty=\int_0^t\sqrt{V_s^p}dW_s.
\]
Then $\mu$ is a time-simple point measure with common-filtration compensator
$V_t^p\nu_J(dz)dt$, and $V$ satisfies, in $L^2([0,1])$,
\begin{align}
\label{eq:identified-bondi}
V_t
&=G_0(t)
+b\int_0^t k_\alpha(t-s)V_sds
+\sqrt{\bar c_B}\int_0^t k_\alpha(t-s)\sqrt{V_s^p}dB_s\nonumber\\
&\quad
+\bar\xi\int_{[0,t)}k_\alpha(t-s)\int_{(0,\infty)} z
\left(\mu(ds,dz)-V_s^p\nu_J(dz)ds\right).
\end{align}
Define
\[
Z_t
=
b\int_0^tV_sds+\sqrt{\bar c_B}\,B_t^\infty
+
\bar\xi\int_{(0,t]\times(0,\infty)}
z\left(\mu(ds,dz)-V_s^p\nu_J(dz)ds\right).
\]
Then every subsequential variance/jump limit is canonical:
\[
  (\mathrm{id}\times S_{\bar\xi})_\#\mu
  =
  \mathcal J(Z),
  \qquad S_{\bar\xi}(z)=\bar\xi z.
\]
\end{theorem}

\begin{proof}
Along the considered subsequence, Assumption~\ref{ass:baseline},
Proposition~\ref{prop:regular-volterra-limit}, and
Proposition~\ref{prop:raw-feedback-limit} give, jointly with the canonical
coordinates,
\begin{align*}
  g_0^T&\longrightarrow g_0
  &&\text{in }L^2([0,1]),\\
  \mathcal M^T&\Rightarrow
  \mathcal M,\qquad
  \mathcal M_t
  :=\sqrt{c_B}\int_0^tK(t-s)\sqrt{V_s^p}dB_s
  &&\text{in }L^2([0,1]),\\
  \mathcal J^T&\Rightarrow
  \mathcal J,\qquad
  \mathcal J_t
  :=\xi\int_{[0,t)}K(t-s)\int_{(0,\infty)}z\mu(ds,dz)
  &&\text{in }L^2([0,1]).
\end{align*}
Since addition is continuous on $L^2([0,1])$, the exact prelimit identity
$V^T=g_0^T+\mathcal M^T+\mathcal J^T$ in \eqref{eq:VT-mild} passes to
\begin{equation}
\label{eq:identified-raw-mild}
  V_t
  =
  g_0(t)
  +\sqrt{c_B}\int_0^tK(t-s)\sqrt{V_s^p}dB_s
  +\xi\int_{[0,t)}K(t-s)\int_{(0,\infty)}z\mu(ds,dz)
\end{equation}
in $L^2([0,1])$.

The uniform first- and second-moment estimates in
Propositions~\ref{prop:first-moment} and \ref{prop:second-moment} pass to the
limit by lower semicontinuity. Proposition
\ref{prop:joint-tightness-identification} identifies the compensator of
$\mu$ as $V_s^p\nu_J(dz)ds$. Hence the hypotheses of
Lemma~\ref{lem:stochastic-convolution-associativity} and
Corollary~\ref{cor:compensated-rewriting} hold for
\eqref{eq:identified-raw-mild}; those results give
\eqref{eq:identified-bondi} with the coefficients in
\eqref{eq:limiting-coefficient-map}.

It remains to verify canonicality. Proposition
\ref{prop:joint-tightness-identification} places all terms in the same
filtration and Proposition~\ref{prop:point-process-stability} makes $\mu$
simple in time. The drift, $B^\infty$, and the compensator part of $Z$ have
continuous paths. Hence, at every atom $(t,z)$ of $\mu$,
\[
  \Delta Z_t=\bar\xi z.
\]
Conversely, these are all jumps of $Z$. Since $\bar\xi>0$ and the marks are
positive, the time-simple atoms correspond one-to-one.
More precisely, for every nonnegative
$f\in C_c([0,1]\times(0,\infty))$,
\begin{align*}
  \int_{[0,1]\times(0,\infty)}
  f(t,y)\mathcal J(Z)(dt,dy)
  &=
  \sum_{t\in[0,1]:\Delta Z_t\ne0}
  f(t,\Delta Z_t)\\
  &=
  \int_{[0,1]\times(0,\infty)}
  f(t,\bar\xi z)\mu(dt,dz)\\
  &=
  \int_{[0,1]\times(0,\infty)}
  f(t,y)
  \left((\mathrm{id}\times S_{\bar\xi})_\#\mu\right)(dt,dy).
\end{align*}
Thus $(\mathrm{id}\times S_{\bar\xi})_\#\mu=\mathcal J(Z)$.
Thus every subsequential variance/jump limit belongs to the
canonical solution class.
\end{proof}

The previous theorem identifies all subsequential variance/jump limits.
Because the microscopic coefficient $\bar\xi$ is strictly positive, state
uniqueness and recovery of the jump measure remove subsequences for that pair.

\begin{corollary}[Full variance and jump-measure convergence]
\label{cor:variance-convergence}
Under \eqref{eq:near-critical}, \eqref{eq:pT},
\eqref{eq:FJ-moments}, and \eqref{eq:rT}, together with
Assumptions~\ref{ass:kernel-tail-regularity}, \ref{ass:subcritical}, and
\ref{ass:baseline},
\[
  \left(V^T,\bar\mu_J^T\right)\Rightarrow(\sigma^2,\mu)
\]
in $L^2([0,1])\times
\mathcal M_p([0,1]\times(0,\infty))$, where
$(\sigma^2,\mu)$ is the reduced-form variance/jump pair solving the limiting
affine Volterra equation with compensator
\[
  (\sigma_t^2)^p\nu_J(dz)dt.
\]
Its joint law is determined by Proposition~\ref{prop:wellposed} and
Lemma~\ref{lem:recover-jump-measure}.
\end{corollary}

\begin{proof}
Set
\[
  \mathsf S
  :=L^2_+([0,1])\times
  \mathcal M_p([0,1]\times(0,\infty)),
  \qquad
  \mathsf P_T:=\mathcal L(V^T,\bar\mu_J^T).
\]
By Proposition~\ref{prop:joint-tightness-identification},
$(\mathsf P_T)$ is tight on $\mathsf S$. Let $(T_n)$ be arbitrary. There
exists a further subsequence, still denoted $(T_n)$, such that
\[
  (V^{T_n},\bar\mu_J^{T_n})\Rightarrow(V,\mu).
\]
By Theorem~\ref{thm:subseq}, $(V,\mu)$ is a canonical weak solution of
\eqref{eq:identified-bondi}. The construction \eqref{eq:nuJ}, the two moment
assumptions on $F_J$, and $\Lambda\ge0$ imply
Assumption~\ref{ass:jump-integrability}; Assumption~\ref{ass:baseline} and
Remark~\ref{rem:G0-monotonicity} imply Assumption~\ref{ass:wellposed}.
Consequently,
\[
  \mathcal L(V)=\mathcal L(\sigma^2)
\]
by Proposition~\ref{prop:wellposed}.
By Lemma~\ref{lem:recover-jump-measure}, there exists a Borel recovery map
$\mathfrak R_{\bar\xi}$ such that
\[
  \mu=\mathfrak R_{\bar\xi}(V)
  \qquad\text{almost surely}
\]
for every canonical solution. Hence
\[
  \mathcal L(V,\mu)
  =\bigl(\operatorname{id},\mathfrak R_{\bar\xi}\bigr)_\#
    \mathcal L(\sigma^2)
  =:\mathsf P_\star.
\]
Thus, every subsequential limit of $(\mathsf P_T)$ equals
$\mathsf P_\star$. Tightness and the subsequence principle therefore give
\[
  \mathsf P_T\Rightarrow\mathsf P_\star,
\]
which is the claimed full convergence.
\end{proof}

\subsection{Price convergence}
\label{sec:price-convergence}

The price process contains the same common-jump measure as the variance
equation. The following lemma shows that the compensated jump-price functional
is continuous along the convergent jump-measure and compensator sequence. We
first restrict the marks to $[\delta,R]$, where the relevant mapping is
continuous at point measures without boundary atoms, and then let
$\delta\downarrow0$ and $R\uparrow\infty$ using the first- and second-moment
tail bounds below.
In its application below, the limiting point measure is almost surely
simple in time by Proposition~\ref{prop:point-process-stability} and has no
atoms at the time endpoints by
Proposition~\ref{prop:joint-tightness-identification}.

We repeatedly use the following approximation principle. Let
$(S,d)$ be a metric space, and let $X^T,X^{T,n},X^n,X$ be $S$-valued random
elements. If
\begin{equation}
\label{eq:converging-together-hypotheses}
  X^{T,n}\Rightarrow X^n\quad\text{as }T\to\infty
  \quad\text{for every fixed }n,
  \qquad
  X^n\Rightarrow X\quad\text{as }n\to\infty,
\end{equation}
and, for every $\eta>0$,
\begin{equation}
\label{eq:converging-together-error}
  \lim_{n\to\infty}\limsup_{T\to\infty}
  \Q\left(d\left(X^{T,n},X^T\right)>\eta\right)=0,
\end{equation}
then $X^T\Rightarrow X$. This is the converging-together lemma
\citep[Theorem~3.2]{billingsley1999convergence}. For two cutoff parameters,
we apply it along a deterministic sequence of admissible cutoff pairs.

\begin{lemma}[Jump-price functional convergence]
\label{lem:jump-price-functional-convergence}
Assume $(V^T,\bar\mu_J^T)\Rightarrow(V,\mu)$ in
$L^2([0,1])\times\mathcal M_p([0,1]\times(0,\infty))$, and assume the compensators
satisfy, jointly with this convergence,
\[
  \bar\nu_J^T(ds,dz)\Rightarrow V_s^p\nu_J(dz)ds
\]
vaguely on compact sets. Assume also that the limiting point measure
$\mu$ is almost surely simple in time, that is,
$\mu(\{t\}\times(0,\infty))\le1$ for every $t$, and has no atoms at the
time endpoints:
\[
  \mu(\{0,1\}\times(0,\infty))=0
  \qquad\text{almost surely}.
\]
These two properties are guaranteed in the application by
Propositions~\ref{prop:point-process-stability} and
\ref{prop:joint-tightness-identification}, respectively.
Let $\nu_J=2\kappa F_J$ with $\kappa>0$, and suppose that
Assumption~\ref{ass:jump-integrability} holds. Then, as $T\to\infty$,
\[
  P^{T,J}\Rightarrow
  -\Lambda\int_0^\cdot\int_{(0,\infty)} z
  \left(\mu(ds,dz)-V_s^p\nu_J(dz)ds\right)
  -\int_0^\cdot\int_{(0,\infty)} (e^{-\Lambda z}-1+\Lambda z)V_s^p\nu_J(dz)ds
\]
in $D([0,1],\R)$.
\end{lemma}

\begin{proof}
Write
\[
  \nu(ds,dz):=V_s^p\nu_J(dz)ds,
  \qquad
  g_\Lambda(z):=e^{-\Lambda z}-1+\Lambda z.
\]
For $0<\delta<R<\infty$, define
\[
  f_{1,\delta,R}(z):=z\mathbf 1_{[\delta,R]}(z),
  \qquad
  f_{2,\delta,R}(z):=g_\Lambda(z)\mathbf 1_{[\delta,R]}(z),
\]
and the truncated price functionals
\begin{align*}
  P_t^{T,J;\delta,R}
  &:=-\Lambda\int_{(0,t]}f_{1,\delta,R}(z)
  \left(\bar\mu_J^T-\bar\nu_J^T\right)(ds,dz)
  -\int_{(0,t]}f_{2,\delta,R}(z)\bar\nu_J^T(ds,dz),\\
  P_t^{J;\delta,R}
  &:=-\Lambda\int_{(0,t]}f_{1,\delta,R}(z)
  (\mu-\nu)(ds,dz)
  -\int_{(0,t]}f_{2,\delta,R}(z)\nu(ds,dz).
\end{align*}

We first prove convergence for admissible cutoff levels. For almost every
$z>0$,
\[
  \Q\big(\mu([0,1]\times\{z\})>0\big)=0;
\]
this follows from Tonelli's theorem because the set of mark coordinates of a
boundedly finite point measure is pathwise countable. Excluding also the at
most countable set of atoms of $\nu_J$, we may choose deterministic
$0<\delta<R$ such that
\begin{equation}
\label{eq:admissible-price-cutoffs}
  \mu([0,1]\times\{\delta,R\})=0,
  \qquad
  \nu_J(\{\delta,R\})=0
  \quad\text{almost surely}.
\end{equation}

Use a Skorokhod representation along the considered subsequence, jointly for
$\bar\mu_J^T$ and $\bar\nu_J^T$, so that both vague convergences hold almost
surely. All limits in the remainder of this cutoff argument are understood
almost surely as $T\to\infty$ along the considered subsequence. On
$[0,1]\times[\delta,R]$ write
\[
  \mu=\sum_{k=1}^N\delta_{(t_k,z_k)},
  \qquad 0<t_1<\cdots<t_N<1.
\]
The finiteness of $N$, the strict ordering of the times, and the absence of
endpoint atoms follow from the hypotheses. Vague convergence on the
continuity set in \eqref{eq:admissible-price-cutoffs} gives, for all large
$T$, a labeling
\[
  \bar\mu_J^T\big|_{[0,1]\times[\delta,R]}
  =\sum_{k=1}^N\delta_{(t_k^T,z_k^T)},
  \qquad
  \left(t_k^T,z_k^T\right)\longrightarrow(t_k,z_k).
\]
Let $\lambda_T$ be the increasing piecewise-linear bijection of $[0,1]$
that maps $t_k$ to $t_k^T$ and fixes $0$ and $1$. Then,
\[
  \|\lambda_T-\mathrm{id}\|_\infty\longrightarrow0,
\]
and, with
\[
  X_T(t):=\int_{(0,t]}f_{1,\delta,R}(z)\bar\mu_J^T(ds,dz),
  \qquad
  X(t):=\int_{(0,t]}f_{1,\delta,R}(z)\mu(ds,dz),
\]
we have
\[
  \|X_T\circ\lambda_T-X\|_\infty
  \le
\sum_{k=1}^N\left|f_{1,\delta,R}\left(z_k^T\right)-f_{1,\delta,R}(z_k)\right|
  \longrightarrow0.
\]
Thus $X_T\to X$ in the $J_1$ topology.

For $j=1,2$, set
\[
  H_T^j(t):=\int_{(0,t]}f_{j,\delta,R}(z)\bar\nu_J^T(ds,dz),
  \qquad
  H^j(t):=\int_{(0,t]}f_{j,\delta,R}(z)\nu(ds,dz).
\]
The limiting measure $\nu$ is atomless in time. Vague convergence of the
finite weighted measures and continuity of $H^j$ therefore imply
\[
  \max_{j=1,2}\left\|H_T^j-H^j\right\|_\infty\longrightarrow0.
\]
Combining this uniform convergence with the $J_1$ convergence of $X_T$ gives
\begin{equation}
\label{eq:truncated-jump-price-convergence}
  P^{T,J;\delta,R}\Rightarrow P^{J;\delta,R}
  \qquad\text{in }(D([0,1],\mathbb R),J_1).
\end{equation}
Time simplicity is essential here: it prevents distinct prelimit jumps from
merging into one limiting time, which would in general destroy $J_1$
continuity.

For the martingale term in the cutoff errors, we use Doob's $L^2$ inequality
and the random-measure isometry. Thus the error caused by excluding marks in $(0,\delta)$ vanishes as
$\delta\downarrow0$, by the estimates:
\[
  \E\left[
  \sup_{t\le1}\left|
  \int_0^t\int_{(0,\delta)}z
  \left(\bar\mu_J^T-\bar\nu_J^T\right)(ds,dz)
  \right|^2
  \right]
  \le
  C\int_{(0,\delta)}z^2F_J(dz)\sup_{T\ge T_0}\E\int_0^1V_s^Tds,
\]
and
\[
  \E\int_0^1\int_{(0,\delta)}
  \left|e^{-\Lambda z}-1+\Lambda z\right|\bar\nu_J^T(ds,dz)
  \le
  C\int_{(0,\delta)}z^2F_J(dz)\sup_{T\ge T_0}\E\int_0^1V_s^Tds.
\]
The upper tail is controlled similarly by
\[
C\int_{(R,\infty)}z^2F_J(dz)\sup_{T\ge T_0}\E\int_0^1V_s^Tds
+
    C\int_{(R,\infty)}\left|e^{-\Lambda z}-1+\Lambda z\right|F_J(dz)
    \sup_{T\ge T_0}\E\int_0^1V_s^Tds.
\]
For the lower compensator correction we used
$|e^{-\Lambda z}-1+\Lambda z|\le C_\Lambda z^2$ on bounded neighborhoods of
zero, for fixed $\Lambda$. The first upper-tail term vanishes by
$m_2<\infty$, equivalently by the constructed second moment
$\int_{(0,\infty)} z^2\nu_J(dz)<\infty$. For the second upper-tail term,
$\nu_J=2\kappa F_J$ with $\kappa>0$, and $\kappa_T$ is bounded, so the
$F_J$-tail is controlled by the corresponding $\nu_J$-tail in Assumption
\ref{ass:jump-integrability}. Hence all tails vanish as
$\delta\downarrow0$ and $R\to\infty$, using also Proposition
\ref{prop:first-moment}. More explicitly, for every $\varepsilon>0$,
\begin{equation}
\label{eq:prelimit-jump-price-converging-together}
  \lim_{\substack{\delta\downarrow0\\R\uparrow\infty}}
  \sup_{T\ge T_0}
  \Q\left(
  \left\|P^{T,J}-P^{T,J;\delta,R}\right\|_\infty>\varepsilon
  \right)=0.
\end{equation}
The limiting estimates are identical with $F_J$ replaced by $\nu_J$, and
give
\begin{equation}
\label{eq:limit-jump-price-converging-together}
  \lim_{\substack{\delta\downarrow0\\R\uparrow\infty}}
  \Q\left(
  \left\|P^J-P^{J;\delta,R}\right\|_\infty>\varepsilon
  \right)=0,
\end{equation}
where
\[
  P_t^J
  =
  -\Lambda\int_{(0,t]}z(\mu-\nu)(ds,dz)
  -\int_{(0,t]}g_\Lambda(z)\nu(ds,dz).
\]
The fixed-cutoff convergence
\eqref{eq:truncated-jump-price-convergence}, together with
\eqref{eq:prelimit-jump-price-converging-together} and
\eqref{eq:limit-jump-price-converging-together}, yields
$P^{T,J}\Rightarrow P^J$ by the converging-together lemma: take any
admissible sequence $(\delta_n,R_n)$ satisfying
\eqref{eq:admissible-price-cutoffs}, with
$\delta_n\downarrow0$ and $R_n\uparrow\infty$, and use
$X^{T,n}=P^{T,J;\delta_n,R_n}$ and
$X^n=P^{J;\delta_n,R_n}$ in
\eqref{eq:converging-together-hypotheses}--
\eqref{eq:converging-together-error}.
\end{proof}

\subsection{Proof of Theorem~\ref{thm:main}}
\label{subsec:proof_of_main_theorem}

\begin{proof}[Proof of Theorem~\ref{thm:main}]
Let
\[
  \Pi_T
  :=
  \operatorname{Law}\left(V^T,B^T,W^T,\bar\mu_J^T,P^T\right)
\]
on the product space in the theorem. Proposition
\ref{prop:joint-tightness-identification} gives tightness of the first four
coordinates.
Define the cumulative-integral map
$A:L^2([0,1])\to C([0,1])$ by
\[
  [A(v)]_t=:A_t(v):=\int_0^tv_sds.
\]
Thus $A(v)$ denotes the entire cumulative-integral path, whereas
$A_t(v)$ denotes its value at time $t$. With this notation,
\[
  P^{T,c}=W^T-\frac12A(V^T),
\]
and
\begin{equation}
\label{eq:A-map-continuity-main-proof}
  \|A(v)-A(w)\|_\infty
  \le \|v-w\|_{L^1([0,1])}
  \le \|v-w\|_{L^2([0,1])}.
\end{equation}
The vanishing-jump bracket estimates for $W^T$ make it $C$-tight, while
\eqref{eq:A-map-continuity-main-proof} makes $A(V^T)$ tight in
$C([0,1])$. Thus $P^{T,c}$ is $C$-tight. For fixed admissible
$0<\delta<R$, the truncated jump-price functional is tight by the atom and
compensator mapping argument leading to
\eqref{eq:truncated-jump-price-convergence}. The uniform estimate
\eqref{eq:prelimit-jump-price-converging-together} then gives tightness of
$P^{T,J}$. Product tightness gives tightness of
$(P^{T,c},P^{T,J})$, and addition in the $J_1$ topology is continuous at
pairs whose first coordinate is continuous. Since
$P^T=P^{T,c}+P^{T,J}$, the family $(\Pi_T)$ is tight.

Now let $(T_n)$ be an arbitrary sequence tending to infinity. By tightness,
there is a subsequence, still denoted $(T_n)$, and a probability law $\Pi$
such that
\begin{equation}
\label{eq:arbitrary-augmented-limit}
  \left(V^{T_n},B^{T_n},W^{T_n},\bar\mu_J^{T_n},P^{T_n}\right)
  \Rightarrow
  \left(V,B^\infty,W^\infty,\mu,P\right)
  \quad\text{under }\Pi.
\end{equation}
Write $A:=A(V)$, so that
$A_t=[A(V)]_t=\int_0^tV_sds$.
Proposition~\ref{prop:joint-tightness-identification} identifies, in one
common filtration under $\Pi$,
\begin{equation}
\label{eq:augmented-limit-characteristics}
  \langle B^\infty\rangle=A,
  \qquad
  \langle W^\infty\rangle=A,
  \qquad
  \langle B^\infty,W^\infty\rangle=\rho A,
  \qquad
  \nu^\mu(dt,dz)=V_t^p\nu_J(dz)dt,
\end{equation}
together with continuous--jump orthogonality and absence of atoms at the time
endpoints. Proposition~\ref{prop:point-process-stability} gives time
simplicity. Theorem~\ref{thm:subseq} then gives the variance equation
\eqref{eq:identified-bondi} and the canonical relation
\begin{equation}
\label{eq:augmented-limit-canonical-relation}
  (\mathrm{id}\times S_{\bar\xi})_\#\mu=\mathcal J(Z).
\end{equation}

Corollary~\ref{cor:variance-convergence} identifies the variance/jump pair,
and Lemma~\ref{lem:limiting-compensator-identification} supplies the joint
vague convergence of the compensators required in
Lemma~\ref{lem:jump-price-functional-convergence}. Applying that lemma along
\eqref{eq:arbitrary-augmented-limit} identifies the last coordinate as
\[
\begin{aligned}
P_t
&=W_t^\infty-\frac12A_t
-\Lambda\int_{(0,t]\times(0,\infty)}
z\left(\mu(ds,dz)-V_s^p\nu_J(dz)ds\right)\\
&\quad-\int_0^t\int_{(0,\infty)}
\left(e^{-\Lambda z}-1+\Lambda z\right)
V_s^p\nu_J(dz)ds,
\end{aligned}
\]
which is an adapted functional of $(A,W^\infty,\mu)$. Equations
\eqref{eq:augmented-limit-characteristics},
\eqref{eq:identified-bondi}, and
\eqref{eq:augmented-limit-canonical-relation}, together with this price
identity, show that $\Pi$ is the law of a complete canonical weak solution in
the sense of Definition~\ref{def:complete-canonical-solution}.

The constructed coefficients satisfy $\bar c_B>0$, $\bar\xi>0$, and
$\rho\in(-1,0)$. Let $\Pi_\star$ denote the unique complete canonical law
given by Proposition~\ref{prop:complete-weak-uniqueness}. The preceding
argument shows that every subsequential limit $\Pi$ equals $\Pi_\star$:
\[
  \Pi=\Pi_\star.
\]
Because the initial sequence $(T_n)$ was arbitrary, the subsequence principle
gives
\[
  \Pi_T\Rightarrow\Pi_\star
  \qquad\text{as }T\to\infty.
\]
Projection onto $(V^T,\bar\mu_J^T)$ and
$(P^T,V^T,\bar\mu_J^T)$ gives the two remaining assertions of the theorem.
\end{proof}

\section{Numerical illustrations}
\label{sec:numerical-illustrations}

We first illustrate the variance convergence in
Corollary~\ref{cor:variance-convergence}, for the microscopic variance proxy
$V^T$ defined in \eqref{eq:VT-definition} and the limiting variance
$\sigma^2$ solving \eqref{eq:limit-bondi-form}, using the canonical kernel
\eqref{eq:canonical-kernel}.
We set
\[
 \alpha=0.65,\quad \mu_0=\bar p=m_\star=1,\quad
 \theta=0.2,\quad c_J=0.4,\quad r_\star=1.25,
\]
and use the exact finite-$T$ specifications
\[
  a_T=1-\mu_0T^{-\alpha},\qquad
  \bar p_T\equiv\bar p,\qquad
  p_T=(1-a_T)^2\bar p,\qquad
  r_T=\frac{r_\star}{1-a_T},
\]
which instantiate \eqref{eq:near-critical}, \eqref{eq:pT}, and
\eqref{eq:rT}. We take $F_J=\operatorname{Exp}(2)$, so that $m_1=1/2$ in
\eqref{eq:FJ-moments}, set $q=0.5$ and $\beta=3$ in
\eqref{eq:sign-activity-weights} and \eqref{eq:rho}, and use
$(\theta,m_\star)$ in the normalization \eqref{eq:VT-definition} and $c_J$
in the marked intensity \eqref{eq:lambda-marked}. The raw limiting baseline
in Assumption~\ref{ass:baseline} is set to the constant $g_0=0.04$.
At each integer $T=2,\ldots,50$, the exact
immigration--branching representation generates $100{,}000$ independent
marked Hawkes paths.  The value $T=1$ is outside the finite-scale parameter
region because it gives $a_1=0$ and $p_1=1$.  Both $V_1^T$ and
$\int_0^1V_t^Tdt$ are evaluated analytically from the simulated event times.
The limiting laws are recovered deterministically from their affine
Riccati--Volterra transforms using 40 positive exponential kernel factors,
16,384 transform time steps, and 2,048 COS terms.

Figure~\ref{fig:marked-hawkes-convergence} reports the empirical
Wasserstein--1 distances
\[
  \mathcal{W}_1\!\left(\mathcal L(V_1^T),\mathcal L(\sigma_1^2)\right)
  \quad\text{and}\quad
  \mathcal{W}_1\!\left(
    \mathcal L\!\left(\int_0^1V_t^Tdt\right),
    \mathcal L\!\left(\int_0^1\sigma_t^2dt\right)
  \right),
\]
where $V^T$ is defined in \eqref{eq:VT-definition} and $\sigma^2$ solves
\eqref{eq:limit-bondi-form}.
The shaded bands are the
5--95\% ranges obtained by repeatedly comparing two independent samples of
size $100{,}000$ drawn from the limiting law.  From $T=2$ to $T=50$, the
terminal distance decreases from $2.37\times10^{-2}$ to $3.09\times10^{-3}$,
while the integrated-variance distance decreases from $1.20\times10^{-2}$ to
$7.63\times10^{-4}$.  The larger sample makes the decreasing trend and the
large-$T$ plateau substantially smoother.  The plateau remains above the
narrower Monte Carlo resolution bands, especially for the integrated
variance, and therefore records a finite-$T$ pre-asymptotic discrepancy
together with the numerical approximation of the affine reference law.
The right panel uses a continuous functional on $L^2([0,1])$, while the left
panel provides the complementary endpoint-marginal comparison.  The dense
grid therefore displays the distributional approximation over the
pre-asymptotic range without fitting a rate below numerical resolution.

\begin{figure}[H]
  \centering
  \includegraphics[width=\textwidth]{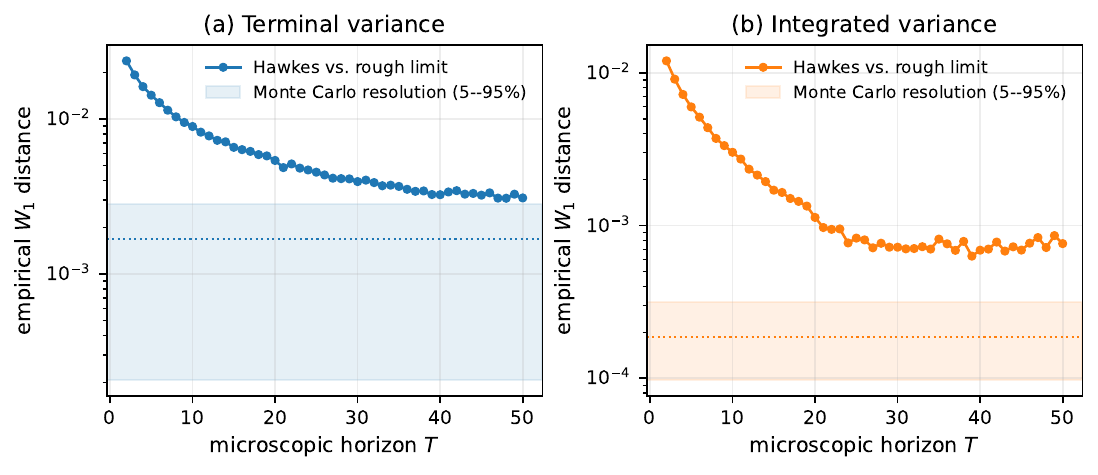}
  \caption{Distributional approximation of the finite-$T$ marked Hawkes
  variance proxy by its rough Volterra limit.  The left panel compares the
  terminal marginals $V_1^T$ and $\sigma_1^2$; the right panel compares the
  integrated variances.  Each point uses $100{,}000$ exact-cluster Hawkes
  simulations; the limiting laws are obtained by deterministic affine-transform
  inversion.  Shading shows the 5--95\% Monte Carlo resolution range for two
  samples of size $100{,}000$ from the limiting law.}
  \label{fig:marked-hawkes-convergence}
\end{figure}

We next compare option prices generated by the finite-$T$ microscopic model
with those of its rough Hawkes--Heston limit.  We adopt the exponential jump
law and the parameter vector of \citet{bondi2024}:
\[
 \alpha=0.527,\quad \rho=-0.731,\quad b=-1.812,\quad
 \bar c_B=0.115,\quad \Lambda=0.276,
 \quad G_0(t)=0.0079+\frac{0.049t^\alpha}{\Gamma(\alpha+1)}.
\]
These are the reduced-form coefficients of \eqref{eq:limit-bondi-form}, with
$k_\alpha$ given by \eqref{k:alpha}. The corresponding microscopic
parameters are obtained by inverting the coefficient map
\eqref{eq:limiting-coefficient-map}.
Bondi et al. normalize the variance-jump loading to one.  Taking
$\bar\xi=1$, $F_J=\operatorname{Exp}(1)$, and applying
Proposition~\ref{prop:coefficient-attainability} with $\beta=8$ gives the
following concrete microscopic instance:
\[
 (\mu_0,\bar p,c_J,r_\star,\theta,m_\star,\beta,q,F_J)
 = (5.2653,0.006772,1,26.2566,0.07131,1,8,0.041605,\operatorname{Exp}(1)).
\]
These values recover $\rho=-0.731$, $b=-1.812$, and $\bar c_B=0.115$.  Since
$\rho(\beta,1)\in(-1/\sqrt{2},0)$, the target $\rho=-0.731<-1/\sqrt{2}$
rules out $q=1$ in \eqref{eq:rho}.  Moreover,
\[
 \max_{q\in(0,1]}|\rho(\beta,q)|=\frac{\beta-1}{\beta+1},
 \qquad\text{attained at }q=1/\beta,
\]
so any realization requires $\beta\ge(1+0.731)/(1-0.731)=6.435$;
our admissible choice $\beta=8$ gives $q=0.041605$.  The corresponding
weights, defined in \eqref{eq:sign-activity-weights}, are
$\omega_+=1.9201$ and $\omega_-=0.07989$: positive regular
events are more frequent but have the smaller tick impact, while negative
regular events carry the larger variance loading.  The value $\beta=8$ is one
admissible realization; the lower bound is the implication identified by
the reduced-form correlation.  The microscopic
baseline is chosen by inverting the fractional baseline map so that its limit
is the displayed curve $G_0$.  At scaling horizon $T=1000$, we simulate
$100,000$ exact immigration--branching paths of the microscopic price
$P^T$.  OTM put and call prices are sample averages of payoffs of
$S^T=\exp(P^T)$, with strikes measured relative to the corresponding
Monte Carlo forward.  The affine limit uses 32 kernel factors and the Fourier
inversion of \citet{bondi2024}.

Figure~\ref{fig:bondi-smiles} compares the two SPX smiles at maturities
$t=23/730$ and $t=66/730$.  The microscopic prices reproduce the deep-put
negative skew, the at-the-money minimum, and the short right-wing upturn of
the affine limit.  The error bars are 95\% Monte Carlo intervals obtained by
propagating the payoff standard errors through the implied-volatility
inversion; they are largest for the deepest puts at the shorter maturity.

\begin{figure}[H]
  \centering
  \includegraphics[width=\textwidth]{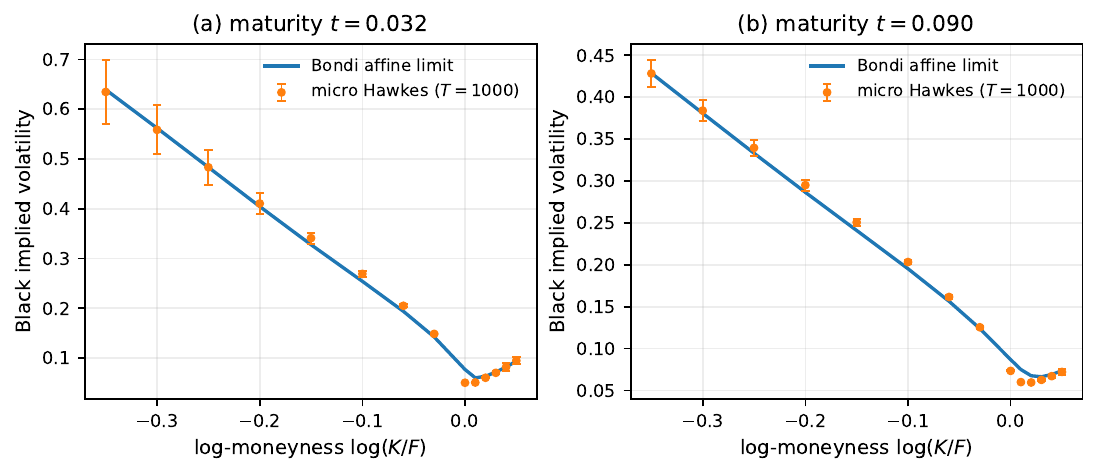}
  \caption{Finite-scale microscopic and affine-limit SPX implied-volatility
  smiles under the parameterization of \citet{bondi2024}.  The orange points
  use $100{,}000$ exact-cluster paths at microscopic scaling horizon $T=1000$;
  error bars show 95\% Monte Carlo intervals.  The blue curves are computed
  from the rough Hawkes--Heston affine transform.  Log-moneyness is measured
  relative to each model's forward price.}
  \label{fig:bondi-smiles}
\end{figure}

\section{Conclusion}
\label{sec:conclusion}

We provided a microstructural foundation for the rough Hawkes--Heston model by constructing a Poisson-embedded marked Hawkes order-flow model. Ordinary arrivals generate rough continuous volatility and leverage through a nearly unstable heavy-tailed Hawkes mechanism, while rare marked arrivals represent common shock events that produce simultaneous price jumps and volatility excitation. Under the nearly unstable scaling and the reduced-form admissibility conditions,
we showed the convergence along the full sequence for 
the complete rescaled price/variance/jump system to the unique complete canonical rough Hawkes--Heston weak solution.
The Hawkes renewal structure
yields a Mittag--Leffler Volterra representation, which is then rewritten in
Riemann--Liouville fractional form. We expressed the limiting coefficients 
in terms of the microscopic parameters. The effective Volterra drift
is determined by the balance between the near-critical Hawkes gap and the
predictable mean of common-jump feedback. Within the subcritical scaling
regime, this balance restricts the drift to be negative. The continuous price
martingale and its leverage correlation are generated by regular order flow
through a martingale-balanced asymmetry of sign-specific arrival rates and
tick sizes. Our construction combined this endogenous mechanism with
reduced-form exponential price normalization and direct marked-event price
loading and provided a
microstructural foundation for the variance and common-jump mechanism of the
rough Hawkes--Heston model.
Numerical experiments were provided 
to illustrate the convergence 
of our microstructural foundation to the rough Hawkes-Heston model.

\section*{Acknowledgements}

The authors would like to thank Mathieu Rosenbaum for helpful discussions.

\appendix
\section{Technical estimates}
\label{app:technical}

\subsection{Measurable recovery of the variance driver}
\label{app:measurable-recovery}

We complete the measurable-recovery step used in
Lemma~\ref{lem:recover-jump-measure}. For
$v\in L^2_{\mathrm{loc}}(\mathbb R_+)$, set
\[
  F_v:=L_\alpha*(v-G_0).
\]
On every compact interval, $L_\alpha\in L^1$, so Young's inequality makes
$v\mapsto F_v$ continuous from local $L^2$ to local $L^2$. For
$q\in\mathbb Q_+$, define
\[
  R_{n,q}(v):=n\int_q^{q+1/n}F_v(r)dr,
\]
and let $\ell_q(v)$ be its limit as $n\to\infty$ when the limit exists, and
zero otherwise. Each $\ell_q$ is Borel because the existence and value of the
limit are described by countably many relations among continuous functions.

Consider the evaluation embedding
\[
  \iota:D_{\mathrm{loc}}(\mathbb R_+)
  \longrightarrow\mathbb R^{\mathbb Q_+},
  \qquad x\longmapsto(x(q))_{q\in\mathbb Q_+}.
\]
It is injective and Borel. Since $D_{\mathrm{loc}}$ is a standard Borel
space, the Lusin--Souslin theorem implies that $\iota(D_{\mathrm{loc}})$ is
Borel and that $\iota^{-1}$ is Borel on its image. Hence
\[
  \mathcal R(v)
  =
  \begin{cases}
  \iota^{-1}\left((\ell_q(v))_{q\in\mathbb Q_+}\right),
  &(\ell_q(v))_q\in\iota(D_{\mathrm{loc}}),\\
  0,&\text{otherwise},
  \end{cases}
\]
defines a Borel map from $L^2_{\mathrm{loc}}$ to $D_{\mathrm{loc}}$.

For an actual weak solution, the resolvent formula in the proof of
Lemma~\ref{lem:recover-jump-measure} gives $F_{\sigma^2}=Z$ almost everywhere.
Right continuity of $Z$ then yields, for every $q\in\mathbb Q_+$,
\[
  R_{n,q}(\sigma^2)
  =n\int_q^{q+1/n}Z_rdr\longrightarrow Z_q,
\]
and therefore $\mathcal R(\sigma^2)=Z$.

Finally, the map
\[
  \mathcal J(x)=\sum_{s:\,\Delta x_s\ne0}\delta_{(s,\Delta x_s)}
\]
is Borel on $D_{\mathrm{loc}}$. Indeed, on each compact time interval the
finitely many jumps of absolute size at least $1/m$ can be ordered by Borel
hitting-time maps; allowing the horizon and $m$ to range over the positive
integers constructs $\mathcal J$ as a countable limit of Borel point-measure
maps. This proves the measurable assertions used in the lemma.

\subsection{Renewal-density estimates for the canonical kernel}
\label{app:renewal-density}

The next lemmas record the renewal-density asymptotics needed to pass from the
nearly unstable microscopic Hawkes kernel to the macroscopic Mittag--Leffler
kernel. The convergence is obtained from the stable density local limit theorem
together with a subordination representation of $K$, and the fractional
envelope used repeatedly in tightness and stochastic-convolution estimates
follows by combining the local limit theorem with the local large-deviation
estimate below.

\begin{lemma}[Density local limit theorem for the canonical kernel]
\label{lem:density-llt}
Let $\phi(t)=\alpha(1+t)^{-1-\alpha}$ with $\alpha\in(1/2,1)$, and let
$g_\alpha$ denote the density of the one-sided $\alpha$-stable law with
Laplace transform
\[
  \int_0^\infty e^{-\lambda y}g_\alpha(y)\,dy
  =\exp\!\left(-\Gamma(1-\alpha)\lambda^\alpha\right),
  \qquad \lambda\ge0 .
\]
Then
\begin{equation}
\label{eq:density-llt}
  \lim_{k\to\infty}\ \sup_{y>0}\
  \left|k^{1/\alpha}\phi^{*k}\!\left(k^{1/\alpha}y\right)-g_\alpha(y)\right|=0 .
\end{equation}
In particular, there exists $C_*<\infty$ such that
$\|\phi^{*k}\|_\infty\le C_*k^{-1/\alpha}$ for every $k\ge1$.
\end{lemma}

\begin{proof}
Since $\overline\Phi(t)=\int_t^\infty\phi(s)\,ds=(1+t)^{-\alpha}$, the law
with density $\phi$ belongs to the domain of normal attraction of the
one-sided $\alpha$-stable law; see, for example,
\citet[Section~8.3]{bingham1989regular}. Let
\[
  \varphi(u)=\int_0^\infty e^{iut}\phi(t)\,dt
\]
be its characteristic function. The expansion
\[
  1-\widehat\phi(\lambda)
  =\lambda\int_0^\infty e^{-\lambda t}\,\overline\Phi(t)\,dt
  =\lambda\int_0^\infty e^{-\lambda t}(1+t)^{-\alpha}dt
  \sim\Gamma(1-\alpha)\lambda^\alpha,
  \qquad\lambda\downarrow0,
\]
obtained by integration by parts and Karamata's Abelian theorem, together
with the cited stable domain-of-attraction theorem, implies
\begin{equation}
\label{eq:char-stable-domain}
\varphi\left(u/k^{1/\alpha}\right)^k\longrightarrow\widehat g_\alpha(u),
  \qquad u\in\mathbb R,
\end{equation}
where
$\widehat g_\alpha(u)=
\exp\{-\Gamma(1-\alpha)(-iu)^\alpha\}$, with the principal branch.
In particular, $|\widehat g_\alpha(u)|\le e^{-c_\alpha|u|^\alpha}$.

We prove that the convergence in \eqref{eq:char-stable-domain} holds in
$L^1(du)$. First, $\phi\in L^2(\mathbb R)$ after extension by zero to the
negative half-line, so $\varphi\in L^2(\mathbb R)$ by Plancherel's theorem.
Let $X,X'$ be independent with density $\phi$. For sufficiently small
$0<|u|\le u_0$,
\[
  1-|\varphi(u)|^2
  =\E\left[1-\cos(u(X-X'))\right].
\]
Restrict the expectation to
\[
  X'\in[0,1],
  \qquad
  X\in\left[|u|^{-1}+1,\,2|u|^{-1}\right].
\]
On this event $1\le |u|(X-X')\le2$, so
$1-\cos(u(X-X'))\ge c_0>0$. Moreover,
$\Q(X'\in[0,1])>0$, and the exact tail formula implies that
\[
\begin{aligned}
\Q\!\left(
X\in\left[|u|^{-1}+1,\,2|u|^{-1}\right]
\right)
&=
\left(2+|u|^{-1}\right)^{-\alpha}
-\left(1+2|u|^{-1}\right)^{-\alpha}\\
&=
|u|^\alpha\left(\left(1+2|u|\right)^{-\alpha}
-\left(2+|u|\right)^{-\alpha}\right)
\ge c_1|u|^\alpha,
\end{aligned}
\]
for all sufficiently small $|u|$, since the expression in parentheses
tends to $1-2^{-\alpha}>0$. Consequently
\[
  1-|\varphi(u)|^2\ge c|u|^\alpha.
\]
Using $\sqrt{1-y}\le e^{-y/2}$ for $0\le y\le1$, after decreasing $u_0$
and changing the constant, this yields
\[
  |\varphi(u)|\le \exp(-c|u|^\alpha),
  \qquad |u|\le u_0.
\]
Moreover, absolute continuity makes the distribution non-lattice, hence
$|\varphi(u)|<1$ for $u\ne0$, while the Riemann--Lebesgue lemma gives
$\varphi(u)\to0$ as $|u|\to\infty$. Consequently, for every $\delta>0$,
\[
  \rho_\delta:=\sup_{|u|\ge\delta}|\varphi(u)|<1.
\]

We split the Fourier integral into small, intermediate, and large
frequencies. On $|u|\le u_0 k^{1/\alpha}$, the preceding exponential
bound implies that
\[
  \left|\varphi\left(u/k^{1/\alpha}\right)\right|^k\le e^{-c|u|^\alpha},
\]
so dominated convergence and \eqref{eq:char-stable-domain} apply. For a
fixed $R>u_0$, the intermediate region
$u_0 k^{1/\alpha}<|u|\le Rk^{1/\alpha}$ is bounded by
$2Rk^{1/\alpha}\rho_{u_0}^k$. On the large-frequency region, Plancherel's theorem
gives
\[
  \int_{|u|>Rk^{1/\alpha}}
  \left|\varphi\left(u/k^{1/\alpha}\right)\right|^kdu
  =
  k^{1/\alpha}\int_{|v|>R}|\varphi(v)|^kdv
  \le
  k^{1/\alpha}\rho_R^{\,k-2}\|\varphi\|_2^2.
\]
Both bounds tend to zero. The corresponding tails of
$\widehat g_\alpha$ vanish by its exponential bound. Therefore
\[
  \int_{\mathbb R}
  \left|\varphi\left(u/k^{1/\alpha}\right)^k-\widehat g_\alpha(u)\right|du
  \longrightarrow0.
\]

Since $\varphi^k\in L^1$ for $k\ge2$, Fourier inversion yields
\[
  k^{1/\alpha}\phi^{*k}\left(k^{1/\alpha}x\right)
  =
  \frac1{2\pi}\int_{\mathbb R}e^{-iux}
  \varphi\left(u/k^{1/\alpha}\right)^kdu.
\]
The $L^1$ convergence of the characteristic functions implies uniform
convergence of the densities, proving \eqref{eq:density-llt}. It also gives
$\|\phi^{*k}\|_\infty\le Ck^{-1/\alpha}$ for all sufficiently large $k$.
For the finitely many remaining indices,
$\|\phi^{*k}\|_\infty\le\|\phi\|_\infty=\alpha$, and enlarging $C$ proves
the asserted bound for every $k\ge1$.
\end{proof}

\begin{lemma}[Local large-deviation upper bound]
\label{lem:lld}
There exist constants $c_0>0$ and $C_0<\infty$ such that
\begin{equation}
\label{eq:lld}
  \phi^{*k}(x)\ \le\ C_0\,k\,\phi(x)
  \qquad\text{for all }k\ge1\text{ and all }x\ge c_0\,k^{1/\alpha}.
\end{equation}
\end{lemma}

\begin{proof}
Set
\[
  A(x):=(1+x)^\alpha,
  \qquad
  a_k:=A^{-1}(k)=k^{1/\alpha}-1\sim k^{1/\alpha},
\]
and let $S_k:=X_1+\cdots+X_k$ and $M_k:=\max_{i\le k}X_i$ for independent
variables with density $\phi$. We use the constrained estimate in
\citet[Theorem~1.1]{caravennadoney2019}: for fixed $\gamma\in(0,1)$,
\begin{equation}
\label{eq:CD-constrained}
  \Q(S_k\in x+[0,1],\,M_k\le\gamma x)
  \le
  Ca_k^{-1}\left(\frac{k}{A(x)}\right)^{\lceil1/\gamma\rceil}.
\end{equation}
The case $k=1$ is immediate; below we take $k\ge2$, and any finite set of
remaining indices is absorbed by enlarging the constants.

Take $\gamma_0=1/5$, $\gamma_1=1/4$, and $\gamma_2=1/3$. Denote by
$q_{k,\gamma}(x)$ the contribution to the density of $S_k$ from
$M_k\le\gamma x$, namely
\[
q_{k,\gamma}(x)
=
\int_{\substack{x_1,\ldots,x_{k-1}\ge0\\
x_k=x-\sum_{j<k}x_j\ge0\\
\max_i x_i\le\gamma x}}
\prod_{i=1}^k\phi(x_i)\,
dx_1\cdots dx_{k-1}.
\]
For fixed $h\in[0,1]$, apply the transformation
\[
  (x_1,x_2,\ldots,x_k)
  \longmapsto
  (x_1+h,x_2,\ldots,x_k).
\]
The transformed variables sum to $x+h$ and the Jacobian is one. If
$x_i\le\gamma_0x$ for every $i$, then for sufficiently large $x$,
\[
  x_1+h\le\gamma_0x+1\le\gamma_1(x+h),
\]
because $\gamma_0=1/5<1/4=\gamma_1$; and, for $i\ge2$,
\[
  x_i\le\gamma_0x\le\gamma_1(x+h).
\]
Furthermore,
\[
  \phi(x_1)
  \le2^{1+\alpha}\phi(x_1+h),
  \qquad 0\le h\le1.
\]
The simplex change of variables therefore implies that
\[
  q_{k,\gamma_0}(x)
  \le 2^{1+\alpha}q_{k,\gamma_1}(x+h),
  \qquad 0\le h\le1.
\]
Integrating over $h\in[0,1]$ and using, for sufficiently large $x$,
\[
  \gamma_1(x+1)\le\gamma_2x,
  \qquad\gamma_2=\frac13,
\]
implies that
\[
  q_{k,\gamma_0}(x)
  \le  C\Q(S_k\in[x,x+1],\,M_k\le\gamma_2x)
  \le
  Ca_k^{-1}\left(\frac{k}{A(x)}\right)^3,
\]
where the last step is \eqref{eq:CD-constrained}.

For $x\ge c_0k^{1/\alpha}$, with $c_0$ sufficiently large, $a_k\ge
ck^{1/\alpha}$ and, since $1-2\alpha<0$,
\[
\begin{aligned}
  a_k^{-1}\left(\frac{k}{A(x)}\right)^3
  &\le
  C\frac{k}{(1+x)A(x)}
  =
  C'\,k\phi(x).
\end{aligned}
\]
Indeed, the exact ratio is
\[
\frac{
a_k^{-1}(k/A(x))^3
}{
k/((1+x)A(x))
}
=
k^2a_k^{-1}(1+x)^{1-2\alpha}.
\]
Since $a_k\ge ck^{1/\alpha}$ for $k\ge2$ and
$1-2\alpha<0$, the condition $x\ge c_0k^{1/\alpha}$ implies that
\[
  k^2a_k^{-1}(1+x)^{1-2\alpha}
  \le
  Ck^{2-1/\alpha}k^{(1-2\alpha)/\alpha}
  =C.
\]

It remains to control the contribution
$r_{k,\gamma_0}(x)$ from a big jump. By the union bound and symmetry,
\[
  r_{k,\gamma_0}(x)
  \le
  k\int_{\gamma_0x}^{x}
  \phi(y)\phi^{*(k-1)}(x-y)dy
  \le
  Ck\phi(x)\int_0^x\phi^{*(k-1)}(u)du
  \le Ck\phi(x),
\]
because $\phi(y)\le C\phi(x)$ for $y\in[\gamma_0x,x]$. Combining the
small-maximum and big-jump contributions proves
$\phi^{*k}(x)\le C_0k\phi(x)$ whenever
$x\ge c_0k^{1/\alpha}$.
\end{proof}

\begin{lemma}[Subordination representation of $K$]
\label{lem:subordination}
The kernel $K$ characterized by
$\widehat K(\lambda)=(\mu_0+\Gamma(1-\alpha)\lambda^\alpha)^{-1}$ admits the
representation
\begin{equation}
\label{eq:subordination}
  K(u)=\int_0^\infty e^{-\mu_0 v}\,v^{-1/\alpha}\,
  g_\alpha\!\left(u\,v^{-1/\alpha}\right)dv,
  \qquad u>0,
\end{equation}
where $g_\alpha$ is as in Lemma~\ref{lem:density-llt}. In particular, $K$ is
continuous and strictly positive on $(0,\infty)$.
\end{lemma}

\begin{proof}
All integrands below are nonnegative, so Tonelli's theorem applies. Using the
Laplace transform of $g_\alpha$,
\begin{align*}
  \int_0^\infty e^{-\lambda u}
  \int_0^\infty e^{-\mu_0 v}v^{-1/\alpha}
  g_\alpha\left(uv^{-1/\alpha}\right)dv\,du
  &=\int_0^\infty e^{-\mu_0 v}
  e^{-\Gamma(1-\alpha)\lambda^\alpha v}\,dv
  \\
  &=\frac{1}{\mu_0+\Gamma(1-\alpha)\lambda^\alpha}
  =\widehat K(\lambda).
\end{align*}
Uniqueness of Laplace transforms of locally integrable functions identifies
the right-hand side of \eqref{eq:subordination} with $K$. Continuity on
$(0,\infty)$ follows by dominated convergence, using the continuity and
boundedness of $g_\alpha$ together with the tail
$g_\alpha(y)\sim\alpha y^{-1-\alpha}$ as $y\to\infty$, which makes the
integrand $\mathcal{O}(v)$ as $v\downarrow0$, locally uniformly in $u$; strict
positivity follows from $g_\alpha>0$ on $(0,\infty)$.
\end{proof}

\begin{lemma}[Local renewal asymptotics for the canonical near-critical resolvent]
\label{lem:local-strong-renewal}
Let $\phi(t)=\alpha(1+t)^{-1-\alpha}$ with $\alpha\in(1/2,1)$, let
$a_T\in(0,1)$ satisfy $1-a_T\sim\mu_0T^{-\alpha}$ for some $\mu_0>0$, and
let
$\Psi_T:=\sum_{k\ge1}(a_T\phi)^{*k}$. Then, with
$\widehat K(\lambda)=(\mu_0+\Gamma(1-\alpha)\lambda^\alpha)^{-1}$:
\begin{enumerate}
\item[(i)] As $T\to\infty$, $T^{1-\alpha}\Psi_T(Tu)\to K(u)$ uniformly on compact subsets of
$(0,\infty)$.
\item[(ii)] There exists $C<\infty$ such that
\[
  T^{1-\alpha}\Psi_T(Tu)\le C\,u^{\alpha-1},
  \qquad u>0,\ T\ge1 .
\]
\end{enumerate}
\end{lemma}

\begin{proof}
\emph{Part (i).} Fix $0<u_0<R<\infty$; we prove uniform convergence on
$[u_0,R]$. For $\epsilon\in(0,1)$ and $M>1$ decompose
\[
  K_T(u):=T^{1-\alpha}\Psi_T(Tu)
  =\Sigma^{(1)}_T(u)+\Sigma^{(2)}_T(u)+\Sigma^{(3)}_T(u),
\]
where $\Sigma^{(1)}_T$, $\Sigma^{(2)}_T$, $\Sigma^{(3)}_T$ collect the terms
$T^{1-\alpha}a_T^k\phi^{*k}(Tu)$ with $1\le k<\epsilon T^\alpha$,
$\epsilon T^\alpha\le k\le MT^\alpha$ and $k>MT^\alpha$, respectively.

\emph{Step 1 (small-$k$ range).} Choose
$\epsilon\le(u_0/c_0)^\alpha$, with $c_0$ from Lemma~\ref{lem:lld}. For
$k<\epsilon T^\alpha$ and $u\ge u_0$ we have
$Tu\ge u_0T\ge c_0(\epsilon T^\alpha)^{1/\alpha}\ge c_0k^{1/\alpha}$, so
Lemma~\ref{lem:lld} and $\phi(Tu)\le\alpha(Tu)^{-1-\alpha}$ give
\begin{equation}
\label{eq:small-k-bound}
  0\le\Sigma^{(1)}_T(u)
  \le C_0\,T^{1-\alpha}\phi(Tu)\!\!\sum_{k<\epsilon T^\alpha}\!\!k
  \le C\,\epsilon^{2}\,u^{-1-\alpha}
  \le C\,\epsilon^{2}\,u_0^{-1-\alpha},
\end{equation}
uniformly in $T\ge1$ and $u\in[u_0,R]$.

\emph{Step 2 (large-$k$ range).} By Lemma~\ref{lem:density-llt},
$\|\phi^{*k}\|_\infty\le C_*k^{-1/\alpha}$, and $\log a_T\le-(1-a_T)$ gives
$a_T^k\le e^{-k(1-a_T)}$. Hence,
\[
  0\le\Sigma^{(3)}_T(u)
  \le C_*\,T^{1-\alpha}\left(MT^\alpha\right)^{-1/\alpha}
  \frac{e^{-MT^\alpha(1-a_T)}}{1-e^{-(1-a_T)}}
  = C_*\,M^{-1/\alpha}\,
  \frac{T^{-\alpha}}{1-e^{-(1-a_T)}}\,e^{-MT^\alpha(1-a_T)} .
\]
Since $T^\alpha(1-a_T)\to\mu_0$, for all sufficiently large $T$ we have
$T^{-\alpha}/(1-e^{-(1-a_T)})\le2/\mu_0$ and
$e^{-MT^\alpha(1-a_T)}\le e^{-M\mu_0/2}$, so
\begin{equation}
\label{eq:large-k-bound}
  0\le\Sigma^{(3)}_T(u)\le C\,M^{-1/\alpha}e^{-\mu_0M/2},
\end{equation}
uniformly in $u$.

\emph{Step 3 (bulk range).} For $\epsilon T^\alpha\le k\le MT^\alpha$ set
$v_k:=k/T^\alpha\in[\epsilon,M]$. By Lemma~\ref{lem:density-llt},
\[
  \phi^{*k}(Tu)
  =k^{-1/\alpha}\left[g_\alpha\left(Tuk^{-1/\alpha}\right)
  +\delta_k\left(Tuk^{-1/\alpha}\right)\right],
  \qquad
  \bar\delta_k:=\sup_{y>0}|\delta_k(y)|\longrightarrow0 .
\]
Since $T^{1-\alpha}k^{-1/\alpha}=T^{-\alpha}v_k^{-1/\alpha}$ and
$Tuk^{-1/\alpha}=uv_k^{-1/\alpha}$,
\[
  \Sigma^{(2)}_T(u)
  =T^{-\alpha}\sum_{\epsilon T^\alpha\le k\le MT^\alpha}
  a_T^{k}\,v_k^{-1/\alpha}\,g_\alpha\left(uv_k^{-1/\alpha}\right)
  +\mathcal E_T(u),
  \qquad
  |\mathcal E_T(u)|
  \le \epsilon^{-1/\alpha}M\max_{k\ge\epsilon T^\alpha}\bar\delta_k
  \longrightarrow0 .
\]
Moreover $a_T^{k}=\exp(v_kT^\alpha\log a_T)$ with
$T^\alpha\log a_T\to-\mu_0$, so $a_T^{k}=e^{-\mu_0v_k}(1+o(1))$ uniformly
over $v_k\in[\epsilon,M]$. The map
$(u,v)\mapsto e^{-\mu_0v}v^{-1/\alpha}g_\alpha(uv^{-1/\alpha})$ is continuous,
hence uniformly continuous, on $[u_0,R]\times[\epsilon,M]$, and the sums above
are Riemann sums with mesh $T^{-\alpha}$. Therefore,
\begin{equation}
\label{eq:bulk-limit}
  \Sigma^{(2)}_T(u)\longrightarrow
  \int_\epsilon^M e^{-\mu_0v}\,v^{-1/\alpha}
  g_\alpha\!\left(uv^{-1/\alpha}\right)\,dv,
  \qquad\text{uniformly in }u\in[u_0,R].
\end{equation}

\emph{Step 4 (conclusion).}
By the tail asymptotic
$g_\alpha(y)\sim\alpha y^{-1-\alpha}$ as $y\to\infty$,
there exist $C<\infty$ and $\epsilon_0>0$ such that, for
$0<\epsilon\le\epsilon_0$,
\[
  e^{-\mu_0v}v^{-1/\alpha}
  g_\alpha\!\left(uv^{-1/\alpha}\right)
  \le C u_0^{-1-\alpha}v,
  \qquad
  u\in[u_0,R],\quad v\in(0,\epsilon).
\]
Moreover, the boundedness of $g_\alpha$ gives, for $v\ge M$,
\[
  e^{-\mu_0v}v^{-1/\alpha}
  g_\alpha\!\left(uv^{-1/\alpha}\right)
  \le
  \|g_\alpha\|_\infty e^{-\mu_0v}v^{-1/\alpha}.
\]
Consequently,
\begin{equation}
\label{eq:subordination-tail-control}
\begin{aligned}
  \sup_{u\in[u_0,R]}
  \int_0^\epsilon
  e^{-\mu_0v}v^{-1/\alpha}
  g_\alpha\!\left(uv^{-1/\alpha}\right)\,dv
  &\le
  \frac{C}{2}u_0^{-1-\alpha}\epsilon^2
  \xrightarrow[\epsilon\downarrow0]{}0,\\
  \sup_{u\in[u_0,R]}
  \int_M^\infty
  e^{-\mu_0v}v^{-1/\alpha}
  g_\alpha\!\left(uv^{-1/\alpha}\right)\,dv
  &\le
  \frac{\|g_\alpha\|_\infty}{\mu_0}
  M^{-1/\alpha}e^{-\mu_0M}
  \xrightarrow[M\uparrow\infty]{}0.
\end{aligned}
\end{equation}
Given $\eta>0$, choose $\epsilon$ small and $M$ large so that the bounds in
\eqref{eq:small-k-bound}, \eqref{eq:large-k-bound}, and
\eqref{eq:subordination-tail-control}
are each smaller than $\eta$, uniformly in large $T$ and $u\in[u_0,R]$. Then
\eqref{eq:bulk-limit} and Lemma~\ref{lem:subordination} yield
\[
  \limsup_{T\to\infty}\ \sup_{u\in[u_0,R]}
  \left|K_T(u)-K(u)\right|\ \le\ 3\eta .
\]
Since $\eta>0$ was arbitrary, part (i) follows.

\emph{Part (ii).} Define the undamped renewal density
$U(x):=\sum_{k\ge1}\phi^{*k}(x)$, $x>0$. Since $0<a_T<1$,
\begin{equation}
  \label{eq:psi-leq-U}
  \Psi_T(x)=\sum_{k\ge1}a_T^k\phi^{*k}(x)\le U(x),
  \qquad x>0,
\end{equation}
so it suffices to prove a $T$-independent envelope for $U$. For large $x$,
let
\[
N_x:=\left\lfloor(x/c_0)^\alpha\right\rfloor,
\]
where $c_0$ is from Lemma~\ref{lem:lld}. For $k\le N_x$,
$x\ge c_0k^{1/\alpha}$, so that lemma and
$\phi(x)\le C(1+x)^{-1-\alpha}$ give
\[
  \sum_{k\le N_x}\phi^{*k}(x)
  \le C\phi(x)\sum_{k\le N_x}k
  \le Cx^{\alpha-1}.
\]
For $k>N_x$, Lemma~\ref{lem:density-llt} implies that
\[
  \sum_{k>N_x}\phi^{*k}(x)
  \le C\sum_{k>N_x}k^{-1/\alpha}
  \le CN_x^{1-1/\alpha}
  \le Cx^{\alpha-1}.
\]
Thus, for some $x_0$,
\begin{equation}
\label{eq:U-tail}
  U(x)\le Cx^{\alpha-1},\qquad x\ge x_0.
\end{equation}
On $(0,x_0]$ the renewal equation $U=\phi+\phi*U$ and Gr\"{o}nwall's inequality
show that $U$ is bounded; after enlarging $C$,
\begin{equation}
\label{eq:U-local}
  U(x)\le Cx^{\alpha-1},\qquad 0<x\le x_0.
\end{equation}
Combining these bounds with \eqref{eq:psi-leq-U} yields
\[
  T^{1-\alpha}\Psi_T(Tu)\le T^{1-\alpha}U(Tu)\le CT^{1-\alpha}(Tu)^{\alpha-1}
  =Cu^{\alpha-1},
  \qquad u>0,\ T\ge1 .
\]
This proves part (ii).
\end{proof}

\subsection{A priori estimates and tightness}
\label{app:apriori}

\begin{lemma}[Uniform Volterra-resolvent bound]
\label{lem:uniform-volterra-resolvent}
Let $0\le L_T(t)\le C_Lt^{\alpha-1}$ on $(0,1]$, with
$\alpha\in(0,1)$, and define
\[
  R_T:=\sum_{n\ge1}L_T^{*n}
\]
on the finite horizon $[0,1]$. Then
\[
  \sup_{T\ge 1}\|R_T\|_{L^1([0,1])}<\infty,
  \qquad
  \sup_{T\ge 1}\|R_T*f\|_{L^2([0,1])}\le C_R\|f\|_{L^2([0,1])}.
\]
\end{lemma}

\begin{proof}
Let $\ell(t)=C_Lt^{\alpha-1}$. Since the kernels are non-negative,
$L_T^{*n}\le \ell^{*n}$. 
Iterating the beta-integral \eqref{eq:beta-integral}, with
$a=\alpha$ and $b=(n-1)\alpha$ at the $n$th step, gives
\[
  \ell^{*n}(t)
  =
  \frac{C_L^n\Gamma(\alpha)^n}{\Gamma(n\alpha)}
  t^{n\alpha-1},
  \qquad t>0.
\]
Consequently,
\[
  \|\ell^{*n}\|_{L^1([0,1])}
  =
  \frac{C_L^n\Gamma(\alpha)^n}{\Gamma(n\alpha+1)}.
\]
The series
$\sum_{n\ge1}(C_L\Gamma(\alpha))^n/\Gamma(n\alpha+1)$ converges by the
Mittag--Leffler growth of the Gamma function. Hence,
$\sup_{T\ge 1}\|R_T\|_{L^1([0,1])}<\infty$. 
Young's inequality then implies that
\[
  \|R_T*f\|_2\le \|R_T\|_1\|f\|_2,
\]
which proves the $L^2$-operator bound.
\end{proof}

\begin{proposition}[Uniform first moment]
\label{prop:first-moment}
Under Assumptions \ref{ass:subcritical} and \ref{ass:baseline}, 
\[
\sup_{T\ge T_0}\E\int_0^1V_t^Tdt<\infty.
\]
\end{proposition}

\begin{proof}
Let $m_T(t)=\E\left[V_t^T\right]$. By \eqref{eq:VT-mild} and
\eqref{eq:JT-decomposition}, the two martingale terms are $\mathcal M^T$ and
$\bar{\mathcal J}^T$, defined in \eqref{eq:MT} and
\eqref{eq:JT-compensated}, while \eqref{eq:JT-predictable} expresses the
predictable jump term as a Volterra convolution of $V^T$.
Finite-$T$ non-explosion from Lemma~\ref{lem:finite-T-subcriticality},
together with the subcritical cluster representation used in its proof,
implies that the relevant counting compensators have finite expectation on
compact time intervals. Since $m_1<\infty$, the stochastic convolutions in
\eqref{eq:MT} and \eqref{eq:JT-compensated} are integrable martingales on
$[0,1]$. Consequently, we may take expectations in \eqref{eq:VT-mild} and use
\eqref{eq:JT-predictable} to obtain
\[
  m_T=g_0^T+K_T^{J,{\rm mean}}*m_T.
\]
The kernel $K_T^{J,{\rm mean}}$ is non-negative and bounded by
$Cu^{\alpha-1}$ on $(0,1]$. Let $R_T$ be its Volterra resolvent. Lemma
\ref{lem:uniform-volterra-resolvent} gives the uniform $L^1$-bound on $R_T$.
Here and below, $I$ denotes convolution with the Dirac mass $\delta_0$ at
zero.
\[
  m_T=(I+R_T)*g_0^T,
  \qquad
  \|m_T\|_{L^1([0,1])}
  \le (1+\|R_T\|_1)\|g_0^T\|_{L^1([0,1])}.
\]
Assumption \ref{ass:baseline} gives
$\sup_{T\ge T_0}\|g_0^T\|_1<\infty$, and therefore
$\sup_{T\ge T_0}\E\int_0^1V_t^Tdt<\infty$.
Lemma~\ref{lem:finite-T-subcriticality} gives finite-$T$ non-explosion for
every $T\ge T_0$.
\end{proof}

\begin{proposition}[Uniform second moment]
\label{prop:second-moment}
If $F_J$ satisfies the first- and second-moment conditions
\eqref{eq:FJ-moments},
\[
  \sup_{T\ge T_0}\E\left\|V^T\right\|_{L^2([0,1])}^2<\infty.
\]
\end{proposition}

\begin{proof}
By Lemma~\ref{lem:jump-mean}, write
\[
  V^T=g_0^T+\mathcal J^{T,{\rm pred}}+\mathcal M^T
  +\bar{\mathcal J}^T.
\]
The predictable term $\mathcal J^{T,{\rm pred}}$ has kernel
$K_T^{J,{\rm mean}}$. Let
\[
R_T:=\sum_{n\ge1}\left(K_T^{J,{\rm mean}}\right)^{*n}.
\]
By Lemma~\ref{lem:uniform-volterra-resolvent},
$\sup_{T\ge T_0}\|R_T\|_{L^1([0,1])}<\infty$, and $R_T$ acts as a uniformly bounded
operator on $L^2([0,1])$. Therefore
\[
  V^T=(I+R_T)*\left(g_0^T+\mathcal M^T+\bar{\mathcal J}^T\right).
\]

The baseline is bounded in $L^2$ by Assumption~\ref{ass:baseline}.
For the regular martingale term \eqref{eq:MT}, It\^o's isometry and
\[
  \left\langle\bar M^T\right\rangle_t\le C\int_0^t\lambda_s^Tds
\]
give
\[
  \E\left|\mathcal M_t^T\right|^2
  \le C\int_0^t(t-u)^{2\alpha-2}\E\left[V_u^T\right]du.
\]
For the compensated jump martingale \eqref{eq:JT-compensated}, using
$m_2<\infty$,
\[
  \E\left|\bar{\mathcal J}_t^T\right|^2
  \le C\int_0^t(t-u)^{2\alpha-2}\E\left[V_u^T\right]du.
\]
Since $2\alpha-2>-1$, integration over $[0,1]$ is finite. Using Proposition
\ref{prop:first-moment}, the right-hand sides of the martingale isometry estimates
are uniformly integrable over $t\in[0,1]$. More explicitly, integrating the
martingale isometry estimates over $t\in[0,1]$ and using Fubini's theorem gives
\[
  \sup_{T\ge T_0}\E\left\|\mathcal M^T\right\|_{L^2([0,1])}^2
  +\sup_{T\ge T_0}\E\left\|\bar{\mathcal J}^T\right\|_{L^2([0,1])}^2
  <\infty.
\]
Since $I+R_T$ is uniformly bounded on $L^2([0,1])$,
\[
  \E\left\|V^T\right\|_{L^2}^2
  \le
  C\E\left\|g_0^T+\mathcal M^T+\bar{\mathcal J}^T\right\|_{L^2}^2,
\]
and the right-hand side is uniformly bounded in $T$.
\end{proof}

\begin{proposition}[Tightness in $L^2$]
\label{prop:tightness}
The family $(V^T)_{T\ge T_0}$ is tight in $L^2([0,1])$.
\end{proposition}

\begin{proof}
The uniform $L^2$ bound is Proposition \ref{prop:second-moment}. It remains
to control translations. By Lemma~\ref{lem:jump-mean}, use the decomposition
\[
  V^T=g_0^T+\mathcal J^{T,{\rm pred}}+\mathcal M^T
  +\bar{\mathcal J}^T,
  \qquad
  \mathcal J_t^{T,{\rm pred}}
  =\int_0^tK_T^{J,{\rm mean}}(t-s)V_s^Tds.
\]
The baseline $g_0^T$ is controlled by $g_0^T\to g_0$ in $L^2$. The
predictable Volterra term
$\mathcal J^{T,{\rm pred}}=K_T^{J,{\rm mean}}*V^T$ is controlled by the
$L^1$-translation modulus of $K_T^{J,{\rm mean}}$, while the two martingale
convolutions $\mathcal M^T$ and $\bar{\mathcal J}^T$, defined in
\eqref{eq:MT} and \eqref{eq:JT-compensated}, are controlled by the
$L^2$-translation moduli of $T^{1-\alpha}H_T(T\cdot)$ and $K_J^T$,
respectively, using It\^o and random-measure isometries.
The required uniform kernel
translation estimate is proved in Appendix~\ref{app:tightness} by cutting off
the diagonal and using
\[
  \int_0^\varepsilon u^{2\alpha-2}du<\infty.
\]
Together with the uniform $L^2$ bound, these estimates imply that, for every
$\eta>0$, the family $(V^T)_{T\ge T_0}$ is contained with probability at least
$1-\eta$ in a set that is relatively compact in $L^2([0,1])$.
The Kolmogorov--M. Riesz--Fr\'echet compactness criterion in $L^2$
\cite[Theorem~4.26]{brezis2011functional} is applied on $\mathbb R$ to the
convolution extensions defined by zero-extending the kernels, and restriction
of those extensions to $[0,1]$ then yields tightness in $L^2([0,1])$.
\end{proof}

\subsection{Jump-measure convergence from compensators and martingale problems}
\label{app:jump-measure}

\begin{lemma}[Tightness of the rescaled jump measures]
\label{lem:jump-measure-tightness}
The family $(\bar\mu_J^T)_{T\ge T_0}$ is tight in
$\mathcal M_p([0,1]\times(0,\infty))$ endowed with the vague topology.
Consequently, $(V^T,\bar\mu_J^T)$ is tight in
$L^2([0,1])\times\mathcal M_p([0,1]\times(0,\infty))$.
\end{lemma}

\begin{proof}
Let $K\subset[0,1]\times(0,\infty)$ be compact. Then there exists a compact set $C\subset(0,\infty)$ such that
\[
  K\subset [0,1]\times C.
\]
By the compensator identity \eqref{eq:rescaled-compensator},
\[
\E \bar\mu_J^T(K)
\le
\E \bar\mu_J^T([0,1]\times C)
=
\E \bar\nu_J^T([0,1]\times C).
\]
Using
\[
  \bar\nu_J^T(ds,dz)
  =2\kappa_TV_s^TF_J(dz)ds,
\]
we get
\[
\E \bar\nu_J^T([0,1]\times C)
=
2\kappa_TF_J(C)\E\int_0^1V_s^Tds.
\]
Since $\kappa_T\to\kappa$ and Proposition~\ref{prop:first-moment} gives
\[
  \sup_{T\ge T_0}\E\int_0^1V_s^Tds<\infty,
\]
it follows that
\[
  \sup_{T\ge T_0}\E \bar\mu_J^T(K)<\infty
\]
for every compact $K\subset[0,1]\times(0,\infty)$. Therefore, by Markov's inequality,
\[
\sup_{T\ge T_0}\Q\left(\bar\mu_J^T(K)>M\right)
\le
\frac{1}{M}\sup_{T\ge T_0}\E\left[\bar\mu_J^T(K)\right]
\to0,
\qquad M\to\infty.
\]
Thus, the masses of $\bar\mu_J^T$ on every compact set are tight.

Because $[0,1]\times(0,\infty)$ is locally compact and Polish, the
compact-containment criterion gives tightness in the vague topology. The
space of point measures is closed under vague convergence, so the limit
remains in $\mathcal M_p$. Since $V^T$ is tight in $L^2([0,1])$ by
Proposition~\ref{prop:tightness}, product tightness gives the assertion.
\end{proof}

\begin{lemma}[Convergence of the microscopic jump compensators]
\label{lem:limiting-compensator-identification}
Let $V^T\Rightarrow V$ in $L^2([0,1])$. Then the predictable compensators
\[
  \bar\nu_J^T(dt,dz)=2\kappa_TV_t^TF_J(dz)dt
\]
converge vaguely, against test functions in $C_c([0,1]\times(0,\infty))$, to
\[
  V_t^p\nu_J(dz)dt,
  \qquad
  \nu_J(dz)=2\kappa F_J(dz).
\]
Here $V^p$ is the left-difference-quotient version defined above; since
$V^p=V$ $dt$-a.e., this limiting random measure is equivalently
$V_t\nu_J(dz)dt$.
\end{lemma}

\begin{proof}
Let $h\in C_c([0,1]\times(0,\infty))$, and define
\[
  H(t):=\int_{(0,\infty)} h(t,z)F_J(dz).
\]
Then $H$ is bounded and continuous. Since the map
$v\mapsto\int_0^1v_tH(t)dt$
is continuous on $L^2([0,1])$, and since $\kappa_T\to\kappa$, we have
\begin{equation}
\label{eq:jump-compensator-test-limit}
\begin{aligned}
\int_0^1\int_{(0,\infty)} h(t,z)\bar\nu_J^T(dt,dz)
&=2\kappa_T\int_0^1V_t^TH(t)dt\\
&\Rightarrow
2\kappa\int_0^1V_tH(t)dt
=
\int_0^1\int_{(0,\infty)} h(t,z)V_t\nu_J(dz)dt.
\end{aligned}
\end{equation}
Because $V^p=V$ $dt$-a.e., the limiting integral on the right-hand side of
\eqref{eq:jump-compensator-test-limit} is also the integral of $h$ against
$V_t^p\nu_J(dz)dt$. This proves the asserted vague convergence
of compensators.
\end{proof}

\begin{lemma}[No clustering of macroscopic marked jumps]
\label{lem:marked-no-clustering}
For every relatively compact Borel set $C\subset(0,\infty)$,
\[
\lim_{\eta\downarrow0}\limsup_{T\to\infty}
\Q\left(
\sup_{I\subset[0,1], |I|\le\eta}
\bar\mu_J^T(I\times C)\ge2
\right)=0.
\]
\end{lemma}

\begin{proof}
Write
\[
  N_t^{T,C}:=\bar\mu_J^T((0,t]\times C),
  \qquad
  \lambda_t^{T,C}
  =2\kappa_TF_J(C)V_t^T.
\]
Each $N^{T,C}$ is simple: the two underlying Poisson random measures have
diffuse time intensity, and their independent superposition has no
simultaneous time atoms almost surely.

For $I=(a,b]$, let $N_I:=N_b^{T,C}-N_a^{T,C}$ and
$M_t^{T,C}:=N_t^{T,C}-\int_0^t\lambda_s^{T,C}ds$, and write
$N_{a,s-}^{T,C}:=N_{s-}^{T,C}-N_a^{T,C}$ for $s\in I$. Thus,
\[
N_I(N_I-1)
=
2\int_{(a,b]}N_{a,s-}^{T,C}\,dN_s^{T,C}.
\]
Taking expectations and using the compensator,
\[
\E[N_I(N_I-1)]
=
2\E\int_a^bN_{a,s-}^{T,C}\lambda_s^{T,C}ds.
\]
Decompose the count increment into its martingale and compensator parts.
The martingale isometry for the compensated count, followed by Fubini's
theorem, gives
\[
\E\int_I\left|M_{s-}^{T,C}-M_a^{T,C}\right|^2ds
=
\E\int_a^b\int_a^s\lambda_u^{T,C}du\,ds
\le
|I|\E\int_I\lambda_u^{T,C}du.
\]
Moreover,
\[
2\int_a^b
\left(\int_a^s\lambda_u^{T,C}du\right)\lambda_s^{T,C}ds
=
\left(\int_I\lambda_s^{T,C}ds\right)^2
\le
|I|\int_I\left(\lambda_s^{T,C}\right)^2ds.
\]
Cauchy--Schwarz inequality therefore yields
\[
\E[N_I(N_I-1)]
\le
2\left(
|I|\E\int_I\lambda_s^{T,C}ds
\right)^{1/2}
\left(
\E\int_I\left(\lambda_s^{T,C}\right)^2ds
\right)^{1/2}
+
|I|\E\int_I\left(\lambda_s^{T,C}\right)^2ds.
\]

For $\eta\in(0,1/2)$, let $n_\eta:=\lceil\eta^{-1}\rceil$ and define
\[
  J_j:=[j\eta,(j+1)\eta),
  \qquad 0\le j\le n_\eta-2,
  \qquad
  J_{n_\eta-1}:=[(n_\eta-1)\eta,1].
\]
For $0\le j\le n_\eta-2$, set
\[
  K_j:=J_j\cup J_{j+1},
  \qquad
  a_j^T:=\E\int_{K_j}\left(\lambda_s^{T,C}\right)^2ds.
\]
Any two points at distance at most $\eta$ lie in some $K_j$, and
$|K_j|\le2\eta$.
Because $\kappa_T$ is bounded and
$\sup_T\E\int_0^1|V_s^T|^2ds<\infty$, the bounded overlap of the $K_j$
gives
\[
  \sum_{j=0}^{n_\eta-2}a_j^T\le C.
\]
Also
\[
\E\int_{K_j}\lambda_s^{T,C}ds
\le |K_j|^{1/2}\left(a_j^T\right)^{1/2}.
\]
Applying the factorial-moment bound to $K_j$, summing, and using H\"older's
inequality,
\begin{equation}
\label{eq:no-clustering-factorial-bound}
\begin{aligned}
\sum_{j=0}^{n_\eta-2}
\E\left[N^{T,C}(K_j)\left(N^{T,C}(K_j)-1\right)\right]
&\le
C\eta^{3/4}\sum_{j=0}^{n_\eta-2}\left(a_j^T\right)^{3/4}
+C\eta\sum_{j=0}^{n_\eta-2}a_j^T\\
&\le C\eta^{1/2}+C\eta,
\end{aligned}
\end{equation}
since the number of intervals is $\mathcal{O}(\eta^{-1})$. If an interval of length
at most $\eta$ contains two points, some $K_j$ contains both. Markov's
inequality together with \eqref{eq:no-clustering-factorial-bound} proves the
assertion.
\end{proof}

\begin{proposition}[Point-process stability of rare marked jumps]
\label{prop:point-process-stability}
Let $(V^T)_{T\ge T_0}$ be tight in $L^2([0,1])$, and assume the uniform first
and second moment bounds
\[
\sup_{T\ge T_0} \E\int_0^1 V_t^Tdt<\infty,
\qquad
\sup_{T\ge T_0}\E\|V^T\|_{L^2([0,1])}^2<\infty.
\]
Let
\[
  \bar\mu_J^T(dt,dz)
  :=\mu_{J,{\rm com}}^T(Tdt,dz)
\]
be the rescaled common jump-marked measure. Then $(V^T,\bar\mu_J^T)$ is tight
in
\[
  L^2([0,1])\times \mathcal M_p([0,1]\times(0,\infty)).
\]
Moreover, along every jointly convergent subsequence
\[
  \left(V^T,\bar\mu_J^T\right)\Rightarrow(V,\mu),
\]
the limiting coordinate $\mu$ is a boundedly finite point measure that is
simple in time:
\[
  \mu(\{t\}\times(0,\infty))\le1
  \qquad\text{for every }t\in[0,1],
  \quad\text{almost surely}.
\]
\end{proposition}

\begin{proof}
The tightness of $V^T$ is assumed, while the tightness of
$\bar\mu_J^T$ follows from Lemma \ref{lem:jump-measure-tightness}. Hence the
pair is tight.

The limit belongs to $\mathcal M_p$ because the space of boundedly finite
point measures is closed under vague convergence. Fix a jointly convergent
subsequence, index it by $T_n$, and use a Skorokhod representation such that
\begin{equation}
\label{eq:point-stability-a.s.-vague}
  \left(V^{T_n},\bar\mu_J^{T_n}\right)\longrightarrow(V,\mu)
  \qquad\text{almost surely in }L^2\times\mathcal M_p.
\end{equation}

Let $D\subset(0,\infty)$ be a countable dense set of continuity points of
$F_J$, and let
\[
  \mathcal C
  :=
  \{[a,b]:a,b\in D,\ 0<a<b<\infty\}.
\]
This is a countable family such that any two positive marks are contained in
the interior of some $C\in\mathcal C$. We first show that every
$C=[a,b]\in\mathcal C$ is almost surely a mark-continuity set for $\mu$.
For sufficiently small $\varepsilon>0$,
let
\[
U_\varepsilon
:=(a-\varepsilon,a+\varepsilon)
  \cup(b-\varepsilon,b+\varepsilon)
\subset(0,\infty),
\]
and choose
$f_\varepsilon\in C_c((0,\infty))$ with
\[
  0\le f_\varepsilon\le1,
  \qquad
  f_\varepsilon=1\text{ on }\partial C,
  \qquad
  \operatorname{supp}(f_\varepsilon)\subset U_\varepsilon.
\]
The compensator formula and the uniform first-moment bound yield
\begin{align*}
  \sup_n\E\int_{[0,1]\times(0,\infty)}
  f_\varepsilon(z)\bar\mu_J^{T_n}(dt,dz)
  =
  \sup_n 2\kappa_{T_n}F_J(f_\varepsilon)
  \E\int_0^1V_t^{T_n}dt
  \le C F_J(U_\varepsilon).
\end{align*}
By \eqref{eq:point-stability-a.s.-vague}, continuity of integration against
$f_\varepsilon$, and Fatou's lemma,
\begin{align*}
  \E\mu([0,1]\times\partial C)
  \le
  \E\int_{[0,1]\times(0,\infty)}f_\varepsilon(z)\mu(dt,dz)
  \le C F_J(U_\varepsilon).
\end{align*}
Since $F_J(\partial C)=0$, letting $\varepsilon\downarrow0$ gives
\begin{equation}
\label{eq:limit-mark-boundary-zero}
  \mu([0,1]\times\partial C)=0
  \qquad\text{almost surely},\quad C\in\mathcal C.
\end{equation}

For $C\in\mathcal C$, define the limiting collision event
\[
  E_C
  :=
  \left\{\exists t\in[0,1]:\mu(\{t\}\times C)\ge2\right\}
\]
and, for $m\ge1$, the prelimit event
\[
  E_{n,m,C}
  :=
  \left\{
  \sup_{I\subset[0,1],|I|\le2/m}
  \bar\mu_J^{T_n}(I\times C)\ge2
  \right\}.
\]
On $E_C$, choose a collision time $t$. For every $m$, one can choose a
half-open interval $I_m\subset[0,1]$, containing $t$ and with
$|I_m|\le2/m$, whose endpoint sections are not charged by $\mu$ on $C$.
Together with \eqref{eq:limit-mark-boundary-zero}, this makes
$I_m\times C$ a relatively compact continuity set and
\[
  \mu(I_m\times C)\ge2.
\]
The almost sure vague convergence in
\eqref{eq:point-stability-a.s.-vague} then implies
\[
  \bar\mu_J^{T_n}(I_m\times C)\ge2,
\]
for all sufficiently large $n$. Hence, for every fixed $m$,
\begin{equation}
\label{eq:collision-event-inclusion}
  \mathbf 1_{E_C}
  \le
  \liminf_{n\to\infty}\mathbf 1_{E_{n,m,C}},
  \qquad\text{almost surely}.
\end{equation}
Fatou's lemma and \eqref{eq:collision-event-inclusion} give
\begin{align*}
  \Q(E_C)
  \le
  \liminf_{n\to\infty}\Q(E_{n,m,C})
  \le
  \limsup_{T\to\infty}
  \Q\left(
  \sup_{I\subset[0,1],|I|\le2/m}
  \bar\mu_J^T(I\times C)\ge2
  \right).
\end{align*}
Letting $m\to\infty$ and applying
Lemma~\ref{lem:marked-no-clustering} yields $\Q(E_C)=0$. Since
$\mathcal C$ is countable and contains any two positive marks in one of its
members,
\[
  \Q\left(
  \exists t\in[0,1]:
  \mu(\{t\}\times(0,\infty))\ge2
  \right)
  \le\sum_{C\in\mathcal C}\Q(E_C)=0.
\]
This proves time simplicity simultaneously for every $t\in[0,1]$.
\end{proof}

\subsection{Singular Volterra stochastic-convolution convergence}
\label{app:volterra-conv}
\begin{lemma}[Convergence of singular martingale convolutions]
\label{lem:singular-martingale-convolution}
Let $K_T\to K$ locally uniformly on $(0,1]$, and assume
\[
  |K_T(u)|+|K(u)|\le Cu^{\alpha-1},
  \qquad u\in(0,1],
  \qquad \alpha>1/2.
\]
Let $M^T$ be square-integrable martingales with
\[
  d\left\langle M^T\right\rangle_t=q_t^Tdt,
  \qquad
  \sup_{T\ge T_0}\E\int_0^1q_t^Tdt<\infty,
\]
Let $Y^T$ be any additional canonical coordinates with values in a Polish
space, and assume
\[
  \left(M^T,Y^T\right)\Rightarrow(M,Y),
\]
where $M$ is a continuous square-integrable martingale with
$\E\langle M\rangle_1<\infty$. Then
\[
\left(
  Y^T,
  \left(\int_0^tK_T(t-s)dM_s^T\right)_{t\in[0,1]}
\right)
\Rightarrow
\left(
  Y,
  \left(\int_0^tK(t-s)dM_s\right)_{t\in[0,1]}
\right)
\]
in the product with $L^2([0,1])$.
\end{lemma}

\begin{proof}
For $\varepsilon>0$, choose a smooth function
$\chi_\varepsilon:[0,1]\to[0,1]$ such that
$\chi_\varepsilon=0$ on $[0,\varepsilon/2]$ and
$\chi_\varepsilon=1$ on $[\varepsilon,1]$, and define
\[
  H_T^\varepsilon:=\chi_\varepsilon K_T,
  \qquad
  H^\varepsilon:=\chi_\varepsilon K.
\]
The near-diagonal error satisfies, by It\^o's isometry and Fubini's theorem,
\[
\E\left\|
  \int_0^\cdot
  (1-\chi_\varepsilon)(\cdot-s)K_T(\cdot-s)dM_s^T
\right\|_2^2
\le
C\varepsilon^{2\alpha-1}
\sup_{T\ge T_0}\E\left\langle M^T\right\rangle_1.
\]
The same bound holds for $M$ and $K$. It vanishes uniformly as
$\varepsilon\downarrow0$ because $\alpha>1/2$.

For fixed $\varepsilon$, local uniform convergence away from zero gives, as
$T\to\infty$,
\[
  \left\|H_T^\varepsilon-H^\varepsilon\right\|_\infty\longrightarrow0.
\]
Consequently, as $T\to\infty$,
\[
\E\left\|
  \int_0^\cdot
  (H_T^\varepsilon-H^\varepsilon)(\cdot-s)dM_s^T
\right\|_2^2
\le
\|H_T^\varepsilon-H^\varepsilon\|_\infty^2
\sup_{T\ge T_0}\E\left\langle M^T\right\rangle_1
\longrightarrow0.
\]

It remains to treat the fixed cutoff kernel $H^\varepsilon$. Choose
$h_n\in C^1([0,1])$ with $h_n(0)=0$ and
$\|h_n-H^\varepsilon\|_\infty\to0$. This uniform approximation gives the
stochastic-integral estimate used below. Integration by parts gives
\[
  \int_0^t h_n(t-s)dM_s^T
  =
  \int_0^t h_n'(t-s)M_s^Tds,
\]
and the same identity holds for $M$. Since $M$ is continuous, the map
\[
  m\longmapsto
  \left(\int_0^t h_n'(t-s)m_sds\right)_{t\in[0,1]}
\]
is continuous at $M$ from the Skorokhod space into $L^2([0,1])$.
The continuous mapping theorem therefore gives joint convergence with
$Y^T$ for each $n$. The uniform approximation errors are bounded by
\[
  \|h_n-H^\varepsilon\|_\infty^2
  \sup_{T\ge T_0}\E\langle M^T\rangle_1
\]
before the limit and by the corresponding expression with
$\E\langle M\rangle_1$ after the limit. We may thus let
$T\to\infty$, then $n\to\infty$, and finally
$\varepsilon\downarrow0$. Applying the converging-together lemma stated in
\eqref{eq:converging-together-hypotheses}-\eqref{eq:converging-together-error} first to the $C^1$ approximation
$h_n\to H^\varepsilon$ and then to the diagonal cutoff
$\varepsilon\downarrow0$ proves the claim.
\end{proof}

Applying Lemma \ref{lem:singular-martingale-convolution} with
$K_T(u)=T^{1-\alpha}H_T(Tu)$ and $M^T=B^T$ gives the martingale
convergence used in Proposition \ref{prop:regular-volterra-limit}.

\subsection{Raw jump-feedback convergence with mark truncation}
\label{app:raw-jump}
We provide the truncation details for Proposition \ref{prop:raw-feedback-limit}.
After the time change,
\[
  \mathcal J_t^T
  =
  \int_{(0,t)}K_J^T(t-s)\int_{(0,\infty)} z\bar\mu_J^T(ds,dz),
  \qquad
  \left|K_J^T(u)\right|\le Cu^{\alpha-1},
\]
and $K_J^T(u)\to\xi K(u)$ locally uniformly on $(0,1]$. The target is
\[
\xi\int_{[0,t)}K(t-s)\int_{(0,\infty)} z\mu(ds,dz).
\]
After passing to a Skorokhod representation of the considered subsequence,
Fatou's lemma and the uniform moment bounds from Propositions
\ref{prop:first-moment} and \ref{prop:second-moment} give
\[
  \E\int_0^1V_sds<\infty,
  \qquad
  \E\|V\|_{L^2([0,1])}^2<\infty.
\]
For $R>0$ and $\varepsilon>0$, define
\[
\mathcal J_t^{T,R}
:=
\int_{(0,t)}K_J^T(t-s)
\int_{(0,\infty)} z\mathbf 1_{\{z\le R\}}\bar\mu_J^T(ds,dz),
\]
and
\[
\mathcal J_t^{T,R,\varepsilon}
:=
\int_{(0,t)}K_J^T(t-s)\mathbf 1_{\{t-s>\varepsilon\}}
\int_{(0,\infty)} z\mathbf 1_{\{z\le R\}}\bar\mu_J^T(ds,dz).
\]
Let $\mathcal J^{R,\varepsilon}$ and $\mathcal J^R$ denote the corresponding
limits obtained by replacing $K_J^T,\bar\mu_J^T$ with $\xi K,\mu$.

\begin{lemma}[Bounded marks plus diagonal cutoff]
\label{lem:bounded-marks-diagonal-cutoff}
For fixed $R<\infty$ satisfying $F_J(\{R\})=0$ and every
$\varepsilon>0$,
\[
  \mathcal J^{T,R,\varepsilon}\Rightarrow\mathcal J^{R,\varepsilon}
\]
in $L^2([0,1])$.
\end{lemma}

\begin{proof}
Fix $R<\infty$ and $\varepsilon>0$. We first prove the convergence with
a lower mark cutoff and a continuous compactly supported mark function. Let
$0<\delta<R$, and let $h\in C_c((0,\infty))$ be bounded. Later we take
$h$ approximating $z\mathbf 1_{\{\delta\le z\le R\}}$. Define
\[
  F_T^h(t)
  :=
  \int_{[0,t)}
  K_J^T(t-s)\mathbf 1_{\{t-s>\varepsilon\}}h(z)
  \bar\mu_J^T(ds,dz),
\]
and
\[
  F^h(t)
  :=
  \int_{[0,t)}
  \xi K(t-s)\mathbf 1_{\{t-s>\varepsilon\}}h(z)
  \mu(ds,dz).
\]

By the Skorokhod representation theorem, along the considered subsequence we
may assume that
$\bar\mu_J^T\to\mu$
vaguely almost surely. We also have
$K_J^T\to \xi K$
uniformly on $[\varepsilon,1]$. We prove the convergence on this
representation space.

For fixed $t$, the test function
\[
  (s,z)\mapsto
  K_J^T(t-s)\mathbf 1_{\{0<t-s-\varepsilon\}}h(z)
\]
is supported in the compact set
\[
  [0,(t-\varepsilon)_+]\times \operatorname{supp}(h),
\]
but it is discontinuous at the time boundary $s=t-\varepsilon$. If
\[
  \mu(\{t-\varepsilon\}\times \operatorname{supp}(h))=0,
\]
then vague convergence of $\bar\mu_J^T$, together with uniform convergence
of $K_J^T$ on $[\varepsilon,1]$, yields
\[
  F_T^h(t)\to F^h(t).
\]
The exceptional set of such $t$'s has Lebesgue measure zero, because $\mu$
is boundedly finite on compact subsets and hence charges at most countably many
time sections of the compact set $[0,1]\times \operatorname{supp}(h)$.

Moreover, since $K_J^T$ is uniformly bounded on $[\varepsilon,1]$,
\[
  \left|F_T^h(t)\right|
  \le
  C_{\varepsilon,h}
  \bar\mu_J^T([0,1]\times \operatorname{supp}(h)).
\]
For almost every sample path, vague convergence implies that the compact masses
$\bar\mu_J^T([0,1]\times \operatorname{supp}(h))$ are eventually bounded,
provided the compact support is replaced, if necessary, by a slightly larger
compact continuity set. Hence, the preceding bound gives a pathwise $L^2(dt)$
dominating function. Therefore, by dominated convergence,
\[
  \left\|F_T^h-F^h\right\|_{L^2([0,1])}\to0
\]
almost surely, and hence in distribution.

We now pass from continuous mark functions to
$h_{\delta,R}(z):=z\mathbf 1_{\{\delta\le z\le R\}}$. Choose continuous
functions $h_n^-,h_n^+\in C_c((0,\infty))$ such that
\[
  0\le h_n^-\le h_{\delta,R}\le h_n^+,
\]
and
\[
  h_n^+(z)-h_n^-(z)
  \le
  C_R\left(
    \mathbf 1_{\{\delta-1/n<z<\delta+1/n\}}
    +
    \mathbf 1_{\{R-1/n<z<R+1/n\}}
  \right).
\]
For these continuous approximations the convergence has already been proved.
The error caused by replacing $h_{\delta,R}$ by $h_n^\pm$ is controlled by
the compact-boundary mass of the prelimit and limiting point measures. Put
\[
  B_n
  =
  (\delta-1/n,\delta+1/n)
  \cup(R-1/n,R+1/n)
\]
and denote the resulting boundary approximation error by $E^{T,n}$. Since
the cutoff kernel is uniformly bounded for fixed $\varepsilon$ and $R$, the
pathwise estimate is
\[
  \left\|E^{T,n}\right\|_{L^2([0,1])}
  \le
  C_{\varepsilon,R}\,
  \bar\mu_J^T([0,1]\times B_n).
\]
Therefore Markov's inequality and the compensator formula give, for every
$\eta>0$,
\[
  \Q\left(\left\|E^{T,n}\right\|_2>\eta\right)
\le
  \frac{C_{\varepsilon,R}}{\eta}
  \E\bar\mu_J^T([0,1]\times B_n)
  \le
  \frac{C_{\varepsilon,R}}{\eta}F_J(B_n)
  \sup_{T\ge T_0}\E\int_0^1V_s^Tds.
\]
Choose $\delta$ and $R$ to be continuity points of $F_J$ (equivalently of
$\nu_J=2\kappa F_J$). Then $F_J(B_n)\to0$. The identical pathwise and
Markov bounds apply to the limiting point measure with $F_J$ replaced by
$\nu_J$. Hence the approximation error converges to zero in probability in
$L^2([0,1])$. The converging-together lemma, with
$X^{T,n}=F_T^{h_n^-}$, $X^n=F^{h_n^-}$,
$X^T=F_T^{h_{\delta,R}}$, and $X=F^{h_{\delta,R}}$, gives
\[
  F_T^{h_{\delta,R}}\Rightarrow F^{h_{\delta,R}}
\]
in $L^2([0,1])$.

It remains to show that the auxiliary lower-cutoff error vanishes as
$\delta\downarrow0$. This error is obtained by replacing
$z\mathbf 1_{\{z\le R\}}$ with $z\mathbf 1_{\{\delta\le z\le R\}}$, and is
therefore exactly the small-mark contribution associated with
$z\mathbf 1_{\{0<z<\delta\}}$.
Write this contribution as the
sum of its compensated and predictable parts:
\[
\mathrm{SmallComp}^{T,\delta}_t
:=
\int_{[0,t)}
K_J^T(t-s)\mathbf 1_{\{t-s>\varepsilon\}}
\int_{(0,\delta)} z\widetilde M_J^T(ds,dz),
\]
and
\[
\mathrm{SmallPred}^{T,\delta}_t
:=
\int_{[0,t)}
K_J^T(t-s)\mathbf 1_{\{t-s>\varepsilon\}}
\int_{(0,\delta)} z\bar\nu_J^T(ds,dz).
\]
Since the kernel
$K_J^T(\cdot)\mathbf 1_{\{\cdot>\varepsilon\}}$ is uniformly bounded for fixed
$\varepsilon>0$, the random-measure isometry gives
\[
\begin{aligned}
\E\left\|\mathrm{SmallComp}^{T,\delta}\right\|_{L^2([0,1])}^2
&=
\E\int_0^1
\left|
\int_{[0,t)}
K_J^T(t-s)\mathbf 1_{\{t-s>\varepsilon\}}
\int_{(0,\delta)}z\widetilde M_J^T(ds,dz)
\right|^2dt \\
&\le
C_\varepsilon
\E\int_0^1\int_{[0,t)}
\int_{(0,\delta)}z^2\bar\nu_J^T(ds,dz)dt  \\
&\le
C_\varepsilon
\left(\int_{(0,\delta)}z^2F_J(dz)\right)
\sup_{T\ge T_0}\E\int_0^1V_s^Tds .
\end{aligned}
\]
For the predictable part, using
\[
  \bar\nu_J^T(ds,dz)=2\kappa_TV_s^TF_J(dz)ds,
\]
the boundedness of $\kappa_T$, and Young's inequality, we obtain
\[
\begin{aligned}
\mathrm{SmallPred}^{T,\delta}_t
&=
2\kappa_T
\left(\int_{(0,\delta)}zF_J(dz)\right)
\int_0^t
K_J^T(t-s)\mathbf 1_{\{t-s>\varepsilon\}}V_s^Tds,
\end{aligned}
\]
and therefore
\[
\begin{aligned}
\E\left\|\mathrm{SmallPred}^{T,\delta}\right\|_{L^2([0,1])}^2
&\le
C_\varepsilon
\left(\int_{(0,\delta)}zF_J(dz)\right)^2
\E\left\|V^T\right\|_{L^2([0,1])}^2  \\
&\le
C_\varepsilon
\left(\int_{(0,\delta)}zF_J(dz)\right)^2
\sup_{T\ge T_0}\E\left\|V^T\right\|_{L^2([0,1])}^2 .
\end{aligned}
\]
Both bounds vanish as $\delta\downarrow0$, by $m_1<\infty$,
$m_2<\infty$, Proposition~\ref{prop:first-moment}, and
Proposition~\ref{prop:second-moment}. The limiting estimates are identical,
with $F_J$ replaced by $\nu_J$. Hence the convergence with
$h_{\delta,R}$ implies the desired convergence with
$z\mathbf 1_{\{z\le R\}}$. This proves
\[
  \mathcal J^{T,R,\varepsilon}\Rightarrow \mathcal J^{R,\varepsilon}
\]
in $L^2([0,1])$.
\end{proof}

\begin{lemma}[Vanishing of the diagonal cutoff error]
\label{lem:diagonal-cutoff-error}
For every fixed $R<\infty$,
\[
  \lim_{\varepsilon\downarrow0}
  \sup_{T\ge T_0}\E\left\|\mathcal J^{T,R}-\mathcal J^{T,R,\varepsilon}\right\|_{L^2([0,1])}^2
  =0.
\]
The same conclusion holds for the limiting processes
$\mathcal J^{R,\varepsilon}\to\mathcal J^R$.
\end{lemma}

\begin{proof}
Split the difference into the compensated and predictable parts. The
compensated part satisfies the random-measure isometry bound
\[
\E\int_0^1\left|
\int_0^tK_J^T(t-s)\mathbf 1_{\{t-s\le\varepsilon\}}
\int_{(0,\infty)} z\mathbf 1_{\{z\le R\}}\widetilde M_J^T(ds,dz)
\right|^2dt
\le
C_R\varepsilon^{2\alpha-1}\sup_{T\ge T_0}\E\int_0^1V_s^Tds.
\]
For the predictable part, Young's inequality gives
\[
\E\left\|
\left(u^{\alpha-1}\mathbf 1_{\{u\le\varepsilon\}}\right)*V^T
\right\|_{L^2}^2
\le
C\varepsilon^{2\alpha}\sup_{T\ge T_0}\E\left\|V^T\right\|_{L^2([0,1])}^2.
\]
Both bounds vanish as $\varepsilon\downarrow0$, since $\alpha>1/2$. The
limit estimates are identical with $V$ and $\nu_J$.
\end{proof}

\begin{lemma}[Vanishing of the mark truncation error]
\label{lem:mark-truncation-error}
Under $m_2<\infty$,
\[
  \lim_{R\to\infty}
  \sup_{T\ge T_0}\E\left\|\mathcal J^T-\mathcal J^{T,R}\right\|_{L^2([0,1])}^2
  =0.
\]
The same tail estimate holds for the limiting jump feedback.
\end{lemma}

\begin{proof}
By the definitions of $\mathcal J^T$ and $\mathcal J^{T,R}$ in \ref{app:raw-jump}, the mark-truncation
error is the large-mark tail
\begin{equation}
\mathcal J^T_t-\mathcal J^{T,R}_t
=
\int_{(0,t)} K^T_J(t-s)
\int_{(R,\infty)} z\,\bar\mu^T_J(ds,dz).
\label{eq:large-mark-tail}
\end{equation}
Using
\[
\bar\mu^T_J=\widetilde M^T_J+\bar\nu^T_J,
\]
we split \eqref{eq:large-mark-tail} into its compensated and
predictable parts. For the compensated tail, the random-measure isometry gives
\[
\begin{aligned}
&\E\int_0^1\left|
\int_0^tK_J^T(t-s)
\int_{(0,\infty)} z\mathbf 1_{\{z>R\}}\widetilde M_J^T(ds,dz)
\right|^2dt  \\
&\quad
=
\E\int_0^1\int_0^t \left|K_J^T(t-s)\right|^2
\int_{(R,\infty)} z^2\bar\nu_J^T(ds,dz)dt  \\
&\quad
=
2\kappa_T\int_{(R,\infty)}z^2F_J(dz)
\E\int_0^1\int_s^1 \left|K_J^T(t-s)\right|^2dt V_s^Tds  \\
&\quad
\le
C\int_{(R,\infty)}z^2F_J(dz)
\E\int_0^1V_s^Tds.
\end{aligned}
\]
Here, we used $\bar\nu_J^T(ds,dz)=2\kappa_TV_s^TF_J(dz)ds$,
the boundedness of $\kappa_T$, and
\[
  \int_s^1 \left|K_J^T(t-s)\right|^2dt
  \le
  C\int_0^1 u^{2\alpha-2}du
  <\infty,
\]
because $\alpha>1/2$.

For the predictable tail, define
\[
  \mathrm{PredTail}^{T,R}_t
  :=
  \int_0^tK_J^T(t-s)
  \int_{(R,\infty)}z\bar\nu_J^T(ds,dz).
\]
Using again
$\bar\nu_J^T(ds,dz)=2\kappa_TV_s^TF_J(dz)ds$, we have
\[
  \mathrm{PredTail}^{T,R}_t
  =
  2\kappa_T
  \left(\int_{(R,\infty)}zF_J(dz)\right)
  \int_0^tK_J^T(t-s)V_s^Tds.
\]
Since $\kappa_T$ is bounded and
\[
  \left|K_J^T(u)\right|\le Cu^{\alpha-1},\qquad u\in(0,1],
\]
Young's inequality gives
\[
\begin{aligned}
\E\left\|\mathrm{PredTail}^{T,R}\right\|_{L^2([0,1])}^2
&\le
C
\left(\int_{(R,\infty)}zF_J(dz)\right)^2
\E\left\|
  \int_0^\cdot (\cdot-s)^{\alpha-1}V_s^Tds
\right\|_{L^2([0,1])}^2  \\
&\le
C
\left(\int_{(R,\infty)}zF_J(dz)\right)^2
\left\|u^{\alpha-1}\right\|_{L^1([0,1])}^2
\E\left\|V^T\right\|_{L^2([0,1])}^2  \\
&\le
C
\left(\int_{(R,\infty)}zF_J(dz)\right)^2
\sup_{T\ge T_0}\E\left\|V^T\right\|_{L^2([0,1])}^2.
\end{aligned}
\]
Moreover,
\[
  \int_{(R,\infty)}zF_J(dz)\to0,
  \qquad
  \int_{(R,\infty)}z^2F_J(dz)\to0,
\]
because $m_1<\infty$ and $m_2<\infty$. Hence both tails vanish uniformly in
$T$ as $R\to\infty$. The limiting estimate is identical with $F_J$
replaced by the finite measure $\nu_J$.
\end{proof}

Choose $R\uparrow\infty$ through continuity points of $F_J$. For each such
fixed $R$ and every $\varepsilon>0$, Lemma
\ref{lem:bounded-marks-diagonal-cutoff} gives convergence. The approximation
errors are then controlled by Lemmas \ref{lem:diagonal-cutoff-error} and
\ref{lem:mark-truncation-error}: the diagonal cutoff error vanishes as
$\varepsilon\downarrow0$, and the mark truncation error vanishes as
$R\to\infty$, uniformly in $T$. This proves Proposition
\ref{prop:raw-feedback-limit}.

\subsection{Tightness in \texorpdfstring{$L^2$}{L2} via translation estimates}
\label{app:tightness}

Throughout this subsection, kernels on $(0,1]$ are extended by zero to
$\mathbb R$, and each stochastic or deterministic convolution term is
extended to $\mathbb R$ by using the zero-extended kernel in the same
convolution. In particular,
\[
  \widetilde{\mathcal M}_t^T
  :=\sqrt{c_{B,T}}\int_0^1K_{T,\mathrm{ext}}^{\rm res}(t-s)dB_s^T,
  \qquad t\in\mathbb R,
\]
and $\widetilde A^T$ and
$\widetilde{\bar{\mathcal J}}^T$ are defined analogously below in \eqref{tilde:A:J}. Since the
zero-extended kernels are supported in $(0,1]$ and the integration variable
ranges over $[0,1]$, all these extensions are supported in $[0,2]$.
Indeed, $s\in[0,1]$ and $t-s\in(0,1]$ imply $t\in[0,2]$. This auxiliary
convolution agrees with the original terms on $[0,1]$, while its formulation
on $L^2(\mathbb R)$ permits direct use of translation estimates and the
Kolmogorov--M. Riesz--Fr\'echet criterion.

\begin{lemma}[Uniform translation modulus for the rescaled resolvent kernels]
\label{lem:uniform-kernel-translation-cutoff}
Let
\[
  K_T^{\rm res}(u):=T^{1-\alpha}H_T(Tu),
  \qquad u\in(0,1].
\]
Assume $K_T^{\rm res}(u)\to K(u)$ locally uniformly on $(0,1]$, and
$|K_T^{\rm res}(u)|\le Cu^{\alpha-1}$ for $u\in(0,1]$.
Extend all kernels by zero outside $(0,1]$. Then, for every $q\in[1,2]$,
\[
  \lim_{h\to0}
  \limsup_{T\to\infty}
  \left\|
  \tau_hK_T^{\rm res}-K_T^{\rm res}
  \right\|_{L^q([-1,1])}
  =0.
\]
\end{lemma}

\begin{proof}
Fix $\varepsilon\in(0,1)$. The near-diagonal part satisfies
\[
  \sup_{T\ge T_0}
  \left\|K_T^{\rm res}\mathbf 1_{\{0<u\le\varepsilon\}}\right\|_{L^q([-1,1])}^q
  \le
  C\int_0^\varepsilon u^{q(\alpha-1)}du
  =
  \frac{C}{1-q(1-\alpha)}\varepsilon^{1-q(1-\alpha)},
\]
which tends to zero as $\varepsilon\downarrow0$ since $\alpha>1/2$ implies
$1-q(1-\alpha)>0$ for $q\le2$. The same estimate controls the translated
near-diagonal part.

For the cutoff part
$K_T^{{\rm res},\varepsilon}(u):=K_T^{\rm res}(u)\mathbf 1_{\{u>\varepsilon\}}$,
choose $|h|<\varepsilon/2$. Then, except on boundary intervals of total
length $\mathcal O(|h|)$, all arguments of $K_T^{\rm res}$ involved in the
translation difference stay in $[\varepsilon/2,1]$. On
$[\varepsilon/2,1]$, the convergence $K_T^{\rm res}\to K$ is uniform.
Hence, for all large $T$,
\[
\begin{aligned}
\left\|\tau_hK_T^{{\rm res},\varepsilon}-K_T^{{\rm res},\varepsilon}\right\|_{L^q([-1,1])}
&\le
2\|K_T^{\rm res}-K\|_{L^\infty([\varepsilon/2,1])}
+\|\tau_hK^\varepsilon-K^\varepsilon\|_{L^q([-1,1])}
+C_\varepsilon|h|^{1/q}.
\end{aligned}
\]
The first term is arbitrarily small for large $T$. Since
$K\in L^q_{\rm loc}((0,1])$, translations are continuous in $L^q$ on the
cutoff portion, so the second term tends to zero as $h\to0$. Combining and
then letting $\varepsilon\downarrow0$ proves the claim.
\end{proof}

We now complete the tightness proof of Proposition~\ref{prop:tightness} with
the kernel-extension convention stated above. Besides
$K_{T,\mathrm{ext}}^{\rm res}$, let
$K_{T,\mathrm{ext}}^{J,{\rm mean}}$ and $K_{J,\mathrm{ext}}^T$ denote the
zero extensions to $\mathbb R$ of the corresponding kernels on $(0,1]$, and
define, for $t\in\mathbb R$,
\begin{equation}\label{tilde:A:J}
\begin{aligned}
  \widetilde A_t^T
  &:={\int_0^1}K_{T,\mathrm{ext}}^{J,{\rm mean}}(t-s)V_s^Tds,\\
  \widetilde{\bar{\mathcal J}}_t^T
  &:={\int_{(0,1)}}K_{J,\mathrm{ext}}^T(t-s)
  \int_{(0,\infty)}z\widetilde M_J^T(ds,dz).
\end{aligned}
\end{equation}
Together with $\widetilde{\mathcal M}^T$ defined at the start of the
subsection, these extensions agree with their original convolution terms on
$[0,1]$ and are supported in $[0,2]$.

With these definitions, the It\^o-isometry, random-measure-isometry, and
Young's inequality hold on $L^2(\mathbb R)$ with the kernel translation moduli
measured in $L^2(\mathbb R)$ for the martingale terms and in
$L^1(\mathbb R)$ for the predictable term. For example, Fubini's theorem and
$d\langle B^T\rangle_s=(1-p_T)V_s^Tds$ give the exact computation
\[
\E\int_{\mathbb R}
\left|\tau_h\widetilde{\mathcal M}_t^T
-\widetilde{\mathcal M}_t^T\right|^2dt\\
=
c_{B,T}(1-p_T)
\left\|\tau_hK_{T,\mathrm{ext}}^{\rm res}
-K_{T,\mathrm{ext}}^{\rm res}\right\|_{L^2(\mathbb R)}^2
\E\int_0^1V_s^Tds.
\]
Likewise, the compensator identity
$\bar\nu_J^T(ds,dz)=2\kappa_TV_s^TF_J(dz)ds$ and $m_2<\infty$ yield
\[
\E\left\|\tau_h\widetilde{\bar{\mathcal J}}^T
-\widetilde{\bar{\mathcal J}}^T\right\|_{L^2(\mathbb R)}^2
=
2\kappa_Tm_2
\left\|\tau_hK_{J,\mathrm{ext}}^T
-K_{J,\mathrm{ext}}^T\right\|_{L^2(\mathbb R)}^2
\E\int_0^1V_s^Tds,
\]
while Young's inequality yields that
\[
\left\|\tau_h\widetilde A^T-\widetilde A^T\right\|_{L^2(\mathbb R)}
\le
\left\|\tau_hK_{T,\mathrm{ext}}^{J,{\rm mean}}
-K_{T,\mathrm{ext}}^{J,{\rm mean}}\right\|_{L^1(\mathbb R)}
\left\|V^T\right\|_{L^2([0,1])}.
\]
By symmetry of the translation norm it is enough to consider
$0\le h\le1$. With the convention $\tau_hf(t)=f(t+h)$, the supports of a
zero-extended kernel and its positive translate then lie in $[-1,1]$.
Consequently, for $|h|\le1$ the $L^q(\mathbb R)$ translation modulus reduces,
by symmetry, to the $L^q([-1,1])$ modulus in
Lemma~\ref{lem:uniform-kernel-translation-cutoff}. The same observation
applies to $K_J^T$, because
$K_J^T=(\mathfrak c_Tc_Jr_TT^{\alpha-1})K_T^{\rm res}$ with a prefactor
converging to $\xi$, and to $K_T^{J,{\rm mean}}$ in $L^1$. Proposition
\ref{prop:first-moment} and Proposition~\ref{prop:second-moment} therefore
show that all three convolution extensions are uniformly
translation-continuous in mean square.

For the baseline, let $g_{0,\mathrm{ext}}^T$ and $g_{0,\mathrm{ext}}$ be the
zero extensions to $\mathbb R$ of the functions restricted to $[0,1]$. Then
\[
\left\|\tau_hg_{0,\mathrm{ext}}^T-g_{0,\mathrm{ext}}^T\right\|_{L^2(\mathbb R)}
\le
2\left\|g_0^T-g_0\right\|_{L^2([0,1])}
+\left\|\tau_hg_{0,\mathrm{ext}}-g_{0,\mathrm{ext}}\right\|_{L^2(\mathbb R)}.
\]
The first term vanishes as $T\to\infty$, and translation continuity in
$L^2(\mathbb R)$ of the single fixed function $g_{0,\mathrm{ext}}$ makes the
second vanish as $h\to0$.

Finally set
\[
  \widetilde V^T
  :=g_{0,\mathrm{ext}}^T+\widetilde A^T
  +\widetilde{\mathcal M}^T+\widetilde{\bar{\mathcal J}}^T.
\]
The family $(\widetilde V^T)_{T\ge T_0}$ is supported in the common compact
interval $[0,2]$ and is bounded in mean square in $L^2(\mathbb R)$ by the
same It\^{o}'s isometry and Young's inequality estimates together with
Proposition~\ref{prop:second-moment}. The preceding estimates give
\[
  \lim_{h\to0}\limsup_{T\to\infty}
  \E\left\|\tau_h\widetilde V^T-\widetilde V^T
  \right\|_{L^2(\mathbb R)}^2=0.
\]
We spell out the passage from this mean-square translation estimate to random
compact containment. Let $\varphi_\delta$ be a smooth, nonnegative,
compactly supported approximate identity. Jensen's inequality and Fubini's
theorem give
\[
\E\left\|\widetilde V^T-\varphi_\delta*\widetilde V^T\right\|_2^2
\le
\int_{\mathbb R}\varphi_\delta(h)
\E\left\|\tau_h\widetilde V^T-\widetilde V^T\right\|_2^2dh.
\]
The uniform mean-square bound supplies a dominating constant. The reverse
Fatou inequality and the preceding translation estimate therefore imply
\[
\lim_{\delta\downarrow0}\limsup_{T\to\infty}
\E\left\|\widetilde V^T-\varphi_\delta*\widetilde V^T\right\|_2^2=0.
\]
Set $S_\delta f:=\varphi_\delta*f$ and
\[
  H_0:=\{f\in L^2(\mathbb R):f=0\ \text{a.e.\ on }\mathbb R\setminus[0,2]\}.
\]
For fixed $\delta$, $S_\delta:H_0\to L^2(\mathbb R)$ is Hilbert--Schmidt,
because its integral kernel
$\varphi_\delta(t-s)\mathbf 1_{[0,2]}(s)$ belongs to
$L^2(\mathbb R^2)$. It is therefore compact. If
\[
  C_0:=\sup_{T\ge T_0}\E\|\widetilde V^T\|_2^2<\infty,
  \qquad
  K_{\delta,R}:=
  \overline{S_\delta\{f\in H_0:\|f\|_2\le R\}},
\]
then $K_{\delta,R}$ is compact and Markov's inequality gives
\begin{equation}
\label{eq:smoothed-compact-containment}
  \sup_{T\ge T_0}
  \Q\left(S_\delta\widetilde V^T\notin K_{\delta,R}\right)
  \le \frac{C_0}{R^2}.
\end{equation}

We now turn the approximation estimate into compact containment. Fix an
arbitrary sequence $T_j\to\infty$, let $\eta>0$, and set
$\varepsilon_n:=2^{-n}$. By the preceding $L^2$ approximation and Markov's
inequality, one can choose $\delta_n\downarrow0$ and $J_n$ such that
\begin{equation}
\label{eq:mollification-probability-error}
  \sup_{j\ge J_n}
  \Q\left(
    \left\|\widetilde V^{T_j}-S_{\delta_n}\widetilde V^{T_j}\right\|_2
    >\frac{\varepsilon_n}{2}
  \right)
  \le \eta 2^{-n-2}.
\end{equation}
For each $n$, \eqref{eq:smoothed-compact-containment} and compactness of
$K_{\delta_n,R}$ provide a finite $\varepsilon_n/2$-net for the smoothed
variables, outside an event of probability at most $\eta2^{-n-2}$ after
choosing $R$ sufficiently large. The finitely many laws with $j<J_n$ are
tight individually. Enlarging the net to cover those laws, we obtain a finite
set $F_n\subset L^2(\mathbb R)$ such that
\begin{equation}
\label{eq:finite-net-probability}
  \sup_{j\ge1}
  \Q\left(
    d_2\left(\widetilde V^{T_j},F_n\right)>\varepsilon_n
  \right)
  \le \eta2^{-n-1},
  \qquad
  d_2(f,F):=\inf_{g\in F}\|f-g\|_2.
\end{equation}
Define
\[
  K_\eta
  :=\bigcap_{n\ge1}
  \{f\in L^2(\mathbb R):d_2(f,F_n)\le\varepsilon_n\}.
\]
The set $K_\eta$ is closed and totally bounded: for any $r>0$, choose $n$
so that $2\varepsilon_n<r$; then the finite set $F_n$ is an $r$-net for
$K_\eta$. Since $L^2(\mathbb R)$ is complete, $K_\eta$ is compact. Moreover,
the union bound and \eqref{eq:finite-net-probability} yield
\[
  \sup_{j\ge1}\Q\left(\widetilde V^{T_j}\notin K_\eta\right)
  \le\sum_{n\ge1}\eta2^{-n-1}<\eta.
\]
Thus every sequence $T_j\to\infty$ has tight laws. Equivalently, by
Prokhorov's theorem and the Kolmogorov--M. Riesz--Fr\'echet criterion
\cite[Theorem~4.26]{brezis2011functional},
$(\widetilde V^T)$ is tight in $L^2(\mathbb R)$. Since
$\widetilde V^T=V^T$ on $[0,1]$, restriction to $[0,1]$ proves tightness of
$(V^T)_{T\ge T_0}$ in $L^2([0,1])$.

\subsection{Resolvent perturbation details}
\label{app:resolvent-perturbation}
We provide the details for Lemma \ref{lem:resolvent-perturbation}. All kernels
discussed here possess densities, and the inequalities are pointwise density
inequalities. The renewal-density envelope strengthens the available
Laplace-transform convergence to the pointwise perturbation estimate used
below.

We use the resolvent identity. For $a_T$ and $\widetilde a_T:=a_T(1-p_T)$, the resolvent identity yields that
\[
\Psi_{a_T}-\Psi_{\widetilde a_T}
=(a_T-\widetilde a_T)\phi*(\delta_0+\Psi_{a_T})*(\delta_0+\Psi_{\widetilde a_T}).
\]
Expanding the convolutions yields
\[
\phi*(\delta_0+\Psi_{a_T})*(\delta_0+\Psi_{\widetilde a_T})
=
\phi+\phi*\Psi_{a_T}+\phi*\Psi_{\widetilde a_T}
+\phi*\Psi_{a_T}*\Psi_{\widetilde a_T}.
\]
Since
\[
\Psi_{a_T}=a_T\phi+a_T\phi*\Psi_{a_T},
\]
and all kernels are non-negative, for all sufficiently large $T$, $a_T\ge1/2$, and hence
\[
\phi\le a_T^{-1}\Psi_{a_T}\le 2\Psi_{a_T},
\qquad
\phi*\Psi_{a_T}\le a_T^{-1}\Psi_{a_T}\le 2\Psi_{a_T}.
\]
Similarly, since $\widetilde a_T\to1$,
\[
\phi*\Psi_{\widetilde a_T}
\le \widetilde a_T^{-1}\Psi_{\widetilde a_T}
\le C\Psi_{\widetilde a_T}
\le C\Psi_{a_T}.
\]
Since $\widetilde a_T\le a_T$, the non-negative series representation yields that
\[
  \Psi_{\widetilde a_T}\le \Psi_{a_T}.
\]
Therefore,
\[
  \phi*\Psi_{a_T}*\Psi_{\widetilde a_T}
  \le C\Psi_{a_T}*\Psi_{a_T}.
\]
Using the envelope from Proposition \ref{prop:resolvent-scaling}, we bound $\Psi_{a_T}*\Psi_{a_T}$:
\[
\begin{aligned}
(\Psi_{a_T}*\Psi_{a_T})(Tu)
&=
T\int_0^u\Psi_{a_T}(T(u-v))\Psi_{a_T}(Tv)dv\\
&\le
CT^{2\alpha-1}\int_0^u(u-v)^{\alpha-1}v^{\alpha-1}dv
=
CT^{2\alpha-1}B(\alpha,\alpha)u^{2\alpha-1}.
\end{aligned}
\]
Here, $B(\alpha,\alpha)=(\Gamma(\alpha))^2/\Gamma(2\alpha)$ is the Beta
function. Since $u\in(0,1]$ and $\alpha\in(1/2,1)$, we have
$2\alpha-1>0>\alpha-1$. Hence,
$u^{2\alpha-1}=u^{\alpha-1}u^\alpha\le u^{\alpha-1}$. Therefore,
\[
  (\Psi_{a_T}*\Psi_{a_T})(Tu)
  \le
  CT^{2\alpha-1}u^{\alpha-1},
  \qquad u\in(0,1].
\]
Multiplying by $a_T-\widetilde a_T=a_Tp_T\le p_T$ and rescaling yields that
\[
T^{1-\alpha}\left|\Psi_T^{\rm reg}(Tu)-\Psi_T(Tu)\right|
\le
C p_TT^\alpha u^{\alpha-1}.
\]
Set $\epsilon_T:=p_TT^\alpha$. Since $p_T=\mathcal O(T^{-2\alpha})$, we have $\epsilon_T=\mathcal O(T^{-\alpha})$, and the desired estimate follows after increasing the constant $C$.
Combining this estimate with
Proposition~\ref{prop:resolvent-scaling} and
$H_T=(1-p_T)^{-1}\Psi_T^{\rm reg}$ proves both the local uniform convergence
$T^{1-\alpha}H_T(Tu)\to K(u)$ and the uniform envelope
$T^{1-\alpha}H_T(Tu)\le Cu^{\alpha-1}$ used throughout the paper.

\bibliographystyle{abbrvnat}
\bibliography{bibtex}

@article{xu2024diffusion,
  title={Diffusion approximations for self-excited systems with applications to general branching processes},
  author={Xu, Wei},
  journal={Annals of Applied Probability},
  volume={34},
  number={3},
  pages={2650--2713},
  year={2024},
  publisher={Institute of Mathematical Statistics}
}

@Article{dandapani2021quadratic,
  author  = {Dandapani, A. and Jusselin, P. and Rosenbaum, M.},
  title   = {From Quadratic {H}awkes Processes to Super-{H}eston Rough Volatility Models with {Z}umbach Effect},
  journal = {Quantitative Finance},
  volume  = {21},
  number  = {8},
  pages   = {1235--1247},
  year    = {2021}
}

@Article{gatheral2020quadratic,
  author  = {Gatheral, J. and Jusselin, P. and Rosenbaum, M.},
  title   = {The Quadratic Rough {H}eston Model and the Joint {S\&P} 500/{VIX} Smile Calibration Problem},
  journal = {Risk},
  year    = {2020},
  month   = {May}
}

@article{contkokholm2013,
  author  = {Cont, Rama and Kokholm, Thomas},
  title   = {A Consistent Pricing Model for Index Options and Volatility Derivatives},
  journal = {Mathematical Finance},
  year    = {2013},
  volume  = {23},
  number  = {2},
  pages   = {248--274}
}

@article{sepp2008,
  author  = {Sepp, Artur},
  title   = {Pricing options on realized variance in the {H}eston model with jumps in returns and volatility},
  journal = {Journal of Computational Finance},
  year    = {2008},
  volume  = {11},
  pages   = {33--70}
}

@article{rebolledo1980central,
  title={Central limit theorems for local martingales},
  author={Rebolledo, Rolando},
  journal={Zeitschrift für Wahrscheinlichkeitstheorie und Verwandte Gebiete},
  volume={51},
  number={3},
  pages={269--286},
  year={1980}
}

@article{hagerHorstWagenhoferXu2026roughlognormal,
  title         = {Microstructural Foundation of Rough Log-Normal Volatility Models},
  author        = {Hager, Paul P. and Horst, Ulrich and Wagenhofer, Thomas and Xu, Wei},
  year          = {2026},
  journal       = {arXiv preprint arXiv:2603.13170}
}

@book{brezis2011functional,
  title={Functional Analysis, Sobolev Spaces and Partial Differential Equations},
  author={Brezis, H.},
  series={Universitext},
  year={2011},
  publisher={Springer},
  address={New York}
}

@book{folland1999real,
  title={Real Analysis: Modern Techniques and Their Applications},
  author={Folland, Gerald B.},
  edition={2nd},
  year={1999},
  publisher={John Wiley \& Sons},
  address={New York}
}

@book{billingsley1995probability,
  title={Probability and Measure},
  author={Billingsley, P.},
  edition={3},
  year={1995},
  publisher={John Wiley \& Sons},
  address={New York}
}

@book{billingsley1999convergence,
  title={Convergence of Probability Measures},
  author={Billingsley, Patrick},
  edition={2},
  year={1999},
  publisher={John Wiley \& Sons},
  address={New York},
  doi={10.1002/9780470316962}
}

@article{bondi2024affine,
  title={Affine Volterra processes with jumps},
  author={Bondi, A. and Livieri, G. and Pulido, S.},
  journal={Stochastic Processes and their Applications},
  volume={168},
  pages={104264},
  year={2024},
  publisher={Elsevier}
}

@misc{alonso2026microstructure,
  title        = {Microstructure Mechanisms for Stochastic Volatility and Roughness},
  author       = {Noguer I Alonso, Miquel},
  year         = {2026},
  note         = {Available at SSRN 6122526},
  doi          = {10.2139/ssrn.6122526}
}

@article{wang2025rough,
  title={{R}ough {H}eston model as the scaling limit of bivariate cumulative heavy-tailed {INAR} processes: {W}eak-error bounds and option pricing},
  author={Wang, Yingli and Cui, Zhenyu and Zhu, Lingjiong},
  journal={arXiv preprint arXiv:2503.18259},
  year={2025}
}

@Article{guyon2025dispersion,
  author = {Guyon, Julien},
  title = {Dispersion-Constrained Martingale {Schr{\"o}dinger} Bridges: Joint Entropic Calibration of Stochastic Volatility Models to {S{\&}P} 500 and {VIX} Smiles},
  journal = {SIAM Journal on Financial Mathematics},
  volume = {16},
  number = {3},
  pages = {834--874},
  year = {2025}
}

@Article{BayerBreneis2024WeakRoughHeston,
  author = {Bayer, C. and Breneis, S.},
  title = {Efficient option pricing in the rough {H}eston model using weak simulation schemes},
  journal = {Quantitative Finance},
  year = {2024},
  volume = {24},
  number = {9},
  pages = {1247--1261}
}

@Article{BayerFrizGatheral2015PricingRough,
  title = {Pricing under rough volatility},
  author = {Bayer, C. and Friz, P. and Gatheral, J.},
  journal = {Quantitative Finance},
  volume = {16},
  number = {6},
  pages = {887--904},
  year = {2016}
}

@Article{jaisson2015limit,
  title = {{Limit Theorems for nearly unstable Hawkes processes}},
  author = {Jaisson, T. and Rosenbaum, M.},
  journal = {Annals of Applied Probability},
  volume = {25},
  number = {2},
  pages = {600--631},
  year = {2015}
}

@Article{callegaro2021fast,
  title = {{Fast hybrid schemes for fractional Riccati equations (rough is not so tough)}},
  author = {Callegaro, G. and Grasselli, M. and Pages, G.},
  journal = {Mathematics of Operations Research},
  volume = {46},
  number = {1},
  pages = {221--254},
  year = {2021},
  publisher = {INFORMS}
}

@Article{hawkes1974cluster,
  title = {A cluster process representation of a self-exciting process},
  author = {Hawkes, Alan G and Oakes, David},
  journal = {Journal of Applied Probability},
  volume = {11},
  number = {3},
  pages = {493--503},
  year = {1974},
  publisher = {Cambridge University Press}
}

@Article{bacry2013modelling,
  title = {Modelling microstructure noise with mutually exciting point processes},
  author = {Bacry, E. and Delattre, S. and Hoffmann, M. and Muzy, J.-F.},
  journal = {Quantitative Finance},
  volume = {13},
  number = {1},
  pages = {65--77},
  year = {2013},
  publisher = {Taylor \& Francis}
}

@Article{jaisson2016rough,
  title = {{Rough fractional diffusions as scaling limits of nearly unstable heavy tailed Hawkes processes}},
  author = {Jaisson, T. and Rosenbaum, M.},
  journal = {Annals of Applied Probability},
  volume = {26},
  number = {5},
  pages = {2860--2882},
  year = {2016}
}

@Book{bingham1989regular,
  title = {Regular Variation},
  author = {Bingham, N. H. and Goldie, C. M. and Teugels, J. L.},
  year = {1989},
  publisher = {Cambridge University Press}
}

@Book{jacod2013limit,
  title = {Limit Theorems for Stochastic Processes},
  author = {Jacod, J. and Shiryaev, A.},
  volume = {288},
  year = {2013},
  publisher = {Springer Science \& Business Media}
}

@Article{bowsher2007modelling,
  title = {Modelling security market events in continuous time: Intensity based, multivariate point process models},
  author = {Bowsher, C.},
  journal = {Journal of Econometrics},
  volume = {141},
  number = {2},
  pages = {876--912},
  year = {2007},
  publisher = {Elsevier}
}

@Article{horst2023convergence,
  title = {Convergence of heavy-tailed {H}awkes Processes and the microstructure of rough volatility},
  author = {Horst, U. and Xu, W. and Zhang, R.},
  journal = {arXiv preprint arXiv:2312.08784},
  year = {2023}
}

@Article{horst2019scaling,
  title = {A scaling limit for limit order books driven by {H}awkes processes},
  author = {Horst, U. and Xu, W.},
  journal = {SIAM Journal on Financial Mathematics},
  volume = {10},
  number = {2},
  pages = {350--393},
  year = {2019},
  publisher = {SIAM}
}

@Article{horst2022microstructure,
  title = {The microstructure of stochastic volatility models with self-exciting jump dynamics},
  author = {Horst, U. and Xu, W.},
  journal = {Annals of Applied Probability},
  volume = {32},
  number = {6},
  pages = {4568--4610},
  year = {2022},
  publisher = {Institute of Mathematical Statistics}
}

@Article{horstXu2026functional,
  title={Functional limit theorems for {H}awkes processes},
  author={Horst, Ulrich and Xu, Wei},
  journal={Probability Theory and Related Fields},
  pages={917--996},
  volume={194},
  number={1--2},
  year={2026},
  publisher={Springer}
}

@Article{horstXuZhang2024pathdependent,
  author = {Horst, U. and Xu, W. and Zhang, R.},
  title = {Path-dependent fractional {V}olterra Equations and the microstructure of rough volatility models driven by {P}oisson random measures},
  journal = {arXiv preprint arXiv:2412.16436},
  year = {2024}
}

@Article{el2019characteristic,
  title = {The characteristic function of rough {H}eston models},
  author = {El Euch, O. and Rosenbaum, M.},
  journal = {Mathematical Finance},
  volume = {29},
  number = {1},
  pages = {3--38},
  year = {2019},
  publisher = {Wiley Online Library}
}

@article{wang2026scaling,
  title = {Scaling limit of heavy-tailed nearly unstable cumulative {INAR}($\infty$) processes and rough fractional diffusions},
  author = {Cai, Chunhao and He, Ping and Wang, Qinghua and Wang, Yingli},
  journal = {Methodology and Computing in Applied Probability},
  year = {2026},
  volume = {28},
  number = {1},
  pages = {13}
}

@Article{bondi2024,
  author = {Bondi, A. and Pulido, S. and Scotti, S.},
  title = {The rough {H}awkes--{H}eston stochastic volatility model},
  journal = {Mathematical Finance},
  volume = {34},
  number = {4},
  pages = {1197--1241},
  year = {2024}
}

@Article{eleuchfukasawarosenbaum2018,
  author = {El Euch, O. and Fukasawa, M. and Rosenbaum, M.},
  title = {The microstructural foundations of leverage effect and rough volatility},
  journal = {Finance and Stochastics},
  volume = {22},
  number = {2},
  pages = {241--280},
  year = {2018}
}

@Article{gatheral2018,
  author = {Gatheral, J. and Jaisson, T. and Rosenbaum, M.},
  title = {Volatility is rough},
  journal = {Quantitative Finance},
  volume = {18},
  number = {6},
  pages = {933--949},
  year = {2018}
}

@Article{heston1993,
  author = {Heston, S.},
  title = {A closed-form solution for options with stochastic volatility with applications to bond and currency options},
  journal = {Review of Financial Studies},
  volume = {6},
  number = {2},
  pages = {327--343},
  year = {1993}
}

@Article{abijaber2019,
  author = {Abi Jaber, E. and Larsson, M. and Pulido, S.},
  title = {Affine {V}olterra processes},
  journal = {Annals of Applied Probability},
  volume = {29},
  number = {5},
  pages = {3155--3200},
  year = {2019}
}

@Article{abijaber2021,
  author = {Abi Jaber, E.},
  title = {Weak existence and uniqueness for affine stochastic {V}olterra equations with {$L^1$}-kernels},
  journal = {Bernoulli},
  volume = {27},
  number = {3},
  pages = {1583--1615},
  year = {2021}
}

@Book{bremaud1981,
  author = {Br{\'e}maud, P.},
  title = {Point Processes and Queues: Martingale Dynamics},
  publisher = {Springer},
  year = {1981}
}

@article{caravennadoney2019,
  author  = {Caravenna, Francesco and Doney, Ron},
  title   = {Local large deviations and the strong renewal theorem},
  journal = {Electronic Journal of Probability},
  volume  = {24},
  year    = {2019},
  number  = {72},
  pages   = {1--48}
}

\end{document}